\documentclass[sigconf,a4paper,review=false]{acmar}

\usepackage[UKenglish]{babel}
\usepackage{wrapfig}
\usepackage{graphicx}

\pdfoutput=1

\usepackage{microtype}
\usepackage{amssymb}
\usepackage{enumerate}
\usepackage{tikz}
\usepackage{xcolor,latexsym,amsmath,extarrows,alltt}
\usepackage{xspace}
\usepackage{booktabs}
\usepackage{mathtools}
\usepackage{enumitem}
\usepackage{stmaryrd}
\usepackage{microtype}
\usepackage{bussproofs}
\usepackage{multirow}

\usepackage[printwatermark]{xwatermark}

\newcommand{\N}{\mathbb{N}}
\newcommand{\F}{\mathcal{F}}
\newcommand{\V}{\mathcal{V}}
\newcommand{\M}{\mathcal{M}}
\renewcommand{\P}{\mathcal{P}}
\newcommand{\Sorts}{\mathcal{S}}
\newcommand{\Terms}{\mathcal{T}}
\newcommand{\Rules}{\mathcal{R}}
\newcommand{\AlterRules}{S}
\newcommand{\RulesEta}{\Rules^{\mathtt{ext}}} 

\newcommand{\Defineds}{\mathcal{D}}
\newcommand{\FV}{\mathit{FV}}
\newcommand{\FMV}{\mathit{FMV}}
\newcommand{\domain}{\mathtt{dom}}
\newcommand{\DDP}{\mathit{DDP}}
\newcommand{\SDP}{\mathit{SDP}}

\newcommand{\FR}{\mathit{FR}}
\newcommand{\FRA}{\mathit{FA}}

\newcommand{\UR}{\mathit{UR}}
\newcommand{\asort}{\iota}
\newcommand{\bsort}{\kappa}
\newcommand{\atype}{\sigma}
\newcommand{\btype}{\tau}
\newcommand{\ctype}{\pi}

\newcommand{\identifier}[1]{\mathtt{#1}}
\newcommand{\afun}{\identifier{f}}
\newcommand{\bfun}{\identifier{g}}
\newcommand{\cfun}{\identifier{h}}
\newcommand{\adpprob}{M}
\newcommand{\bdpprob}{Q}
\newcommand{\avarormeta}{b}

\newcommand{\no}{\texttt{NO}}

\newcommand{\wanda}{\textsf{WANDA}\xspace}

\newcommand{\csiho}{$\textsf{CSI}^{\textsf{ho}}$\xspace}
\renewcommand{\csiho}{$\textsf{CSI\textasciicircum{}ho}$\xspace}
\newcommand{\acph}{\textsf{ACPH}\xspace}

\newcommand{\app}[2]{#1\ #2}
\newcommand{\apps}[3]{#1\ #2 \cdots #3}
\newcommand{\abs}[2]{\lambda #1.#2}
\newcommand{\meta}[2]{#1[#2]}

\newcommand{\arity}{\mathit{ar}}
\newcommand{\head}{\mathsf{head}}
\newcommand{\symb}[1]{\mathtt{#1}}

\newcommand{\arrtype}{\rightarrow}
\newcommand{\arrdp}{\Rrightarrow}
\newcommand{\arrz}{\Rightarrow}
\newcommand{\arr}[1]{\arrz_{#1}}
\newcommand{\arrr}[1]{\arr{#1}^*}

\newcommand{\supterm}{\rhd}
\newcommand{\suptermeq}{\unrhd}
\newcommand{\bsuptermeq}[1]{\unrhd_{#1}}
\newcommand{\safesup}{\unrhd_{\mathtt{safe}}^{\text{\tiny\cite{suz:kus:bla:11}}}}

\newcommand{\cand}{\mathsf{cand}}
\newcommand{\Proc}{\mathit{Proc}}
\newcommand{\static}{\mathtt{computable}}
\newcommand{\minimal}{\mathtt{minimal}}
\newcommand{\arbitrary}{\mathtt{arbitrary}}
\newcommand{\formative}{\mathtt{formative}}
\newcommand{\nonformative}{\mathtt{all}}
\newcommand{\metafy}{\mathit{metafy}}
\newcommand{\greqsort}{\succeq^{\ext{\Sorts}}}
\newcommand{\leqsort}{\preceq^{\ext{\Sorts}}}
\newcommand{\eqsort}{\approx^{\ext{\Sorts}}}
\newcommand{\grsort}{\succ^{\ext{\Sorts}}}
\newcommand{\gracsortup}{\succeq^{\Sorts}_+}
\newcommand{\gracsortdown}{\succ^{\Sorts}_-}
\newcommand{\gracc}{\unrhd_{\mathtt{acc}}}
\newcommand{\accreduce}[1]{\Rrightarrow_{#1}}
\newcommand{\Acc}{\mathit{Acc}}
\newcommand{\approxp}{\approx_\mac}
\newcommand{\project}{\overline{\nu}}
\newcommand{\ttag}{\mathtt{tag}}
\newcommand{\funs}{\mathtt{funs}}

\newcommand{\unsharp}[1]{#1^\flat}
\newcommand{\etalong}[1]{#1\!\!\uparrow^\eta}
\newcommand{\halfetalong}[1]{\overline{#1}}

\newcommand{\pgt}{\succ} 
\newcommand{\pge}{\succcurlyeq} 
\newcommand{\rge}{\succsim} 

\newcommand{\nul}{\symb{0}}
\newcommand{\one}{\symb{1}}
\newcommand{\nil}{\symb{[]}}
\renewcommand{\nil}{\symb{nil}}
\newcommand{\cons}{\symb{cons}}

\newcommand{\suc}{\symb{s}}
\newcommand{\map}{\symb{map}}

\newcommand{\deriv}{\symb{deriv}}
\newcommand{\Ce}{\mathcal{C}_{\epsilon}}

\newcommand{\nat}{\symb{nat}}
\newcommand{\lijst}{\symb{list}}
\newcommand{\real}{\symb{real}}

\newcommand{\mac}{{\textcolor{red}{e}}}
\renewcommand{\mac}{e}
\newcommand{\mia}{{\textcolor{red}{k}}}
\renewcommand{\mia}{k}
\newcommand{\maa}{{\textcolor{red}{m}}}
\renewcommand{\maa}{m}

\newcommand{\refDef}[1]{Def.~\ref{#1}}
\newcommand{\refEx}[1]{Ex.~\ref{#1}}
\newcommand{\refLemma}[1]{Lemma~\ref{#1}}
\newcommand{\refThm}[1]{Thm.~\ref{#1}}
\newcommand{\refApp}[1]{Appendix~\ref{#1}}
\newcommand{\refSec}[1]{Sec.~\ref{#1}}

\definecolor{cfcolor}{HTML}{FF5000}

\newcommand{\CKold}[1]{}
\newcommand{\CFold}[1]{}

\newcommand{\ext}[1]{\textcolor{red}{#1}}
\renewcommand{\ext}[1]{#1}

\begin{document}


\newtheorem{remark}[theorem]{Remark}


\title{The unified higher-order dependency pair framework}


\acmConference{}{}{}




\author{Carsten Fuhs}
\affiliation{Dept.\ of Computer Science and Information Systems,\\
Birkbeck, University of London, United Kingdom}
\email{carsten@dcs.bbk.ac.uk}
\author{Cynthia Kop}
\affiliation{Dept.\ of Software Science,\\
Radboud University Nijmegen, The Netherlands}
\email{c.kop@cs.ru.nl}








\begin{abstract}
In recent years, two higher-order extensions of the powerful dependency
pair approach for termination analysis of first-order term rewriting
have been defined: the \emph{static} and the \emph{dynamic} approach.
Both approaches offer distinct advantages and disadvantages.
However, a grand unifying theory is thus far missing, and both approaches
lack the modularity present in the dependency pair \emph{framework}
commonly used in first-order rewriting.  Moreover,
neither approach can
be used to prove non-termination.


In this paper, we aim to address this gap by defining a higher-order
dependency pair framework, integrating both approaches into a shared formal setup.
%
The framework has been implemented in the (fully automatic)
higher-order termination tool \wanda.
\end{abstract}


\maketitle


\section{Introduction}

Term rewriting \cite{baa:nip:98,ter:03} is an important area of logic, with
applications in many different areas of computer science
\cite{bac:gan:94, 
der:kap:89, 
fuh:kop:nis:17, 
haf:nip:10, 
hoe:arv:99, 
mea:92, 
ott:bro:ess:gie:10}. 
%
\emph{Higher-order} term rewriting -- which extends the traditional
\emph{first-order} term rewriting with
higher-order
types and binders as in
the $\lambda$-calculus -- offers a formal foundation of functional
programming 
and a tool for equational reasoning in
higher-order logic.
%
%
A key question
in the analysis of both first- and
higher-order term rewriting is \emph{termination}, or strong
normalisation -- both for its own sake, and
as part of confluence and equivalence analysis.

In first-order term rewriting, a highly effective method 
to prove
termination (both manually and automatically) is the
\emph{dependency pair (DP) approach} \cite{art:gie:00}.
This approach has been extended to the \emph{DP framework}
\cite{gie:thi:sch:05:2,gie:thi:sch:fal:06}, a highly
modular
methodology which new techniques for proving termination \emph{and
non-termination} can easily be plugged into
in the form of \emph{processors}.

In higher-order term rewriting, two adaptations of the DP approach have
been defined: \emph{dynamic} \cite{sak:wat:sak:01,kop:raa:12} and
\emph{static} \cite{bla:06,sak:kus:05,kus:iso:sak:bla:09,suz:kus:bla:11}.
Each approach has
distinct costs and benefits;
while
dynamic DPs
are more broadly applicable,
analysis of static DPs is often
easier.
This difference can be problematic 
for defining new
techniques based on the DP approach, as they 
must be
proved correct
for both dynamic and static DPs.
%
%
This problem is exacerbated by the existence of multiple styles of
higher-order rewriting, such as \emph{Algebraic Functional Systems
(AFSs)} \cite{jou:rub:99} (
used in the annual Termination
Competition \cite{termcomp})
and \emph{Higher-order Rewrite Systems (HRSs)} \cite{nip:91,mil:91}
(
used in the annual Confluence Competition \cite{coco}),
which have similar but not fully compatible syntax
and semantics.
%
%
What is more, neither approach offers the modularity and extendability
of the DP framework, nor can they be used to prove non-termination.
Both approaches are less general than they could be.  For example,
most 
versions
of the static approach use a restriction
which does not consider strictly positive inductive
types~\cite{bla:jou:oka:02}.
The dynamic approach is sound for all systems, but only complete for
left-linear ones -- that is, a non-left-linear AFS may have an
infinite dependency chain following \cite{kop:raa:11,kop:raa:12} even
if it is terminating; the static approach is incomplete in this sense
for 
even more systems.

In this paper, we define a \emph{higher-order dependency pair
framework}, which combines the dynamic and static styles, is fully
modular, and can be used for both termination and non-termination
without restrictions.  For broad applicability, we use a new
rewriting formalism, \emph{AFSMs}, designed to capture several flavours
of\linebreak
higher-order rewriting, including AFSs and HRSs with a pattern
restriction.
We have dropped the restriction to left-linear systems for
completeness of dynamic DPs
and 
liberalised both the restrictions to use static DPs and to
obtain a complete analysis if we do.
In addition, we introduce a series of new techniques (``processors'') to
provide key
termination techniques within this framework.
This is a \emph{foundational} paper, focused on defining a general
theoretical framework for higher-order termination analysis rather
than implementation concerns.
%
We have, however, implemented most 
results in the fully
automatic termination 
tool \wanda~\cite{wanda}.

\paragraph{Related Work.}
%
%
There is a vast body of work in the first-order setting regarding the
DP approach \cite{art:gie:00} and framework
\cite{gie:thi:sch:05:2,gie:thi:sch:fal:06,hir:mid:07}.
The approach for context-sensitive rewriting
\cite{ala:gut:luc:10}
is somewhat relevant, as it also admits \emph{collapsing} DPs
(in their case, a DP with a variable as its right-hand side)
and therefore requires some similar
adaptations of common techniques.  However, beyond this, the two
different settings are not really comparable.
%

The static DP approach
is discussed in, e.g.,
\cite{kus:iso:sak:bla:09,sak:kus:05,suz:kus:bla:11,kus:13}. This
approach can be used only for \emph{plain function passing (PFP)}
systems.  The definition of PFP is not fixed, as later papers
sometimes weaken earlier restrictions or transpose
them to a different rewriting formalism, but always concerns the
position of higher-order variables in the left-hand sides of rules.
These works include non-pattern HRSs
\cite{kus:iso:sak:bla:09,suz:kus:bla:11} and polymorphic rewriting
\cite{kus:13}, which we do not consider,
but do not employ formative rules or meta-variable conditions, which
we do.
Importantly, these methods do not consider strictly positive inductive
types, which could be used to significantly broaden the PFP
restriction.  Such types \emph{are} considered in an early paper which
defines a variation of static higher-order dependency
pairs~\cite{bla:06}
based on a computability closure~\cite{bla:jou:oka:02,bla:16}.
However, this work carries different restrictions (e.g., DPs
must be type-preserving and not introduce fresh variables) and
provides only one analysis technique (reduction pairs) on these DPs.
Moreover, although the proof method is based on Tait and  Girard's
notion of computability \cite{tai:67}, the approach thus far does not
exploit this beyond the way the initial set of DPs is obtained.
We will present a variation of PFP for the AFSM formalism that is
strictly more permissive than earlier definitions as applied to AFSMs,
and our framework exploits
the inherent computability by introducing a $\static$
flag
that can be used by the
static subterm criterion processor (\refThm{thm:staticsubtermproc}).
In addition, we allow static DPs to also be used for non-termination
and add features such as formative rules.

Unlike the static approach,
the dynamic approach \cite{aot:yam:05,kop:raa:11,kop:raa:12} is not
restricted, but it allows for collapsing DPs
of the form
$\ell \arrz
\apps{x}{s_1}{s_n}$ with $x$ a variable, which can be
difficult to handle.  Thus far, this approach has been incomplete for
non-left-linear systems due to bound variables that become free in a
dependency pair.
Here, we repair that problem by using a rewriting formalism that
separates variables used for matching from those used as binders.

Both static and dynamic approaches actually lie halfway be\-tween the
original ``DP approach'' of first-order rewriting and a full DP
framework as in \cite{gie:thi:sch:05:2,gie:thi:sch:fal:06} and the
present work.
Most of these works
\cite{kop:raa:11,kop:raa:12,kus:13,kus:iso:sak:bla:09,suz:kus:bla:11}
prove ``non-loopingness'' or ``chain-freeness'' of a set $\P$ of DPs
through a number of theorems.
However, there is no concept of \emph{DP problems}, and the set
$\Rules$ of rules cannot be altered.  They also fix assumptions on
dependency chains -- such as minimality \cite{kus:iso:sak:bla:09} or
being ``tagged'' \cite{kop:raa:12} -- which frustrate extendability
and are more naturally dealt with in a DP framework using flags.

The clear precursor of the present work is \cite{kop:raa:12}, which
provides such a halfway framework for dynamic DPs, 
introduces a notion of
formative rules,
and briefly translates 
a
basic form of static DPs to the same setting.
Our formative \emph{reductions} consider the shape of reductions rather
than 
the rules they use,
and they can be used as a flag in the
framework to
gain additional power in other processors.
Our integration of the two styles also goes deeper, allowing for
static and dynamic DPs to be used in the same proof and giving a
complete method 
using
static DPs for a larger group of systems.
%
%

In addition, we have several completely new features, including
meta-variable conditions (an essential ingredient for a complete
method), new flags to DP problems, and various processors including
ones that modify collapsing DPs.

For a more elaborate discussion of the static and dynamic DP
approaches, we refer to \cite{kop:raa:12,kop:12}.
%

\smallskip
The paper is organised as follows: \refSec{sec:prelim} introduces
higher-order rewriting using AFSMs and recapitulates computability.
In \refSec{sec:dp} we state dynamic and static DPs for AFSMs.
\refSec{sec:framework} formulates the DP framework and a number of DP
processors for existing and new termination proving techniques.
\refSec{sec:conclusions} concludes.
A discussion of the translation of existing static DP approaches to
the AFSM formalism, as well as detailed proofs for all results in
this paper, are available in the appendix.  
In addition, many of the results have been informally
published in the second author's PhD thesis \cite{kop:12}.


\section{Preliminaries}
\label{sec:prelim}

In this section, we first define our notation by introducing the AFSM
formalism.  Although not one of the standards of higher-order rewriting,
AFSMs combine features from various 
forms of
higher-order rewriting and can be
seen as a form of IDTSs~\cite{bla:00} which includes application.
Then we present a definition of \emph{computability}, a technique often
used for higher-order termination. 

\subsection{Higher-order term rewriting using AFSMs}

Unlike first-order term rewriting, there is no single, unified approach
to higher-order term rewriting, but rather a number of
similar but not fully compatible systems 
aiming to combine term rewriting and typed $\lambda$-calculi.
For 
generality, we will use
\emph{Algebraic Functional Systems with Meta-variables}: a formalism
which admits translations from the main formats of higher-order
term rewriting.

\begin{definition}[Simple types]
We fix a set $\Sorts$ of \emph{sorts}.
All sorts are simple types, and if $\atype,\btype$ are simple types,
then so is $\atype \arrtype \btype$.
\end{definition}

We let $\arrtype$ be right-associative.
All types have a unique form
$\atype_1 \arrtype \dots \arrtype \atype_\maa \arrtype \asort$
with $\asort \in \Sorts$.

\begin{definition}[Terms]\label{def:terms}
We fix disjoint sets $\F$ of \emph{function symbols} and $\V$ of
\emph{variables}, each symbol equipped with a type.  We assume that both
$\F$ and $\V$ contain infinitely many symbols of all types.  Terms are
expressions $s$ where $s : \atype$ can be derived for some 
$\atype$ by: 

\begin{tabular}{llll}
(\textsf{V}) & $x : \atype$ & if & $x : \atype \in \V$\\
(\textsf{F}) & $\afun : \atype$ & if & $\afun : \atype \in \F$\\
(\textsf{@}) & $s\ t : \btype$ & if & $s : \atype \arrtype \btype$ and
  $t : \atype$ \\
($\mathsf{\Lambda}$) & $\lambda x.s : \atype \arrtype \btype$ & if &
  $x : \atype \in \V$ and $s : \btype$
\end{tabular}

\noindent
The $\lambda$ binds variables as in the $\lambda$-calculus; unbound
variables are called \emph{free}, and $\FV(s)$ is the set of free
variables in $s$.
A term $s$ is
\emph{closed} if $\FV(s) = \emptyset$.
Terms are considered modulo $\alpha$-conversion.
Application (\textsf{@}) is left-associative;
abstractions ($\mathsf{\Lambda}$) extend as far to the right
as possible.
A term $s$ \emph{has type} $\atype$ if
$s : \atype$; 
it \emph{has base type} if
$\atype \in \Sorts$.
A term $s$ has a \emph{subterm} $t$, notation $s \suptermeq t$, if
(a) $s = t$, (b) $s = \abs{x}{s'}$ and $s' \suptermeq t$, or (c)
$s = \app{s_1}{s_2}$ and $s_1 \suptermeq t$ or $s_2 \suptermeq t$.
Finally, we define $\head(s) = \head(s_1)$
if $s = \app{s_1}{s_2}$, and $\head(s) = s$ otherwise.
\end{definition}

Note that any term $s$ has a form $\apps{t}{s_1}{s_n}$ with $n \geq 0$
and $t = \head(s)$ a variable, function symbol, or abstraction.
Separate from terms, we
use special expressions for matching and rewrite rules:

\begin{definition}[Meta-terms and patterns]\label{def:metaterm}
We fix a set $\M$, disjoint from $\F$ and $\V$, of
\emph{meta-variables}; each meta-variable is equipped with a \emph{type
declaration} $[\atype_1 \times \dots \times \atype_\mia] \arrtype \btype$
(where $\btype$ and all $\atype_i$ are simple types).
Meta-terms are expressions $s$ such that $s : \atype$ can be derived
for some type $\atype$ using 
(\textsf{V}), (\textsf{F}),
(\textsf{@}), ($\mathsf{\Lambda}$), and (\textsf{M}) below:

\begin{tabular}{llll}
(\textsf{M}) & $\meta{Z}{s_1,\dots,s_\mac} : \tau$ & if & $Z :
  [\atype_1 \times \dots \times \atype_\mia] \arrtype \atype_{\mia+1}
  \arrtype$ \\
  &&& \quad $\dots \arrtype \atype_\mac \arrtype \btype \in \M$ and\\
  &&&\quad $s_1 : \atype_1,\dots,s_\mac : \atype_\mac$
\end{tabular}

\noindent
We call $\mia$ the \emph{minimal arity} of $Z$
and write $\mia = \arity(Z)$.  A meta-term is a
\emph{pattern} if it has one of the forms
$\meta{Z}{x_1,\dots,x_\mia}$ with all
$x_i$ distinct variables and $\mia = \arity(Z)$;
$\abs{x}{\ell}$ with $x \in \V$ and $\ell$ a pattern;
or $\apps{a}{\ell_1}{\ell_n}$ with $a \in \F \cup \V$ and all
$\ell_i$ patterns ($n \geq 0$).
$\FMV(s)$ is the set of meta-variables occurring in a
meta-term $s$.
A pattern $\ell$ is \emph{fully extended} if for all occurrences of an
abstraction $\abs{y}{\ell'}$ in $\ell$, the bound variable $y$ is an
argument to all meta-variables in $\FMV(\ell')$.
It is \emph{linear}
if each meta-variable in $\FMV(\ell)$ occurs exactly once.
\end{definition}

Meta-variables are used in early forms of higher-order rewriting
(e.g., \cite{acz:78,klo:oos:raa:93})
and strike a balance between matching
modulo $\beta$ 
and syntactic matching.
%
%
%
Note 
that in earlier applications it is \emph{not} permitted to
give a meta-variable more arguments than its minimal arity.  We
allow
this because 
of applications in the DP framework.
However, in all our examples, meta-variable
applications in the unmodified rules take the expected number of
arguments.

Notationally, we will use $x,y,z$ for variables, $X,Y,Z$ for
meta-variables, $\avarormeta$ for symbols that could be variables or
meta-variables, $\afun,\bfun,\cfun$ or more suggestive notation for
function symbols, and $s,t,u,\linebreak
v,q,w$ for (meta-)terms.  Types are denoted
$\sigma,\tau$, and $\asort,\bsort$ are sorts.
We will regularly
overload notation and write $x \in \V$, $\afun \in \F$
or $Z \in \M$ without stating a type.  For meta-terms $Z[]$ we will
often omit the brackets, writing just $Z$.
In addition, notational conventions and definitions like $\head$ and
\emph{closed} carry over from terms to meta-terms; a meta-term $s$ is
closed if $\FV(s) = \emptyset$, even if $\FMV(s) \neq \emptyset$.


\begin{definition}[Substitution]
A \emph{meta-substitution} is a type-pre\-serving function $\gamma$
from variables and meta-variables to meta-\linebreak
terms;
if $Z : 
[\atype_1 \times \dots \times \atype_\mia] \arrtype
\btype$
then $\gamma(Z)$ has the form
$\abs{y_1 \dots y_\mia}{u} :
\atype_1 \arrtype \dots \arrtype \atype_\mia \arrtype \btype$.
Let $\domain(\gamma) = \{ x \in \V \mid \gamma(x) \neq x \} \cup
\{ Z \in \M \mid\linebreak \gamma(Z) \neq \abs{y_1 \dots y_{\arity(Z)}}{
\meta{Z}{y_1,\dots,y_{\arity(Z)}}} \}$ (the \emph{domain} of $\gamma$).
For\linebreak meta-variables $Z : [\atype_1 \times \dots \times \atype_\mia]
\arrtype \atype_{\mia+1} \arrtype \dots \arrtype \atype_\maa \arrtype
\asort$ with\linebreak $\asort \in \Sorts$ and for
$\mac$ with $\mia \leq \mac \leq \maa$,
we write $\gamma(Z) \approxp \abs{x_1 \dots x_\mac}{s}$ if\linebreak either
$\gamma(Z) = \abs{x_1 \dots x_\mac}{s}$\ext{,} or there is $\mia \leq i < \mac$
such that $\gamma(Z) = \abs{x_1 \dots x_i}{t}$ with $t$ not an
abstraction
and $s = \apps{t}{x_{i+1}}{x_\mac}$.
We let\linebreak $[\avarormeta_1:=s_1,\dots,\avarormeta_n:=s_n]$
be the
meta-substitution $\gamma$ with $\gamma(\avarormeta_i) = s_i$,
$\gamma(z) = z$ for $z \in \V \setminus \{\vec{\avarormeta}
\}$, and
$\gamma(Z) = \abs{y_1\dots
  y_{\arity(Z)}}{\meta{Z}{y_1,\dots, y_{\arity(Z)}}}$\linebreak
for $Z \in \M \setminus \{\vec{\avarormeta}\}$.
We will
also consider meta-substitutions with
infinite domain.
Even if the domain is infinite,
 for all $\avarormeta$ in $\domain(\gamma)$
 we assume
infinitely
many variables $x$ of all types with $x \notin
\FV(\gamma(\avarormeta))$.

A \emph{substitution} is a meta-substitution mapping everything in its
domain to terms.
The result $s\gamma$ of applying a meta-substitution $\gamma$ to a
term $s$ is obtained recursively:

\begin{tabular}{rclll}
$x\gamma$ & $=$ & $\gamma(x)$ & if & $x \in \V$\\
$\afun\gamma$ & $=$ & $\afun$ & if & $\afun \in \F$ \\
$(s\ t)\gamma$ & $=$ & $(s\gamma)\ (t\gamma)$ \\
$(\lambda x.s)\gamma$ & $=$ & $\lambda x.(s\gamma)$ & if & $\gamma(x) = x
  \wedge x \notin \FV(s\gamma)$
\end{tabular}

\noindent
For meta-terms, the result $s\gamma$ is obtained by the clauses above
and:

\noindent
$\meta{Z}{s_1,\dots,s_\mac}\gamma = t[x_1:=s_1\gamma,\dots,x_\mac:=
s_\mac\gamma]$ if $\gamma(Z) \approxp \abs{x_1 \dots x_\mac}{t}$
\end{definition}

Note that for $Z : [\atype_1 \times \dots \times \atype_\mia]
\arrtype \atype_{\mia+1} \arrtype \dots \arrtype \atype_\maa \arrtype
\asort$ with $Z \in \domain(\gamma)$ and $\mia \leq \mac \leq \maa$,
there is always exactly one $t$ such that $\gamma(Z) \approxp
\abs{x_1 \dots x_\mac}{t}$.
The result $s\gamma$ of applying a meta-substitution is
well-defined by induction on the multiset $\{\!\{ s \}\!\} \cup \{\!\{ \gamma(Z) \mid Z \in \FMV(s) \}\!\}$, with meta-terms compared by their sizes.

\smallskip
Essentially, applying a meta-substitution with meta-variables in its
domain combines a substitution with a $\beta$-development.
So
$\symb{d}\ (\abs{x}{\symb{sin}\ (\meta{Z}{x})})[Z
:=\abs{y}{\symb{plus}\ y\ x}]$ equals $\symb{d}\ 
(\abs{z}{\symb{sin}\ (\symb{plus}\ z\ x}))$, and $\meta{X}{\nil,
\symb{0}}[X:=\abs{x}{\symb{plus}\ (\symb{len}\ x)}]$ equals
$\symb{plus}\ (\symb{len}\ \nil)\ \symb{0}$.

\begin{definition}[Rules and rewriting]\label{def:rule}
A \emph{rule} is a pair $\ell \arrz r$ of closed meta-terms of the
same type such that $\ell$ is a pattern of the form
$\apps{\afun}{\ell_1}{\ell_n}$
with $\afun \in \F$ and $\FMV(r) \subseteq \FMV(\ell)$.
A set of rules $\Rules$ defines a rewrite
relation $\arr{\Rules}$ 
as the smallest monotonic relation on terms which includes:

\noindent
\begin{tabular}{l@{\hskip 2pt}c@{\hskip 8pt}l@{\hskip 8pt}cl}
(\textsf{Rule}) &
  $\ell\delta$ & $\arr{\Rules}$ & $r\delta$ & if $\ell \arrz r \in
  \Rules$ and $\delta$ a sub-\\
&&&& stitution on domain $\FMV(\ell)$ \\
(\textsf{Beta}) & $(\abs{x}{s})\ t$ & $\arr{\Rules}$ & $s[x:=t]$ \\
\end{tabular}

\noindent
We say $s \arr{\beta} t$ if $s \arr{\Rules} t$ is derived using a
(\textsf{Beta}) step.
A term $s$ is \emph{terminating} under $\arr{\Rules}$
if there is no
infinite reduction $s = s_0 \arr{\Rules} s_1 \arr{\Rules} \dots$, and $s$
is
$\beta$-normal if there is no $t$ 
with
$s \arr{\beta} t$.
Note that it is allowed to reduce at any position of a term, even
below a $\lambda$.
A relation
$\sqsupset$ is terminating if all terms are terminating
under
$\sqsupset$.
A set $\Rules$ of rules is terminating if $\arr{\Rules}$ is terminating.
%
The set $\Defineds \subseteq \F$ of \emph{defined symbols} consists
of those 
$\afun \in \F$ such that a rule $\apps{\afun}{\ell_1}{\ell_n}
\arrz r$ exists;
all other 
symbols are called \emph{constructors}.
\end{definition}

Note that $\Rules$ is allowed to be infinite -- which is useful for
instance to model polymorphic systems.  Also, right-hand sides of
rules do not have to be in $\beta$-normal form.
While this is rarely used in practical examples, non-$\beta$-normal
rules may arise through transformations, such as the one used in our
\refDef{def:ddp}.

\begin{example}\label{ex:mapintro}
Let $\F \supseteq \{ \nul : \nat,\ \suc : \nat \arrtype \nat,\ \nil :
\lijst,\linebreak \cons : \nat \arrtype \lijst \arrtype \lijst,\ 
\map : (\nat \arrtype \nat) \arrtype \lijst \arrtype \lijst\}$ and
consider the following rules $\Rules$:
\[
\begin{array}{rcl}
\map\ (\abs{x}{\meta{Z}{x}})\ \nil & \arrz & \nil \\
\map\ (\abs{x}{\meta{Z}{x}})\ (\cons\ H\ T) & \arrz &
  \cons\ \meta{Z}{H}\ (\map\ (\abs{x}{\meta{Z}{x}})\ T) \\
\end{array}
\]
Then $\map\ (\abs{y}{\nul})\ (\cons\ (\suc\ \nul)\ \nil) \arr{\Rules} 
\cons\ \nul\ (\map\ (\abs{y}{\nul})\ \nil) \arr{\Rules} \cons\ \nul\ 
\nil$.
Note that the
bound variable $y$ does not need to occur in the body of $\abs{y}{\nul}$
to match $\abs{x}{\meta{Z}{x}}$.
However, note also that
a term like
$\map\ \suc\ (\cons\ \nul\ \nil)$ can\emph{not} be reduced, because
$\suc$ does not instantiate $\abs{x}{\meta{Z}{x}}$. We could
alternatively consider the rules:
\[
\begin{array}{rcl}
\map\ Z\ \nil & \arrz & \nil \\
\map\ Z\ (\cons\ H\ T) & \arrz & \cons\ (Z\ H)\ (\map\ Z\ T) \\
\end{array}
\]
Here, $Z$ has a type declaration $[] \arrtype \nat \arrtype \nat$
instead of $[\nat] \arrtype \nat$, and we use explicit application.
Then $\map\ \suc\ (\cons\ \nul\ \nil) \arr{\Rules} \cons\ (\suc\ 
\nul)\ (\map\ \suc\ \nil)$.  However, we will often need explicit
$\beta$-reductions;
e.g., $\map\ (\abs{y}{\nul})\ (\cons\ (\suc\ \nul)\ \nil) \arr{\Rules}
\cons\ ((\abs{y}{\nul})\ \linebreak
(\suc\ \nul))\ (\map\ (\abs{y}{\nul})\ \nil)
\arr{\beta} \cons\ \nul\ (\map\ (\abs{y}{\nul})\ \nil)$.
\end{example}

For the set of terms to analyse for (non-)termination, it suffices to
consider a minimal number of arguments for each function symbol,
induced by the rewrite rules of the given AFSM. To capture this
minimal number of arguments, we introduce \emph{arity functions}.

\begin{definition}[Arity]\label{def:arity}
An \emph{arity function} is a function $\arity : \F \mapsto \N$ 
with
$0 \leq \arity(\afun) \leq \maa$ for all $\afun : \atype_1 \arrtype
\dots \arrtype \atype_\maa \arrtype \asort \in \F$.  A meta-term $s$
\emph{respects $\arity$} if any $\afun$ occurring in $s$ is applied to
at least $\arity(\afun)$ arguments.
$\Rules$ respects $\arity$ if
$\ell$ and $r$ respect $\arity$ for all $\ell \arrz r \in \Rules$.

For a fixed set of function symbols $\F$ and arity function $\arity$, we
say that the \emph{minimal arity} of $\afun$ in $\F$ is $\arity(\afun)$,
and the \emph{maximal arity} of $\afun$ is the unique number $\maa$ such
that $\afun : \atype_1 \arrtype \dots \arrtype \atype_m \arrtype \asort
\in \F$.
The set of \emph{arity-respecting terms}
is denoted $\Terms(\F,\V,\arity)$.
\end{definition}

An AFSM is
a triple $(\F, \Rules, \arity)$;
types of (meta-)variables can be derived from context.
However, if $\Rules$ res\-pects an arity function $\arity$, then if
there is any non-terminating term, there is one that respects
$\arity$ (see \refApp{app:arity}). 
So for fixed $\Rules$ we can set the arity function to give
the greatest possible values that $\Rules$ respects, and
we do not need to explicitly give $\arity$
(choosing the greatest possible minimal arities is always useful,
as it requires a termination proof for fewer terms).
Thus, we will typically speak of an AFSM $(\F,\Rules)$.

Note that, while we have suggestively used the same notation $\arity$
for the minimal arity of function symbols and meta-variables, the
minimal arity of meta-variables is fixed by their declaration.


\begin{example}[Ordinal recursion]\label{ex:ordrec}
Let $\F \supseteq \{
\nul \!:\! \symb{ord},\ \suc \!:\! \symb{ord} \!\arrtype\! \symb{ord},\linebreak
\symb{lim} \!:\! (\symb{nat} \!\arrtype\! \symb{ord}) \!\arrtype\! \symb{ord},\ 
\symb{rec} \!:\! \symb{ord} \!\arrtype\! \symb{nat} \!\arrtype\! (\symb{ord}
\!\arrtype\! \symb{nat} \!\arrtype\! \symb{nat}) \!\arrtype\linebreak
((\symb{nat} \arrtype
\symb{ord}) \arrtype (\symb{nat} \arrtype \symb{nat}) \arrtype
\symb{nat}) \arrtype \symb{nat} \}$
and $\Rules$ 
given by:
\[
\begin{array}{rcl}
\symb{rec}\ \nul\ K\ F\ G & \arrz & K \\
\symb{rec}\ (\suc\ X)\ K\ F\ G & \arrz & F\ X\ (\symb{rec}\ X\ K\ F\ G) \\
\symb{rec}\ (\symb{lim}\ H)\ K\ F\ G & \arrz & G\ H\ (\abs{m}{
  \symb{rec}\ (H\ m)\ K\ F\ G}) \\
\end{array}
\]
Then we can assume that $\arity(\symb{rec}) = 4$ without explicitly
giving $\arity$.

Observant readers may also notice that by the given constructors, the
type $\nat$ is not inhabited.  However, following \refDef{def:terms},
$\F$ contains infinitely many symbols of all types; thus, constructors
of all sorts (with minimal arity $0$) are implicitly present.
\end{example}

The two most common
formalisms in the context of
termination analysis of higher-order rewriting
are \emph{algebraic functional systems} (AFSs)
and \emph{higher-order rewriting  systems}~\cite{nip:91,mil:91} (HRSs),
often used with a pattern restriction.
AFSs are very similar to our AFSMs, but they use variables
for matching rather than meta-variables; this 
is
trivially
translated to the AFSM format, giving rules where all meta-variables
have minimal arity $0$, like the ``alternative'' rules in \refEx{ex:mapintro}.
%
HRSs use matching
modulo $\beta/\eta$, but the common restriction of \emph{pattern HRSs}
can be directly translated into ASFMs, 
provided
terms are
$\beta$-normalised after every reduction step.
Even without strategy restrictions, termination of the obtained
AFSM still implies termination of the original HRS;
for second-order systems, termination is equivalent.
AFSMs can also naturally encode CRSs~\cite{klo:oos:raa:93}
and several applicative systems
(cf. \cite[Chapter 3]{kop:12}).


\subsection{Computability}

A common technique in higher-order termination is Tait and Gi\-rard's
\emph{computability} notion~\cite{tai:67}.
There are several ways to define computability predicates; here we
follow, e.g., \cite{bla:00,bla:jou:oka:02,bla:jou:rub:15,bla:16} 
in considering \emph{accessible meta-variables} using strictly positive
inductive types.
The definition 
presented
below is adapted from these works, both to
account for the altered formalism and to introduce (and obtain
termination of) a relation $\accreduce{C}$ that we will use in
\refThm{thm:staticsubtermproc}.
This allows for a minimal presentation that avoids the use of ordinals
that would otherwise be needed to obtain $\accreduce{C}$.

\CKold{Observation 2: the original sources are \emph{old}, which makes
  them hard to find and to read: the notion which many papers (in my
  ``articles'' directory) refer to is ``Tait and Girard's notion of
  computability'', but not a single one of them seems to include a
  reference.  This is probably with good reason.  Jaco van de Pol has
  conveniently done some of the lookup work (see
  \texttt{pol\_snproofs.pdf}), which leads to Tait's notion of
  computability in \cite{tai:67} (I also seem to have cited this one in
  my thesis; see \texttt{tait.pdf}), and Girard's version hidden in a
  French-language thesis from before the time people put
  those things online (\url{http://www.worldcat.org/title/interpretation-fonctionnelle-et-elimination-des-coupures-de-larithmetique-dordre-superieur/oclc/65409853?referer=di&ht=edition}).
  Blanqui is also good for sources, although reading his papers I
  think ``a right mess'' is a reasonable summary.  See
  \texttt{blanqui\_tcs16.pdf} for background.  Section 4.6 particularly
  gets interesting, but also before that\dots\ argbargl.}

To define computability, we use the notion of an \emph{RC-set}:

\edef\defRCset{\number\value{theorem}}
\edef\defRCsetSec{\number\value{section}}
\begin{definition}\label{def:RCset}
A \emph{set of reducibility candidates}, or \emph{RC-set}, for
a rewrite relation $\arr{\Rules}$ of an AFSM
is a set $I$ of base-type terms $s : \asort$ such that:
every term in $I$ is terminating under $\arr{\Rules}$;
$I$ is closed under $\arr{\Rules}$ (so if $s \in I$ and $s
  \arr{\Rules} t$ then $t \in I$);
if $s = \apps{x}{s_1}{s_n}$ with $x \in \V$ or $s =
  \apps{(\abs{x}{u})}{s_0}{s_n}$ with $n \geq 0$, and for all $t$ with
  $s \arr{\Rules} t$ we have $t \in I$, then $s \in I$.

We define $I$-computability for an RC-set $I$ by induction on types:
$s : \asort$ is $I$-computable if $s \in I$ ($\asort \in \Sorts$);
 $s : \atype \arrtype \btype$ is $I$-computable if for all
  $t : \atype$ that are $I$-computable, $\app{s}{t}$ is $I$-computable.
\end{definition}

The traditional notion of computability is obtained by taking for $I$
the set of
all terminating base-type terms.
However, we can do better, using the notion of
\emph{accessible arguments},
applied to termination analysis also in the
\emph{General Schema} \cite{bla:jou:oka:02}, the
\emph{Computability Path Ordering}~\cite{bla:jou:rub:15}, and the
\emph{Computability Closure}~\cite{bla:16}.

\begin{definition}[Accessible arguments]\label{def:accArgs}
We fix a quasi-ordering $\greqsort$ on $\Sorts$ \ext{with well-founded}
strict part $\grsort\ :=\ \greqsort \setminus \leqsort$. %
For 
$\atype \equiv\linebreak \atype_1 \!\arrtype\! \dots \!\arrtype\! \atype_\maa
\!\arrtype\! \bsort$ (with $\bsort \in \Sorts$) and sort $\asort$,
let
$\ext{\asort \gracsortup \atype}$ if $\asort \greqsort \bsort$ and\linebreak
$\ext{\asort \gracsortdown \atype_i}$ \ext{for all $i$}, and
let
$\ext{\asort
\gracsortdown \atype}$ if $\asort \grsort \bsort$ and $\ext{\asort
\gracsortup \atype_i}$ \ext{for all $i$}.
\ext{(Here $\ext{\asort \gracsortup \atype}$ corresponds to
  ``$\asort$ occurs only positively in $\atype$'' in
  \cite{bla:00,bla:jou:oka:02,bla:jou:rub:15}.)}

For $\afun : \atype_1 \arrtype \dots \arrtype \atype_\maa \arrtype
\asort \in \F$, let $\Acc(\afun) = \{ i \mid 1 \leq i \leq \maa
\wedge \ext{\asort \gracsortup \atype_i} \}$.
For $x : \atype_1 \arrtype \dots \arrtype \atype_\maa \arrtype \asort
\in \V$, let $\Acc(x) = \{ i \mid 1 \leq i \leq \maa \wedge \atype_i$
has the form $\btype_1 \arrtype \dots \arrtype \btype_n \arrtype
\bsort$ with\linebreak 
$\asort \greqsort \bsort \}$.
We
write $s \gracc t$ if either $s = t$\ext{,} or $s = \abs{x}{s'}$ and $s'
\gracc t$\ext{,} or $s = \apps{a}{s_1}{s_n}$ with $a \in \F \cup \V$ and
$s_i \gracc t$ for some $i\in \Acc(a)$.%
\end{definition}





\begin{example}\label{ex:lim}
Consider a quasi-ordering $\greqsort$ such that
$\symb{ord} \grsort \symb{nat}$. In \refEx{ex:ordrec}, we then have
$\ext{\symb{ord} \: \gracsortup \: \symb{nat} \arrtype \symb{ord}}$.  Therefore,
$1 \in \Acc(\symb{lim})$, which gives $\symb{lim}\ H \gracc H$.
\end{example}

\begin{theorem}\label{thm:defC}
%
Let $\apps{\afun}{s_1}{s_\maa} \accreduce{I} \apps{s_i}{t_1}{t_n}$ if
both sides have base type, $i \in \Acc(\afun)$,
and all $t_j$ are $I$-computable.
There is an RC-set $C$ such that $C = \{ s : \asort \mid \asort \in
\Sorts \wedge s$ is terminating under $\arr{\Rules} \cup \accreduce{C} \ext{\wedge} {}$
if $s \arrr{\Rules} \apps{\afun}{s_1}{s_\maa} : \asort$ then $s_i$ is
$C$-computable for all $i \in \Acc(\afun) \}$.
\end{theorem}

\begin{proof}[Proof sketch]
This follows the proof in, e.g., \cite{bla:jou:oka:02,bla:jou:rub:15},
defining $C$ as the fixpoint of a monotone function operating on RC-sets.
%

The full proof is available in \refApp{app:computability}.
\end{proof}


\section{Higher-order dependency pairs}\label{sec:dp}

In this section we
transpose the definitions of dynamic and static
dependency pairs~\cite{sak:wat:sak:01,kop:raa:12,bla:06,sak:kus:05,kus:iso:sak:bla:09,suz:kus:bla:11}
to AFSMs and thus formulate them in a single unified language.
We
add the new features of \emph{meta-variable conditions},
\emph{formative reductions}, and \emph{computable chains}.

\subsection{Common definitions}

Although we keep
the first-order terminology of dependency \emph{pairs},
the setting with meta-variables makes it better to
use
\emph{triples}.

\begin{definition}[Dependency Pair]\label{def:dp}
A \emph{dependency pair (DP)}
is a
triple $\ell \arrdp p\ (A)$, where $\ell$ is a closed pattern
$\apps{\afun}{\ell_1}{\ell_n}$ with $n \geq \arity(\afun)$, $p$ is
a meta-term, and $A$ is a set of \emph{meta-variable conditions}: pairs
$Z : i$ indicating that $Z$ regards its $i^{\text{th}}$
argument.
A substitution $\gamma$ \emph{respects} a set of
meta-variable conditions $A$ if for all $Z : i$ in $A$ we have
$\gamma(Z) \approxp \abs{x_1 \dots x_\mac}{t}$ with $x_i \in \FV(t)$.
DPs will be used only with substitutions that respect
their meta-variable conditions.  We call the DP \emph{collapsing} if $p$
has the form $\apps{\meta{Z}{s_1,\dots,s_\mac}}{t_1}{t_j}$ with
$\mac,j \geq 0$.
A set of DPs is collapsing if it contains a collapsing DP.
\end{definition}

There are two approaches to generate DPs, originating from
distinct lines of work 
\cite{kop:raa:12,kus:iso:sak:bla:09}.
As in the first-order setting,
both approaches employ \emph{marked symbols}:

\begin{definition}[Marked symbols]\label{def:marked}
Define $\F^\sharp := \F \uplus \{ \afun^\sharp : \atype \mid \afun :
\atype \in \Defineds \}$,
and
$\arity(\afun^\sharp) :=
\arity(\afun)$.
For a meta-term $s$, 
let $s^\sharp = \apps{\afun^\sharp}{s_1}{s_\mia}$
if $s := \apps{\afun}{s_1}{s_\mia}$ with $\mia = \arity(\afun)$ \ext{and $\afun \in \Defineds$;}
$s^\sharp := s$ otherwise.
\end{definition}

Note that $(\apps{\afun}{s_1}{s_n})^\sharp$ is simply
$\apps{\afun}{s_1}{s_n}$ if $n > \arity(\afun)$.

Moreover, we will consider \emph{candidates}.  In the first-order
setting these are subterms of the right-hand sides of rules whose
root symbol is defined.  In the current setting, we have to
consider also meta-variables as well as rules whose right-hand side is not
$\beta$-normal.

\begin{definition}[$\beta$-reduced-sub-meta-term,
  $\bsuptermeq{\beta}$, $\bsuptermeq{A}$]
A meta-term $s$ has a \emph{$\beta$-reduced-sub-meta-term} $t$
(shortly, \emph{BRSMT}),
notation $s \bsuptermeq{\beta} t$, if there exists a set of
meta-variable conditions $A$ such that $s \bsuptermeq{A} t$.
Here $s \bsuptermeq{A} t$ holds if at least one of the following holds:
\begin{itemize}
\item $s = \apps{t}{s_1}{s_n}$ for some $n \geq 0$
\item $s = \abs{x}{u}$ and $u \bsuptermeq{A} t$
\item $s = \apps{(\abs{x}{u})}{s_0}{s_n}$ and
  $\apps{u[x:=s_0]}{s_1}{s_n} \bsuptermeq{A} t$
\item $s = \apps{a}{s_1}{s_n}$ and $s_i \bsuptermeq{A} t$ for some
  $i \in \{1,\dots,n\}$, with $a$ an abstraction,
  meta-variable application, or element of $\V \cup \F$
\item $s = \apps{\meta{Z}{t_1,\dots,t_\mac}}{s_1}{s_n}$ and $t_i
  \bsuptermeq{A} t$ for some $i \in \{1,\dots,\mac\}$ such that
  $(Z : i) \in A$
\end{itemize}
\end{definition}

Essentially, $s \bsuptermeq{A} t$ means that $t$ can be reached from
$s$ by taking $\beta$-reductions at the root and ``subterm''-steps, and
$Z : i$ must be in $A$ whenever we pass into argument $i$ of a
meta-variable $Z$.
We also do not include subterms of $u$ in
$\apps{(\abs{x}{u})}{s_0}{s_n}$. 
We use $\beta$-reduced-sub-meta-terms in the following definition of
\emph{candidates}:

\begin{definition}[Candidates]\label{def:candidates}
For a 
meta-term $s$, the set $\cand(s)$ of
\emph{candidates of $s$} consists of those pairs $t\ (A)$
with
(a) $s \bsuptermeq{A} t$, (b) $t$ has\linebreak either the form
$\apps{\identifier{f}}{s_1}{s_n}$ with $\identifier{f} \in \Defineds$
and $n \geq \arity(\identifier{f})$, or
 the form
$\apps{\meta{Z}{t_1,\dots,t_\mac}}{s_1}{s_n}$,
(c) $t$ does \emph{not} have the form $\meta{Z}{x_1,\dots, x_{
\arity(Z)}}$ with all $x_i$ distinct variables,
and (d) 
there is no 
$A' \subsetneq A$
with
$s \bsuptermeq{A'} t$.
\end{definition}

\begin{example}\label{ex:metanonapplied}
In AFSMs where 
all meta-variables have minimal arity $0$,
the set
$\cand(s)$ for a meta-term $s$ consists 
of the pairs $t\ 
(\emptyset)$ where $t$ is a BRSMT of $s$ that has either the form
$\apps{Z}{s_1}{s_n}$ with $Z$ a meta-variable and $n > 0$, or
$\apps{\afun}{s_1}{s_n}$ with $\afun \in \Defineds$ and $n \geq
\arity(\afun)$.

In the AFSM of \refEx{ex:ordrec}, the set
$\cand(G\ H\ (\abs{m}{\symb{rec}\ (H\ m)\ K\ F\ G}))$
therefore consists of
$G\ H\ (\abs{m}{\symb{rec}\ (H\ m)\ K\ F\ G})\ 
(\emptyset)$ and $G\ H\ (\emptyset)$ and $H\ m\ (\emptyset)$
as well as $\symb{rec}\ (H\ m)\ K\ F\ G\ (\emptyset)$.
\end{example}

If some 
meta-variables \emph{do}
take arguments, the forms considered in \refEx{ex:metanonapplied}
do not suffice: we must also consider candidates such as
$\meta{Z}{s_1,\dots,s_\mac}\ (A)$.  In addition, the meta-variable
conditions
matter:
candidates are pairs $t\ (A)$ where
$A$ contains exactly those pairs $Z : i$ where we pass
through the $i^{\text{th}}$ argument of $Z$ to reach
$t$. 

\begin{example}\label{ex:metaapplied}
Consider an AFSM with the signature from \refEx{ex:ordrec} but a rule
using meta-variables with larger minimal arities:
\[
\begin{array}{c}
\symb{rec}\ (\symb{lim}\ (\abs{n}{\meta{H}{n}}))\ K\ 
  (\abs{xn}{\meta{F}{x,n}})\ (\abs{fg}{\meta{G}{f,g}})\ 
 \arrz \\
  \meta{G}{\abs{n}{\meta{H}{n}},\ 
  \abs{m}{\symb{rec}\ \meta{H}{m}\ K\ (\abs{xn}{\meta{F}{x,n}})\ 
  (\abs{fg}{\meta{G}{f,g}})}}
\end{array}
\]
The candidates of the right-hand side are:
\begin{itemize}
\item
$\symb{rec}\ \meta{H}{m}\ K\ (\abs{xn}{\meta{F}{x,n}})\ 
(\abs{fg}{\meta{G}{f,g}})\ (\{G : 2 \})$ and
\item
$\meta{G}{\abs{n}{\meta{H}{n}},\ 
  \abs{m}{\symb{rec}\ \meta{H}{m}\ K\ (\abs{xn}{\meta{F}{x,n}})\ 
  (\abs{fg}{\meta{G}{f,g}})}}\ (\emptyset)$
\end{itemize}
Note that for instance $\meta{H}{m}$ is \emph{not} the source of a
candidate, as $m$ is a variable and $H$ has minimal arity $1$ (as the
left-hand side of the rule shows).  Note also that $G$ cannot be
partially applied, so there is no counterpart to the candidate $G\ 
H\ (\emptyset)$ in \refEx{ex:metanonapplied}.
\end{example}


\emph{Dynamic} DPs involve also collapsing DPs.
This makes the notion of chains a bit
more complicated than its first-order
analogue.

\begin{definition}[Dependency chain]\label{def:chain}
Let $\P$ be a set of DPs and $\Rules$ a set of rules.
An infinite \emph{$(\P,\Rules)$-dependency chain} (or just
$(\P,\Rules)$-chain) is a sequence
$[(\rho_i,s_i,t_i) \mid i \in \N]$ where each $\rho_i \in \P \cup
\{ \mathtt{beta} \}$ and all $s_i,t_i$ are terms, such that for all
$i$:
\begin{enumerate}
\item\label{depchain:beta}
  if $\rho_i = \mathtt{beta}$, then 
  $s_i =
  \apps{(\abs{x}{u})\ v}{w_1}{w_n}$ and either
  (a) $n > 0$ and $t_i = \apps{u[x:=v]}{w_1}{w_n}$, or
  (b) $n = 0$ and $t_i = q^\sharp[x:=v]$ for some
  $q$ with $u \suptermeq q$ and
  $x \in \FV(q)$ but
  $q \neq x$.
\item\label{depchain:dp}
  if $\rho_i = \ell_i \arrdp p_i\ (A_i) \in \P$ then there exists
   a substitution $\gamma$ on domain $\FMV(\ell_i) \cup \FMV(p_i) \cup
   \FV(p_i)$ such that $\gamma$ maps all variables in $\FV(p_i)$ to
   fresh variables, $s_i = \ell_i\gamma$, and:
   \begin{enumerate}
   \item\label{depchain:dp:instance}
    if $p_i$ is an application or symbol $\afun^\sharp$,
     then $t_i = p_i\gamma$
   \item\label{depchain:dp:meta}
      if $p_i 
      =
      \meta{Z}{u_1,\dots,
      u_\mac}$ and $\gamma(Z) \approxp \abs{x_1 \dots x_\mac}{w}$, then
      $t_i = v^\sharp[x_1:=u_1\gamma,\dots,x_\mac:=u_\mac\gamma]$ for some
      non-variable subterm $v$ of $w$ such that $\{ x_1,\dots,x_\mac \}
      \cap \FV(v) \neq \emptyset$
    \item for all $(Z : j) \in A_i$: if $\gamma(Z) \approxp \abs{x_1
      \dots x_\mac}{u}$ then $x_j \in \FV(u)$
   \end{enumerate}
\item\label{depchain:reduce}
  $t_i = s_{i+1}$ or we can write $t_i = \apps{\afun}{u_1}{u_n
  }$ with $\afun \in \F^\sharp$, $s_{i+1} = \apps{\afun}{w_1}{w_n}$ and
  each $u_j \arrr{\Rules} w_j$
\end{enumerate}
\end{definition}


Both cases~\ref{depchain:beta}
and~\ref{depchain:dp:meta} essentially perform a $\beta$-step and
then mark a \emph{specific} subterm of the result: this subterm
previously occurred
between a $\lambda$-abstraction and an occurrence of its bound variable.
This
often makes it possible to use reduction triples that do not satisfy the
subterm property, as observed in~\cite{kop:raa:12}
and \refThm{thm:tagredpair}.

Dependency chains
can exhibit some
particular properties:

\begin{definition}[Minimal chain, formative chain, formative reduction]\label{def:formative}
A $(\P,\Rules)$-chain \ext{$[(\rho_i,s_i,t_i) \mid i \in \N]$} is
\emph{minimal} if the strict subterms of all
$t_i$ are terminating under $\arr{\Rules}$.
\ext{It} is \emph{formative} if for all $i$ with
$\rho_{i+1}$ having the form $\ell_{i+1} \arrdp r_{i+1}\ (A) \in \P$,
the reduction $t_i \arrr{\Rules} s_{i+1}$ is $\ell_{i+1}$-formative.

Here, for a pattern $\ell$, substitution $\gamma$ and
term $s$, a reduction $s \arrr{\Rules} \ell\gamma$ is
\emph{$\ell$-formative} if 
one of the following statements holds:
\begin{itemize}
\item $\ell$ is not a fully extended linear pattern
\item $\ell$ is a meta-variable application $\meta{Z}{x_1,\dots,x_\mia}$
  and $s = \ell\gamma$
\item $s = \apps{a}{s_1}{s_n}$ and $\ell = \apps{a}{\ell_1}{\ell_n}$
  with $a \in \F^\sharp \cup \V$ and each $s_i \arrr{\Rules} \ell_i
  \gamma$ by an $\ell_i$-formative reduction
\item $s = \abs{x}{s'}$ and $\ell = \abs{x}{\ell'}$ and $s' \arrr{\Rules}
  \ell'\gamma$ by an $\ell'$-formative reduction
\item $s = \apps{(\abs{x}{u})\ v}{w_1}{w_n}$ and
  $\apps{u[x:=v]}{w_1}{w_n} \arrr{\Rules} \ell\gamma$ by an
  $\ell$-formative reduction
\item $\ell$ is not a meta-variable application, and there
  exist 
  $\delta$ and $\ell' \arrz r' \in \Rules$ and meta-variables
  $Z_1 \dots Z_n$ ($n \geq 0$)
  such that $s \arrr{\Rules} (\apps{\ell'}{Z_1}{Z_n})\delta$ by an
  $\ell'$-formative reduction, and $(\apps{r'}{Z_1}{Z_n})\delta
  \arrr{\Rules} \ell\gamma$ by an $\ell$-formative reduction.
\end{itemize}
\end{definition}

Formative reductions are used as a proof technique in
\cite{kop:raa:12} and are formally introduced for the first-order
DP framework in \cite{fuh:kop:14}.
The property will be essential in our Theorems
\ref{thm:tagredpair} and \ref{thm:formativeproc}.

\subsection{Dynamic higher-order dependency pairs}\label{sec:ddp}

With these preparations, we move on to \emph{dynamic} DPs.
Since rules of functional type 
sometimes cause non-termination
only in a cer\-tain applicative context
(e.g., 
if
$\Rules = \{ \afun\ \nul \arrz \abs{x}{\afun\ x\ x} \}$,
then
$\afun\ \nul$ is terminating, but $\afun\ \nul\ \nul$ is
not), 
we use
an
extended set of rules to include applicative contexts, 
using
a
variant of $\eta$-saturation 
\cite{hir:mid:zan:08}.

\begin{definition}[$\DDP$]\label{def:ddp}
Let $\RulesEta :=
\{ \apps{\ell}{Z_1}{Z_i} \arrz \apps{r}{Z_1}{Z_i} \mid \ell \arrz r \in
\Rules
\wedge \ell : \atype_1 \arrtype \dots \arrtype \atype_\maa \arrtype
\asort
\wedge 0 \leq i \leq \maa
\wedge Z_1 : \atype_1,\dots,Z_i : \atype_i
  \text{ fresh meta-variables} \}$.
%
%
Now $\DDP(\Rules) = \{ \ell^\sharp \arrdp p^\sharp\ (A) \mid \ell \arrz
r \in \RulesEta \wedge p\ (A) \in \cand(r) \wedge \neg (\ell
\supterm p) \}$.
\end{definition}

  \begin{remark}\label{rem:removeuseless}
  The corresponding definition in~\cite{kop:12} also excludes all DPs
  $(\apps{\ell}{Z_1}{Z_i})^\sharp \arrdp p^\sharp\ (A)$ if
  $\DDP(\Rules)$ already contains a DP $\ell^\sharp \arrdp p^\sharp\ 
  (A)$ -- so most of those DPs generated from rules in
  $\RulesEta \setminus \Rules$.
  For a simpler definition, we haven chosen not to do \ext{so}.
  \ext{Instead of excluding these DPs immediately, we can remove them}
  afterwards using a processor
(see \refThm{thm:removeuseless} in the appendix).
  \end{remark}

\begin{example}\label{ex:derivdynamic}
Consider an AFSM $(\F,\Rules)$ with $\F \supseteq \{ \symb{sin},
\symb{cos} : \real \arrtype \real,\ \symb{times} :
\real \arrtype \real \arrtype \real,\  \deriv
: (\real \arrtype \real) \arrtype \real \arrtype
\real \}$ and $\Rules = \{
\deriv\ (\abs{x}{\symb{sin}\ \meta{F}{x}}) 
\arrz 
  \abs{y}{\symb{times}\ (\deriv\ (\abs{x}{\meta{F}{x}})\ y)\ 
  (\symb{cos}\ \meta{F}{y})} 
\}
$.
Then $\RulesEta = \Rules \cup \linebreak
\{ \deriv\ (\abs{x}{\symb{sin}\ \meta{F}{x}})\ Y \arrz
(\abs{y}{\symb{times}\ (\deriv\ (\abs{x}{\meta{F}{x}})\ y)\ 
(\symb{cos}\ \linebreak
\meta{F}{y})})\ Y \}$.
Then $\DDP(\Rules)$ consists of:
\[
\begin{array}{rcll}
\deriv^\sharp\ (\abs{x}{\symb{sin}\ \meta{F}{x}}) & \arrdp &
  \deriv\ (\abs{x}{\meta{F}{x}})\ y & (\emptyset) \\
\deriv^\sharp\ (\abs{x}{\symb{sin}\ \meta{F}{x}}) & \arrdp &
  \deriv^\sharp\ (\abs{x}{\meta{F}{x}}) & (\emptyset) \\
\deriv\ (\abs{x}{\symb{sin}\ \meta{F}{x}})\ Y & \arrdp &
  \deriv\ (\abs{x}{\meta{F}{x}})\ Y & (\emptyset) \\
\deriv\ (\abs{x}{\symb{sin}\ \meta{F}{x}})\ Y & \arrdp &
  \deriv^\sharp\ (\abs{x}{\meta{F}{x}}) & (\emptyset) \\
\deriv\ (\abs{x}{\symb{sin}\ \meta{F}{x}})\ Y & \arrdp &
  \meta{F}{Y} & (\emptyset) \\
\end{array}
\]
The first two DPs come from
$\Rules$,
the last three from
$\RulesEta \setminus \Rules$.
\end{example}

As observed before, $\DDP(\Rules)$ may contain collapsing dependency
pairs $\ell \arrdp \apps{\meta{Z}{t_1,\dots,t_\mac}}{s_1}{s_n}\ (A)$,
in contrast to the first-order DP framework.
This is somewhat problematic for extending techniques like the
subterm criterion or the dependency graph processor that rely on
the shape of the right-hand side of DPs.

\begin{example}\label{ex:map:ddp}
For $\Rules$ the first two rules in \refEx{ex:mapintro}, $\DDP(\Rules) =$
\[
\left\{
\begin{array}{rcll}
\symb{map}^\sharp\ (\abs{x}{\meta{Z}{x}})\ (\symb{cons}\ H\ T) & \arrdp &
  \symb{map}^\sharp\ (\abs{x}{\meta{Z}{x}})\ T & (\emptyset) \\
\symb{map}^\sharp\ (\abs{x}{\meta{Z}{x}})\ (\symb{cons}\ H\ T) & \arrdp &
  \meta{Z}{H} & (\emptyset) \\
\end{array}
\right\}
\]
\end{example}

Key to the DP framework is the relationship between
dependency chains and termination: an AFSM with rules $\Rules$ is
terminating if and only if there is no $(\DDP(\Rules),
\Rules)$-dependency chain.
Indeed, we can limit interest to specific chains, following
\refDef{def:formative}.


\edef\thmDdpChain{\number\value{theorem}}
\edef\thmDdpChainSec{\number\value{section}}
\newcommand{\ddpChainTheThm}{
\begin{theorem}[Thm.\ 6.44 in~\cite{kop:12}]\label{thm:chain}
If 
$\arr{\Rules}$
is non-terminating, then there is an infinite
minimal formative $(\DDP(\Rules),\Rules)$-
chain.
If there is an infinite $(\DDP(\Rules),\Rules)$-
chain, then
$\arr{\Rules}$ is non-terminating.
\end{theorem}
}
\ddpChainTheThm

\begin{proof}[Proof sketch]
%
The proof of the first claim follows the proof of
\cite[Thm.~5.7]{kop:raa:12}:
we select a \emph{minimal} non-terminating term (MNT)
$s$ (all whose subterms terminate) and an infinite reduction starting
in $s$. 
Then we stepwise build an infinite minimal $(\DDP(\Rules),\Rules)
$-dependency chain as follows.
If $s = \apps{(\abs{x}{u})}{s_0}{s_n}$ with $n > 0$, then also
$\apps{u[x:=s_0]}{s_1}{s_n} =: s'$ is non-terminating; we continue with
a MNT subterm $w$ of $s'$.
Otherwise $s = \apps{\afun}{s_1}{s_n}$ and there is $\ell \arrz r
\in \RulesEta$ such that $s \arrr{\Rules} \ell\gamma$ by reductions in
the $s_i$, and $r\gamma$ is still non-terminating.  We can identify a
candidate $t\ (A)$ of $r$ such that $\gamma$ respects $A$ and $t\gamma$
is a MNT subterm of $r\gamma$; we continue with
$t\gamma$.
For the \emph{formative} property, we note that if $s \arrr{\Rules}
\ell\gamma$ and $s$ terminates, then $s \arrr{\Rules} \ell\delta$ by
an $\ell$-formative reduction for some $\delta$ where each $\delta(Z)
\arrr{\Rules} \gamma(Z)$; this follows by induction first on $s$ using
$\arr{\Rules} \mathop{\cup} \supterm$, second on the reduction length.

For the second claim, we show by induction on the definition of
$\bsuptermeq{A}$ that $s \bsuptermeq{A} t$ implies $s\gamma\ (\supterm
\mathop{\cup} \arr{\beta})^*\ t\gamma$ for all substitutions which respect
$\gamma$; thus, any infinite $(\DDP(\Rules),\Rules)$-chain induces an
infinite $(\arr{\Rules} \mathop{\cup} \supterm)$-reduction, which
contradicts termination of $\arr{\Rules}$.


The full proof is available in \refApp{app:ddp}.
\end{proof}

\refThm{thm:chain} is similar to \cite[Thm.~5.7]{kop:raa:12}, but
provides progress by considering AFSMs and meta-variable conditions,
and \ext{by} regarding formative chains; also, in~\cite{kop:raa:12} the second
statement holds only if all left-hand sides in $\Rules$ are linear,
as their definition of $\DDP$ replaces fresh variables in the
right-hand sides of DPs by constants.
Here, this is not needed due to the
\ext{distinction between variables and meta-variables}.

\begin{example}[Encoding the untyped $\lambda$-calculus]\label{ex:lambdadynamic}
Consider an\linebreak AFSM with $\F \supseteq \{ \symb{ap} : \symb{o} \arrtype
\symb{o} \arrtype \symb{o},
\ \symb{lm} : (\symb{o} \arrtype \symb{o}) \arrtype \symb{o} \}$ and
$\Rules = \{ \symb{ap}\ (\symb{lm}\ F) \arrz F \}$ (note that the only
rule has type $\symb{o} \arrtype \symb{o}$).  Then
$\RulesEta = \Rules \cup \{ \symb{ap}\ (\symb{lm}\ F)\ X \arrz F\ X \}$,
and $\DDP(\Rules) = \{
\symb{ap}^\sharp\ (\symb{lm}\ F) \arrdp F\ (\emptyset),\ 
\symb{ap}\ (\symb{lm}\ F)\ X \arrdp F\ X\ (\emptyset)
\}$.
There is an infinite dependency chain with, for each \emph{odd}
integer $i$:
$\rho_i = \symb{ap}\ (\symb{lm}\ F)\ X \arrdp F\ X\ 
  (\emptyset)$ and
$\rho_{i+1} = \mathtt{beta}$ and
$s_i = t_{i+1} = \symb{ap}\ (\symb{lm}\ (\abs{y}{\symb{ap}\ y\ 
  y}))\ (\symb{lm}\ (\abs{y}{\symb{ap}\ y\ y}))$ and
$t_i = s_{i+1} = (\abs{y}{\symb{ap}\ y\ y})\ (\symb{lm}\ 
  (\abs{y}{\symb{ap}\ y\ y}))$.
Note that in the ``subterm'' step in the chain, always $v =
\symb{ap}\ y\ y$ with $v$ as in \refDef{def:formative}.
\end{example}

\subsection{Static higher-order dependency pairs}\label{sec:static}

Unlike the dynamic approach, which may be used for all AFSMs, the
static approach can be applied only on systems whose rules are
\emph{accessible function passing (AFP)}.  Intuitively:
meta-variables of a higher type may occur only in ``safe'' places in
the left-hand sides of rules.
Rules like \refEx{ex:lambdadynamic},
where a higher-order meta-variable is lifted out of a base-type term,
are not admitted.

\begin{definition}[Accessible function passing]\label{def:apfp}
An AFSM $(\F,\Rules,\linebreak
\arity)$ is \emph{accessible function passing (AFP)}
if there exists a sort ordering $\greqsort$ following
\refDef{def:accArgs} such that:
\begin{itemize}
\item for all $\afun : \atype_1 \arrtype \dots \arrtype \atype_\maa
  \arrtype \asort$ we have $\arity(\afun) = \maa$;
\item for all $\apps{\afun}{\ell_1}{\ell_\maa} \arrz r \in \Rules$ and all
  $Z \in \FMV(r)$: there are 
  va\-riables $x_1,\dots,x_k$
  and some
  $i$ such that $\ell_i \gracc \meta{Z}{x_1,\dots,x_k}$.
\end{itemize}
\end{definition}

This definition is strictly more liberal than the notions of
\emph{plain function passing} in
\cite{kus:iso:sak:bla:09,suz:kus:bla:11} as adapted to AFSMs; this will
allow us to handle examples like ordinal recursion (\refEx{ex:ordrec})
which are not covered \ext{by} \cite{kus:iso:sak:bla:09,suz:kus:bla:11}.
However, note that \cite{kus:iso:sak:bla:09,suz:kus:bla:11} consider
a different formalism, which does take polymorphism and rules whose
left-hand side is not a pattern into account (which we do not
consider).
Our restriction more closely \ext{resembles} the ``admissible'' rules
in \cite{bla:06} which are defined using a pattern computability
closure \cite{bla:00}.

\begin{example}
The AFSM from \refEx{ex:mapintro} is obviously AFP (for instance by
equating all types under $\greqsort$).
The AFSM from \refEx{ex:ordrec} is AFP if a
sort ordering $\symb{ord} \grsort \symb{nat}$ is chosen (following
\refEx{ex:lim}).
The AFSM from \refEx{ex:lambdadynamic} is not, because $\arity(
\symb{ap}) = 1$.
An AFSM with the same signature but $\Rules = \{ \symb{ap}\ 
(\symb{lm}\ F)\ X \arrz F\ X \}$ (and $\arity(\symb{ap}) = 2$) is
not AFP either, because $\Acc(\symb{lm}) = \emptyset$.
This is good because, as we will see, the set $\SDP(\Rules)$ of static
DPs is empty, which would lead us to falsely conclude termination
without the restriction.
\end{example}

The restriction on arities excludes rules of non-base type and makes
sure that always $(\apps{\afun}{s_1}{s_n})^\sharp =
\apps{\afun^\sharp}{s_1}{s_n}$.  
We can transform any AFSM to satisfy this restriction:

\begin{definition}[$\Rules^\uparrow$]\label{def:etalong}
Given 
rules
$\Rules$, let their \emph{$\eta$-expansion} $\Rules^\uparrow = \{
\etalong{(\apps{\ell}{Z_1}{Z_\maa})}
\ \arrz \etalong{(\apps{r}{Z_1}{Z_\maa})}
\mid \ell \arrz r \in \Rules$ with $r : \atype_1 \arrtype \dots \arrtype
\atype_\maa \arrtype \asort$,
$\asort \in \Sorts$,
and
$Z_1,\dots, Z_\maa$ fresh meta-variables$\}$,
where
\begin{itemize}
\item $\etalong{s} = \abs{x_1 \dots x_\maa}{\apps{\halfetalong{s}}{
  (\etalong{x_1})}{(\etalong{x_\maa})}}$ if $s$ is an application or
  element of $\V \cup \F$, and $\etalong{s} = \halfetalong{s}$
  otherwise;
\item $\halfetalong{\afun} = \afun$ for $\afun \in \F$
  and $\halfetalong{x} = x$ for $x \in \V$, while
  $\halfetalong{\meta{Z}{s_1,\dots,s_\mac}} =
  \meta{Z}{\halfetalong{s_1},\dots,\halfetalong{s_\mac}}$ and
  $\halfetalong{(\abs{x}{s})} = \abs{x}{(\etalong{s})}$ and
  $\halfetalong{s_1\ s_2} = \halfetalong{s_1}\ (\etalong{s_2})$.
\end{itemize}
\end{definition}

Note that $\etalong{\ell}$ is a pattern for patterns $\ell$.  By
\cite[Thm.~2.16]{kop:12}, a relation $\arr{\Rules}$ is
terminating if $\arr{\Rules^\uparrow}$ is terminating.
However,
this transformation can introduce non-termination in some 
cases,
 e.g.,
 the terminating rule
 $\symb{f}\ X \arrz \symb{g}\ \symb{f}$ with $\symb{f} : \symb{o}
 \arrtype \symb{o}$ and $\symb{g} : (\symb{o} \arrtype \symb{o})
 \arrtype \symb{o}$, whose $\eta$-expansion
 $\symb{f}\ X \arrz \symb{g}\ (\abs{x}{(\symb{f}\ x)})$ is
 non-terminating.

\begin{example}\label{ex:derivextended}
The AFSM from \refEx{ex:derivdynamic} is $\eta$-expanded into an AFSM
with the single rule
$\deriv\ (\abs{x}{\symb{sin}\ \meta{F}{x}})\ Y \arrz
(\abs{y}{\symb{times}\ \linebreak
(\deriv\ (\abs{x}{\meta{F}{x}})\ y)\ (\symb{cos}\ \meta{F}{y})})\ 
Y$.
\end{example}

The original static approaches define the set of DPs as the set of
pairs $\ell^\sharp \arrdp p^\sharp$ where $\ell \arrz r$ is a rule
and $p$ a subterm of $r$ of the form $\apps{\afun}{r_1}{r_\maa}$ --
as their rules are built using terms, not meta-terms.  This can allow
variables which are bound in $r$ to become free in $p$.  In the
current setting, we use candidates rather than subterms, and we replace
such variables by meta-variables.
\newpage

\begin{definition}[$\SDP$]\label{def:sdp}
Let $s$ be a meta-term and $(\F,\Rules,\arity)$ an AFSM.
Let $\metafy(s)$ denote $s$ with all free variables replaced by
meta-variables.
Now $\SDP(\Rules) = \{ \ell^\sharp \arrdp \metafy(p^\sharp)\ 
(A) \mid \ell \arrz r \in \Rules \wedge p\ (A) \in \cand(r)
\wedge p$ has the form $\apps{\identifier{f}}{p_1}{p_\maa}$ with
$\maa = \arity(\identifier{f}) \}$.
\end{definition}

Thus, $\SDP(\Rules)$ is not collapsing,
and \ext{(for AFP $\Rules$)} both sides of a static DP have base type due
to the arity restriction.  This
simplifies reasoning and
make\ext{s} it possible to use
more powerful processors,
since cases (\ref{depchain:beta}) and (\ref{depchain:dp:meta})
in \refDef{def:chain}
no longer apply.
However, right-hand sides of static DPs may contain
meta-variables that do not occur on the left, which can be troublesome,
as we will see in \refEx{ex:staticbad}.


\begin{example}\label{ex:derivsdp}
For $\Rules$ as in \refEx{ex:derivextended}, the set $\SDP(\Rules)$
has one element:
$\deriv^\sharp\ (\abs{x}{\symb{sin}\ \meta{F}{x}})\ Y \arrdp
\deriv^\sharp\ (\abs{x}{\meta{F}{x}})\ Y\ (\emptyset)$.
\end{example}

\begin{example}\label{ex:ordrecstatic}
The AFSM from \refEx{ex:ordrec} is
AFP if a
sort ordering $\symb{ord} \grsort \symb{nat}$ is chosen (following
\refEx{ex:lim}).  $\SDP(\Rules)$ is given by:
\[
\begin{array}{rcll}
\symb{rec}^\sharp\ (\suc\ X)\ K\ F\ G & \arrdp & \symb{rec}^\sharp\ 
  X\ K\ F\ G\ (\emptyset) \\
\symb{rec}^\sharp\ (\symb{lim}\ H)\ K\ F\ G & \arrdp &
  \symb{rec}^\sharp\ (H\ M)\ K\ F\ G\ (\emptyset) \\
\end{array}
\]
This
AFSM is not PFP following \cite{suz:kus:bla:11}.
Note that
 the right-hand side of the second DP contains
a meta-variable that does not appear on the left.
As we will see in
\refEx{ex:ordrecdone}, that is not problematic
here.
\end{example}

As with dynamic DPs, static DPs are organised in a chain.

\edef\thmSChain{\number\value{theorem}}
\edef\thmSChainSec{\number\value{section}}
\newcommand{\sChainTheThm}{
\begin{theorem}\label{thm:schain}
If an AFSM $(\F,\Rules)$ is non-terminating and 
AFP,
then there is a minimal formative $(\SDP(\Rules),\Rules)$-dependency
chain.
\end{theorem}
}\sChainTheThm

\begin{proof}[Proof sketch]
Adaption of
 the proof in~\cite{kus:iso:sak:bla:09}
to the more permissive definition of
\emph{AFP}
over
\emph{PFP},
meta-variable conditions, and formative reductions
(see also \refThm{thm:chain} and \refThm{thm:sschain} below).
\end{proof}

This result transposes the work of~\cite{kus:iso:sak:bla:09} to
AFSMs and extends it by using a more liberal restriction,
by limiting interest to \emph{formative} chains, and by including
meta-variable conditions.
The relation with \cite{bla:06} is less clear: \refThm{thm:schain}
strictly extends its main result (both with formative chains and
meta-variable conditions, and by dropping the restriction that the
right-hand side of each DP has the same type as the left and does not
introduce fresh (meta-)\linebreak variables), but its admissibility restriction
does not require $\arity$ to be maximal.  The restriction
that DPs must be type-preserving
eliminates most systems that might be admissible but not AFP, however.


Note that the reverse result does \emph{not} hold, even with the
addition of meta-variable conditions:
one can have a minimal formative
$(\SDP(\Rules),\Rules)$-dependency chain
even for a terminating AFSM. 

\begin{example}\label{ex:staticbad}
Consider an AFSM
with $\F \supseteq \{
\nul,\one : \nat$,
$\symb{f} : \nat \arrtype \nat$,
$\symb{g} : (\nat \arrtype \nat)
\arrtype \nat\}$ and
$\Rules = \{ \symb{f}\ \nul \arrz \symb{g}\ 
(\abs{x}{\symb{f}\ x}),\linebreak 
\symb{g}\ (\abs{x}{
\meta{F}{x}}) \arrz
\meta{F}{\one} \}$.
It is AFP, with
$\SDP(\Rules) = \{ \symb{f}^\sharp\ \nul \arrdp \symb{g}^\sharp\ (\abs{x}{\symb{f}\ x})\linebreak
(\emptyset), \symb{f}^\sharp\ \nul \arrdp \symb{f}^\sharp\ X\
(\emptyset)\}$
.  Although
$\arr{\Rules}$
is terminating, there is a minimal $(\SDP(\Rules),\Rules)$-
chain $[(\identifier{f}^\sharp\ \nul \arrdp \identifier{f}^\sharp\ X,
\identifier{f}^\sharp\ \nul, \identifier{f}^\sharp\ \nul) \mid i \in
\N]$.
\end{example}


\smallskip
Using the computability inherent in the construction
of dependency chains using $\SDP$, we can strengthen the result of
\refThm{thm:schain}: rather than considering \emph{minimal} chains we
can require that (some of) the subterms of all $t_i$ are \emph{computable}:

\begin{definition}\label{def:Ccomputable}
Let $C_{\AlterRules}$ be an RC-set satisfying the properties of
\refThm{thm:defC} for a rewrite relation $\arr{\AlterRules}$.
A $(\P,\Rules)$-dependency chain $[(\rho_i,s_i,t_i) \mid i \in \N]$ is
\emph{$\AlterRules$-computable} if $\arr{\AlterRules} \mathop{\supseteq} \arr{\Rules}$,
for all $i \in \N$ there exists
a substitution
$\gamma_i$ such that $\rho_i = \ell_i
\arrdp p_i\ (A_i)$ with $s_i = \ell_i\gamma_i$ and $t_i = p_i\gamma_i$,
and $(\abs{x_1 \dots x_n}{v})\gamma_i$ is $C_{\AlterRules}$-computable
for all $v$ and $B$ such that $p_i \bsuptermeq{B} v$, $\gamma_i$ respects
$B$, and $\FV(v) = \{x_1,\dots,x_n\}$.
\end{definition}


\edef\thmSSChain{\number\value{theorem}}
\edef\thmSSChainSec{\number\value{section}}
\newcommand{\ssChainTheThm}{
\begin{theorem}\label{thm:sschain}
If an AFSM $(\F,\Rules)$ is non-terminating and 
AFP, then there is an $\Rules$-computable formative
$(\SDP(\Rules),\Rules)$-
chain.
\end{theorem}
}
\ssChainTheThm

\begin{proof}[Proof sketch]
The proof echoes the proof of \refThm{thm:chain}, but considers
\emph{minimal non-computable (MNC)} terms $\apps{\afun}{s_1}{s_\maa}$ (where
all $s_i$ are $C_S$-computable) rather
than minimal non-terminating terms.
We can avoid the $\mathtt{beta}$ step because MNC terms do not have the
right shape for headmost $\beta$-reductions.  By induction on the
definition of $\gracc$ we can show that if $\ell \arrz r$ is an AFP
rule and $\ell\gamma$ is a MNC term, then $\gamma(Z)$
is $C$-computable for all $Z \in \FMV(r)$.  Rather than a minimal
candidate with respect to non-termination, we select a
$\bsuptermeq{\beta}$-minimal candidate $p\ (A)$ such that $\delta$
respects $A$ and $p(\delta \cup \zeta)$ is non-computable for some
substitution $\zeta$ mapping $\FV(p)$ to $C$-computable terms.  As all
$\gamma(Z)$ are computable, $p$ has the right form $\apps{\afun}{
p_1}{p_\maa}$ to give $\ell^\sharp \arrdp p^\sharp\ (A) \in
\SDP(\Rules)$.  By minimality of the choice, the conditions for a
$\Rules$-\emph{computable} chain are satisfied.

The full proof is available in \refApp{app:sdp}.
\end{proof}

\ext{As} it is easily seen that all $C_{\AlterRules}$-computable terms are
$\arr{\AlterRules}$-termi\-nating and therefore $\arr{\Rules}$-terminating,
every $\AlterRules$-computable $(\P,\Rules)$-dependency chain is also
minimal.
\ext{The new flag does not give an inverse of \refThm{thm:sschain}, though:
  the chain in \refEx{ex:staticbad} \emph{is} $\Rules$-computable.}

\section{The higher-order  DP framework}
\label{sec:framework}

Extending an earlier methodology to reason about DPs \cite{kop:raa:12},
the higher-order DP framework follows the ideas of the first-order DP
framework~\cite{gie:thi:sch:05:2}: it is an extendable framework for
proving termination and non-termination,
which new termination methods
can easily be plugged into, in the form of \emph{processors}.


Thus far, we have reduced the problem of termination to the non-existence of
certain chains.  Following the first-order DP\linebreak framework,
we formalise this further in the notion of a \emph{DP problem}:

\begin{definition}[DP problem]\label{def:dpproblem}
A \emph{DP problem} is a tuple $(\P,\Rules,m,f)$ with $\P$ a set of
DPs, $\Rules$ a set of rules, $m \in \{ \minimal,
\arbitrary \} \cup \{ \static_\AlterRules \mid \text{any set of rules}\ 
\AlterRules \}$, and $f \in
\{ \formative, \nonformative \}$.%
\footnote{Our framework is implicitly parametrised by
  the signature $\F^\sharp$ used for term formation,
  the arity function $\arity$ (\refDef{def:arity}), and
  an $\arity$-preserving marking function $()^\sharp$ from
  $\F^\sharp$ to $\F^\sharp$
  following
  \refDef{def:marked} (this function
  is 
  the identity on symbols not in $\Defineds$).
  As
  none of the processors in this paper modify these components,
  we leave them implicit.}

A DP problem $(\P,\Rules,m,f)$ is called \emph{finite} if there exists no
infinite $(\P,\Rules)$-chain that is $\AlterRules$-computable if
$m = \static_\AlterRules$, \ext{is} minimal if $m = \minimal$,
and is formative if $f = \formative$.
It is \emph{infinite} if it is not finite or $\Rules$ is non-terminating.

To capture the different levels of permissiveness in the $m$ flag,
we use a transitive-reflexive relation $\succeq$ generated by
$\static_S\linebreak \succeq \minimal$ and $\minimal \succeq \arbitrary$.
\end{definition}

Thus, the combination of Theorems~\ref{thm:chain} and~\ref{thm:sschain}
can be rephra\-sed as: an AFSM $(\F,\Rules)$ is terminating if and only
if $(\DDP(\Rules),\Rules,\linebreak
\minimal,\formative)$ is finite, or if (but
not only if) it is 
AFP
and $(\SDP(\Rules),\Rules,
\static_\Rules,\formative)$ is finite.

\smallskip
The core idea of the DP framework is to iteratively simplify a set of
DP problems via \emph{processors} until nothing remains to be proved:

\begin{definition}[Processor]\label{def:proc}
A \emph{dependency pair processor} (or just \emph{processor}) is a
function that takes a DP problem and returns either \no\ or
a set of DP problems.
A processor $\Proc$ is \emph{sound} if a DP problem $\adpprob$ is finite
whenever $\Proc(\adpprob) \neq \no$ and all elements of $\Proc(\adpprob)$
are finite.
A processor $\Proc$ 
is \emph{complete} if a DP problem
$\adpprob$ is
infinite whenever $\Proc(\adpprob) = \no$ or contains an infinite
element.
\end{definition}

%


To prove finiteness of a DP problem $\adpprob$ with the DP framework,
we proceed analogously to the first-order DP framework
\cite{gie:thi:sch:fal:06}: we repeatedly apply sound DP processors
starting from $\adpprob$ until none remain.  That is, we execute the
following rough
\ext{procedure}:
  (1) let $A := \{ \adpprob \}$;
  (2) while $A \neq \emptyset$: select a problem $\bdpprob \in A$ and a
  sound processor $\Proc$ with $\Proc(A) \neq \no$\ext{,} and let $A := (A
  \setminus \{ \bdpprob \}) \cup \Proc(\bdpprob)$.
If this
\ext{procedure}
terminates, then $\adpprob$ is a finite DP problem.
%
To prove termination of an AFSM $(\F,\Rules)$,
we would use as initial DP problem
either $(\DDP(\Rules),\Rules,\minimal,\formative)$
(see \refThm{thm:chain}) or alternatively
$(\SDP(\Rules),\Rules,\static_{\Rules},\formative)$
(the latter only if $\Rules$ is 
AFP,
see
\refThm{thm:schain} and \refThm{thm:sschain};
here $\eta$-expansion following \refDef{def:etalong}
may be applied first).
A proof of its finiteness by the DP framework
then implies termination of $\Rules$.

Similarly, we can use the DP framework to prove infiniteness:
  (1) let $A := \{ \adpprob \}$;
  (2) while $A \neq \no$: select a problem $\bdpprob \in A$ and a complete
    processor $\Proc$, and let $A := \no$ if $\Proc(\bdpprob) = \no$, or
    $A := (A \setminus \{ \bdpprob \}) \cup \Proc(\bdpprob)$ otherwise.
For non-termination of  $(\F,\Rules)$\ext{,} the initial DP problem should be
$(\DDP(\Rules),\Rules,\minimal,\formative)$\linebreak (see \refThm{thm:chain}).
Note that the algorithms coincide
while
all processors are
sound \emph{and} complete.
In
a tool,
automation (or the user)
must resolve
the non-determinism
and select
suitable processors.

\smallskip
Below, we will present a number of processors within the framework.
We will typically present processors by writing
``for a DP problem $\adpprob$ satisfying $X$, $Y$, $Z$,
$\Proc(\adpprob) = \dots$''.
In these cases, we let
$\Proc(\adpprob) = \{ \adpprob \}$ for any problem $\adpprob$ not
satisfying the properties.
Many more processors are possible, but we have chosen to present a
selection which touches on all aspects of the DP framework:
\begin{itemize}
\item processors which map a DP problem to $\no$
  (\refThm{def:nontermproc}), a singleton set (most) and
  a non-singleton set (\refThm{def:depgraph});
\item
  manipulating
  meta-variable conditions
  (\refThm{def:depgraph}, \ref{def:addcond:mean});
\item changing not just the set $\P$, but also the set $\Rules$
  (\refThm{thm:formativeproc}, \refThm{thm:usable}) as well as
  the various flags (\refThm{thm:usable});
\item using specific values of the
  $f$
  (\refThm{thm:tagredpair}, \refThm{thm:formativeproc}) and
  $m$ flags (\refThm{thm:subtermproc}, \ref{thm:usable},
  and \refThm{thm:staticsubtermproc} for
  $m =
  \static_\AlterRules$);
\item using term orderings (\refThm{thm:basetriple},
  \ref{thm:tagredpair}), a key part of many termination
  proofs
  in the first- and higher-order settings.
\end{itemize}

All sound- and completeness claims are proved in
\refApp{app:processors}.

%
%
%
%

\subsection{The dependency graph}\label{subsec:graph}

We can leverage reachability information to \emph{decompose} DP
problems. In first-order rewriting, a graph structure
is used to track which DPs can possibly follow one another in a
chain~\cite{art:gie:00}. In our higher-order setting, we define this
\emph{dependency graph} as follows.

\begin{definition}[Dependency graph]\label{def:depgraph}
A DP problem $(\P,\Rules,m,f)$ induces a graph structure $\mathit{DG}$,
called its \emph{dependency graph}, whose nodes are the elements of
$\P$. There is a (directed) edge from
$\ell_1 \arrdp p_1\ (A_1)$ to $\ell_2 \arrdp p_2\ (A_2)$
in $\mathit{DG}$ iff
one of the following holds:
\begin{itemize}
\item $p_1$ has the form $\apps{\meta{Z}{s_1,\dots,s_\mac}}{t_1}{t_n}$
  (with $\mac,n \geq 0$)
\item $p_1$ has the form $\apps{\afun}{s_1}{s_n}$, $\ell_2$ has the
  form $\apps{\afun}{u_1}{u_n}$ (for the same $n$) and there exist
  substitutions $\gamma$ and $\delta$ that respect $A_1$ and $A_2$
  respectively and that map variables to variables, such that each
  $s_i\gamma \arrr{\Rules} u_i\delta$.
\end{itemize}
\end{definition}

\begin{example}\label{ex:graphconditions}
Consider an AFSM with $\Rules = \{ \symb{f}\ (\abs{x}{\meta{F}{x}}) \arrz
\meta{F}{\symb{f}\ (\abs{x}{\nul})} \}$ for $\symb{f} : (\nat \arrtype
\nat) \arrtype \nat$.  Let $\P := \DDP(\Rules) =$
\[
\left\{
\begin{array}{lrcll}
(1) & \symb{f}^\sharp\ (\abs{x}{\meta{F}{x}}) & \arrdp &
  \meta{F}{\symb{f}\ (\abs{x}{\nul})}\ & (\emptyset), \\
(2) & \symb{f}^\sharp\ (\abs{x}{\meta{F}{x}}) & \arrdp &
  \symb{f}^\sharp\ (\abs{x}{\nul}) & (\{ F : 1 \}) \\
\end{array}
\right\}
\]
The dependency graph of $(\P,\Rules,\minimal,\formative)$ is:
\begin{center}
\includegraphics[scale=0.4]{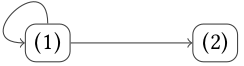}
\end{center}
There is no edge from (2) to itself because there is no
substitution $\gamma$ such that $(\abs{x}{\nul})\gamma$ can be reduced
to a term $(\abs{x}{\meta{F}{x}})\delta$ where $\delta(F)$ regards its
first argument (as $\arrr{\Rules}$ cannot introduce new variables).
\end{example}

In general, the dependency graph for a given DP problem is undecidable,
which is why we consider \emph{approximations}.

\begin{definition}[Dependency graph approximation \cite{kop:raa:12}]
A finite graph $G_\theta$ \emph{approximates} 
$\mathit{DG}$ if $\theta$ is a function that maps the nodes of
$\mathit{DG}$ to the nodes of $G_\theta$ such that, whenever there is
an edge from $\rho_1$ to $\rho_2$ in $\mathit{DG}$, there is an edge from
$\theta(\rho_1)$ to $\theta(\rho_2)$ in $G_\theta$. (There may be edges
in $G_\theta$ that have no corresponding edge in $\mathit{DG}$.)
\end{definition}

Note that this definition allows for an \emph{infinite} graph to be
approximated by a \emph{finite} one; infinite graphs may occur if
$\Rules$ is infinite\linebreak
(e.g., the union of all simply-typed instances of
polymorphic rules).

If $\P$ is finite, we can
take $G := G_{\mathtt{id}}$ with the same nodes as $\mathit{DG}$.
A simple approximation
may
have
an edge whenever $p_1$ is headed by a meta-variable 
or
$\head(p_1) = \head(\ell_2)$.
However,
one can
also take the meta-variable conditions
into account, as
we did
in \refEx{ex:graphconditions}.

\edef\procDPGraph{\number\value{theorem}}
\edef\procDPGraphSec{\number\value{section}}
\newcommand{\DPGraphTheProc}{
\begin{theorem}[Dependency graph processor]\label{def:graphproc}
The processor $\Proc_{G_\theta}$\ that maps a DP problem $\adpprob =
(\P,\Rules,m,f)$
to $\{ (\{ \rho \in \P \mid \theta(\rho) \in C_i \},\Rules,m,f) \mid
1 \leq i \leq n \}$ if $G_\theta$ is an approximation of the
dependency graph of $\adpprob$ and $C_1,\dots,C_n$
are the (nodes of the)
strongly connected components (SCCs) of $G_\theta$,
is both sound and complete.
\end{theorem}
}
\DPGraphTheProc

\begin{proof}[Proof sketch]
In an infinite $(\P,\Rules)$-chain $[(\rho_i,s_i,t_i) \mid i \in \N]$,
if $\rho_i$ and $\rho_{i+j}$ are both not $\mathtt{beta}$, then there is
a path from $\rho_i$ to $\rho_{i+j}$ in $DG$.  Since $G_\theta$ is
finite, every infinite path in $DG$ eventually remains in a cycle in
$G_\theta$. This cycle is part of an SCC.
\end{proof}


\begin{example}\label{ex:finishgraph}
Let $\Rules$ be the set of rules from \refEx{ex:graphconditions}
\ext{and} $G$ be the graph given there\ext{.}
\ext{Then} $\Proc_{G}(\DDP(\Rules),\Rules,
\minimal,\formative)\linebreak
= \{ (\{ \symb{f}^\sharp\ (\abs{x}{\meta{F}{x}}) 
\arrdp  \meta{F}{\symb{f}\ (\abs{x}{\nul})}\ (\emptyset) \},
\Rules,\minimal,\formative) \}$.
\end{example}

\begin{example}\label{ex:staticbad:dynamicgood}
The AFSM from \refEx{ex:staticbad} has dynamic DPs
(1) $\symb{f}^\sharp\ \nul\linebreak \arrdp \symb{g}^\sharp\ (\abs{x}{\symb{f}\ x})\ 
  (\emptyset)$,
(2) $\symb{f}^\sharp\ \nul \arrdp \symb{f}^\sharp\ x\ (\emptyset)$ and
(3) $\symb{g}^\sharp\ (\abs{x}{\meta{F}{x}}) \arrdp \meta{F}{\one}\ 
  (\emptyset)$.
Since variables cannot be reduced, we can choose an approximation
$G := G_{\mathtt{id}}$ where there is no outgoing edge from (2).  Thus,
$\Proc_G$ maps $(\DDP(\Rules),\Rules,m,f)$ to $\{\ (\{(1),(3)\},
\Rules,m,f)\ \}$.
\end{example}

\subsection{Processors based on reduction triples}\label{subsec:triple}

At the heart of most DP-based approaches to termination proving lie
well-founded orderings to prove that certain DPs (or
rules) can be used only finitely often
. 
For this, we use
\emph{reduction triples}~\cite{hir:mid:07,kop:raa:12}.

\begin{definition}[Reduction triple]\label{def:redpair}
A \emph{reduction triple} $(\rge,\pge,\pgt)$ consists of
two quasi-orderings $\rge$ and $\pge$ and a strict ordering $\pgt$
on meta-terms such that $\rge$ is monotonic,
all of $\rge,\pge,\pgt$ are meta-stable (that is, $\ell \rge r$
implies $\ell\gamma \rge r\gamma$ if $\ell$ is a closed pattern and
$\gamma$ a substitution on domain $\FMV(\ell) \cup \FMV(r)$, and the
same for $\pge$ and $\pgt$),
$\arr{\beta} \mathop{\subseteq} \rge$,
and both $\rge \circ \pgt \mathop{\subseteq} \pgt$
and $\pge \circ \pgt \mathop{\subseteq} \pgt$.
%
\end{definition}


In the first-order framework, the reduction pair processor
\cite{gie:thi:sch:05:2} seeks to orient all rules with $\rge$ and all
DPs with either $\rge$ or $\pgt$;
if this succeeds, those pairs oriented with $\pgt$ may be removed.

For us,
this is not ideal: the left- and right-hand side of
a 
DP
may have different types;
e.g.,
a DP $\symb{f}\ X \arrdp \symb{f}^\sharp\ (\emptyset)$
with $\symb{f}\ X : \nat$ and $\symb{f}^\sharp : \nat \arrtype \nat$.
Orderings like
HORPO
\cite{jou:rub:99}
or polynomial interpretations \cite{fuh:kop:12}
compare only (meta-)terms of the same type
(modulo renaming of sorts).
%
  They cannot deal well with fresh meta-variables or variables in the
  right-hand side either -- and while the former are a likely source of
  non-termination, the latter are
essentially harmless.  Thus, we
adapt \cite[Thm.~5.21]{kop:raa:12}
using alternative ordering
requirements to negate these problems.

\edef\procBasetypeRedpair{\number\value{theorem}}
\edef\procBasetypeRedpairSec{\number\value{section}}
\newcommand{\basetypeRedpairTheProc}{
\begin{theorem}[Reduction triple processor]\label{thm:basetriple}
Let $\mathsf{Bot}$ be a set $ \{ \bot_\atype : \atype \mid$ all types
$\atype \} \subseteq \F^\sharp$ of unused constructors,
$\adpprob = (\P_1 \uplus \P_2,\Rules,m,f)$ a DP problem and
$(\rge,\pge,\pgt)$ a reduction triple such that:
\begin{enumerate}
\item
  for all $\ell \arrdp p\ (A) \in \P_1 \uplus \P_2$ with $\ell :
  \atype_1 \arrtype \dots \arrtype \atype_\maa \arrtype \asort$
  and $p : \btype_1 \arrtype \dots \arrtype \btype_n \arrtype \bsort$
  we have:
  \begin{itemize}
  \item $\apps{\ell}{Z_1}{Z_\maa} \pgt \apps{p[\vec{x}:=\vec{\bot}]}{
    \bot_{\btype_1}}{\bot_{\btype_n}}$ if $\ell \arrdp p\ (A) \in \P_1$
  \item $\apps{\ell}{Z_1}{Z_\maa} \pge \apps{p[\vec{x}:=\vec{\bot}]}{
    \bot_{\btype_1}}{\bot_{\btype_n}}$ if $\ell \arrdp p\ (A) \in \P_2$
  \end{itemize}
  where $Z_1 : \atype_1,\dots,Z_\maa : \atype_\maa$ are fresh
  meta-variables and $[\vec{x}:=\vec{\bot}]$ is the substitution
  mapping all $x : \atype \in \FV(p)$ to $\bot_\atype$;
\item
  for all $\ell \arrz r \in \Rules$, we have $\ell \rge r$;
\item
  if $\P_1 \uplus \P_2$ contains a collapsing DP, then also:
  \begin{itemize}
  \item $\apps{a}{X_1}{X_\maa} \pge \apps{X_i}{\bot_{\btype_1}}{
    \bot_{\btype_n}}$ if $a : \atype_1 \arrtype \dots \arrtype
    \atype_\maa \arrtype \asort \in \M \cup \F^\sharp$,
    $1 \leq i \leq \maa$ and $X_i : \btype_1
    \arrtype \dots \arrtype \btype_n \arrtype \bsort \in \M$
    ($\asort,\bsort \in \Sorts$);
  \item $\apps{\afun}{X_1}{X_{\arity(\afun)}} \pge \apps{\afun^\sharp}{X_1}{
    X_{\arity(\afun)}}$ for all $\afun \in \F$.
  \end{itemize}
\end{enumerate}
Then the processor $\Proc_{(\rge,\pge,\pgt)}$ that maps $M$ to
$\{(\P_2,\Rules,m,f)\}$ is both sound and complete.
\end{theorem}
}
\basetypeRedpairTheProc

Here, the elements of $\mathsf{Bot}$ take the role of minimal terms for
the ordering.
We use them here to eliminate the free variables in
the right-hand sides of ordering requirements, which makes it easier to
apply traditional methods of generating a reduction triple.

While $\pgt$ and $\pge$ may still have to orient meta-terms of
distinct types,
these
are always \emph{base} types, which we
could
collapse
to a single sort.  The only relation
required to be monotonic, $\rge$, only has to
regard pairs of
meta-terms of the same type.
Also, while right-hand sides of ordering requirements
contain no
fresh variables, they may still contain fresh \emph{meta-}variables.
This is one of the weaknesses of static DPs; however,
we may
sometimes
be able to choose reduction triples that do not
regard such meta-variables.


\begin{example}\label{ex:disregard}
Suppose $\F \supseteq \{ \symb{a} : \nat,\ 
\afun^\sharp : \nat \arrtype \nat \arrtype \nat \}$, and
$\P = \{ \afun^\sharp\ \symb{a}\ X
\arrdp \afun^\sharp\ y\ Z\ (\emptyset) \}$.
If $\Rules$ is empty, we have a requirement:
$
\afun^\sharp\ \symb{a}\ X \pgt \afun^\sharp\ \bot_{\nat}\ Z$.
Since there is no collapsing DP, function symbols do not have to regard
all their arguments, and we can orient the requirement above by using a
polynomial interpretation \cite{pol:96,fuh:kop:12} $\mathcal{J}$ with
$\mathcal{J}({\symb{a}}) = 1$,\ 
$\mathcal{J}({\bot_{\nat}}) = 0$ and
$\mathcal{J}({\afun^\sharp}(n_1, n_2)) =
n_1$.
\end{example}

For collapsing $\P$, the last condition makes it hard to harness
the main strength of the first-order DP approach:
filtering arguments of function symbols, as we did in
\refEx{ex:disregard}.  This is one of the weaknesses of dynamic DPs,
which generate collapsing DPs.%

Fortunately, we can often weaken the requirement.
In~\cite{kop:raa:12}, this was
accomplished by considering \emph{local} AFSs, where reductions
inside an abstraction could
be postponed.  In the current setting, 
this restriction can
be generalised to the following condition:

\begin{definition}[Abstraction-simple]
A pair $(\P,\Rules)$ is \emph{abstraction-simple} if for all
left-hand sides $\ell$ of a rule in $\Rules$ or DP in $\P$:
\begin{itemize}
\item $\ell$ is a fully extended linear pattern;
\item meta-variables occur only
  as $\abs{x_1 \dots x_\mia}{
  \meta{Z}{x_1,\dots,x_\mia}}$ in $\ell$;
\item if some element of $\Rules$ does not have base type, then
  all meta-variables in $\ell$ have arity $\leq 1$.
\end{itemize}
\end{definition}

Note that the last requirement is always satisfied in AFSMs obtained
from AFSs, where all meta-variables have arity $0$, and in those
obtained from HRSs, where all rules have base type.

Our reduction triple processor for DP problems $(\P,\Rules,m,f)$\linebreak
with abstraction-simple
$(\P,\Rules)$ distinguishes function symbol occurrences inside
abstractions by using a $\ttag$ function.

\begin{definition}[$\ttag$]\label{def:ttag}
For a set of DPs $\P$ and a set of rules $\Rules$,
let $\funs(\P,\Rules)$ be the set of all $\afun \in \F^\sharp$
occurring in $\P$ or $\Rules$. Then let $\funs^-(\P,\Rules)$ be a
subset of $\F$ disjoint of $\funs(\P,\Rules)$ and $\mathsf{Bot}$
that contains, for
every $\afun : \atype \in \funs(\P,\Rules)$, a symbol $\afun^- :
\atype$ of the same arity.
For any arity-respecting closed meta-term $s$ over $\funs(\P,\Rules)$,
let $\ttag(s)$ denote $s$ with all sub-expressions
$\apps{\afun}{s_1}{s_\mia}$ with $\mia = \arity(\afun)$ and
$\FV(\apps{\afun}{s_1}{s_\mia}) \neq \emptyset$ replaced
by $\apps{\afun^-}{s_1}{s_\mia}$.
\end{definition}

The $\ttag$ function adds a special mark to any function
symbol
between a lambda-abstraction $\lambda x$ and an
occurrence of the bound variable $x$.
Thus, $\ttag(\abs{x}{\afun\ x}) = \abs{x}{\afun^-\ x}$ and
$\ttag(\abs{x}{\afun\ \nul}) = \abs{x}{\afun\ \nul}$ (if $\arity(\afun)
= 1$).
While \cite{kop:raa:12} uses the $\ttag$ function as part
of the definition of dependency chain,
the DP framework
confines it to processors that actually use it,
by harnessing the $\formative$ flag:

\edef\procAbsSimpleRedpair{\number\value{theorem}}
\edef\procAbsSimpleRedpairSec{\number\value{section}}
\newcommand{\absSimpleRedpairTheProc}{
\begin{theorem}[Abstraction-simple reduction triple processor]\label{thm:tagredpair}
Let $\mathsf{Bot}$ be a set $ \{ \bot_\atype : \atype \mid$ all types
$\atype \} \subseteq \F^\sharp$ of unused
constructors.
Let $\adpprob = (\ext{\P},\Rules,m,f)$ be a DP problem
\ext{with $\P = \P_1 \uplus \P_2$}.
Let $(\rge,\pge,\pgt)$ be a reduction triple such that:
\begin{enumerate}
\item $(\ext{\P},\Rules)$ is abstraction-simple and $f = \formative$;
\item
  for all $\ell \arrdp p\ (A) \in \ext{\P}$ with $\ell :
  \atype_1 \arrtype \dots \arrtype \atype_\maa \arrtype \asort$ and
  $p : \btype_1 \arrtype \dots \arrtype \btype_n \arrtype \bsort$ we
  have:
  \begin{itemize}
  \item $\apps{\ell}{Z_1}{Z_\maa}\! \pgt\! \ttag(\apps{p[\vec{x}:=
    \vec{\bot}]}{\bot_{\btype_1}}{\bot_{\btype_n}})$ if
    $\ell\! \arrdp\! p\ (A)\! \in \P_1$
  \item $\apps{\ell}{Z_1}{Z_\maa}\! \pge\! \ttag(\apps{p[\vec{x}:=
    \vec{\bot}]}{\bot_{\btype_1}}{\bot_{\btype_n}})$ if
    $\ell\! \arrdp\! p\ (A)\! \in \P_2$
  \end{itemize}
  where $Z_1 : \atype_1,\dots,Z_\maa : \atype_\maa$ are fresh
  meta-variables and $[\vec{x}:=\vec{\bot}]$ is the substitution
  mapping all $x : \atype \in \FV(p)$ to $\bot_\atype$;
\item\label{it:redpair:rules}
  for all $\ell \arrz r \in \Rules$, we have $\ell \rge \ttag(r)$;
\item
  $\apps{\afun^-}{X_1}{X_{\arity(\afun)}} \rge \apps{\afun}{X_1}{X_{\arity(\afun)}}$ for all
  $\afun \in \funs(\P,\Rules)$;
\item\
  if $\ext{\P}$ contains a collapsing DP, then also:
  \begin{itemize}
  \item $\apps{a}{X_1}{X_\maa} \pge \apps{X_i}{\bot_{\btype_1}}{
    \bot_{\btype_n}}$ if $a : \atype_1 \arrtype \dots \arrtype
    \atype_\maa \arrtype \asort \in \M \cup \funs^-(\P,
    \Rules)$, $1 \leq i \leq \maa$ and
    $X_i : \btype_1 \arrtype \dots \arrtype \btype_n
    \arrtype \bsort \in \M$ ($\asort,\bsort \in \Sorts$);
  \item $\apps{\afun^-}{X_1}{X_{\arity(\afun)}} \pge \apps{\afun^\sharp}{X_1}{
    X_{\arity(\afun)}}$ for all $\afun \in \funs(\P,\Rules)$.
  \end{itemize}
\end{enumerate}
Then the processor $\Proc_{\ttag(\rge,\pge,\pgt)}$ that maps $M$ to
$\{ (\P_2,\Rules,m,f) \}$ is both sound and complete.
\end{theorem}
}
\absSimpleRedpairTheProc

\begin{proof}[Proof sketch]
If $s \arrr{\Rules} \ell\gamma$ by an $\ell$-formative reduction, then
$\ttag(s)
\arrr{\Rules} \ell[\vec{Z}:=\ttag(\gamma(\vec{Z}))]$.  Thus, an
infinite formative $(\P,\Rules)$-chain induces an infinite $(\rge \cup
\pge \cup \pgt)$-reduction, with every DP in $\P_1$ corresponding to a
$\pgt$ step.  
Because
$x \in \FV(q)$ in
case (\ref{depchain:beta}) of \refDef{def:chain} and 
$\FV(v) \cap
\{\vec{x}\} \neq \emptyset$ in case (\ref{depchain:dp:meta}), we do not
have to pass through untagged symbols when reducing to a subterm.
\end{proof}


\begin{example}\label{ex:staticbad:finish}
To remove both pairs in the remaining DP problem
from \refEx{ex:staticbad:dynamicgood} \ext{by}
\refThm{thm:tagredpair}\ext{,} we must satisfy the requirements:
\[
\begin{array}{rcl|rcl|rcl}
\symb{f}^\sharp\ \nul & \!\!\!\pgt\!\!\! & \symb{g}^\sharp\ 
  (\abs{x}{\symb{f}^-\ x})\! &
\symb{g}^\sharp\ (\abs{x}{\meta{F}{x}}) & \!\!\!\pgt\!\!\! & 
  \meta{F}{\one}\! &
\symb{f}^-\ X & \!\!\!\pge\!\!\! & \symb{f}^\sharp\ X \\
\symb{f}\ \nul & \!\!\!\rge\!\!\! & \symb{g}\ (\abs{x}{\symb{f}^-\ x})\! &
\symb{g}\ (\abs{x}{\meta{F}{x}}) & \!\!\!\rge\!\!\! & \meta{F}{\one}\! &
\symb{f}^-\ X & \!\!\!\rge\!\!\! & \symb{f}\ X \\
\end{array}
\]
And both $\symb{f}^-\ X \pge X$ and always $\apps{Z}{X_1}{X_\maa} \pge
X_i\ \vec{\bot}$.

We can soundly extend higher-order polynomial interpretations
\ext{\cite{pol:96,fuh:kop:12}} to
meta-variables with arguments
by
$\llbracket \meta{Z}{s_1,\dots,
s_\mac} \rrbracket_{\mathcal{J},\alpha} = \alpha(Z)(\llbracket
s_1\rrbracket_{\mathcal{J},\alpha},\dots,
\llbracket
s_\mac\rrbracket_{\mathcal{J},\alpha})$\ext{, as done in \cite[Chapter 4]{kop:12}}; the
  \ext{given}
ordering
requirements are
satisfied by
  \ext{taking}
$\mathcal{J}(@^{\atype \arrtype \btype}(f, x)) =
\max(f(x),
x(\vec{0}))$ and $\mathcal{J}(\bot_{\vec{\atype} \arrtype
\asort}(\vec{x})) = 0$.
\end{example}

\subsection{From dynamic to static DPs}\label{subsec:staticproc}

While \refThm{thm:sschain} does not yield an equivalence result, it is
observed in~\cite{kop:raa:12} that the static approach \emph{is}
complete if $\SDP(\Rules) \subseteq \DDP(\Rules)$, as then every
$(\SDP(\Rules),\Rules)$-chain is also a $(\DDP(\Rules),\Rules)$-chain.
By using the dependency graph we can even go beyond this.

\edef\procToSDP{\number\value{theorem}}
\edef\procToSDPSec{\number\value{section}}
\newcommand{\toSDPTheProc}{
\begin{theorem}\label{thm:sdpproc}
Let $G := G_\theta$ be a dependency graph approximation for
$\SDP(\Rules)$, and let $\SDP(\Rules)_G := \{ \rho \in \SDP(\Rules)
\mid \theta(\rho)$ is on a cycle in $G \}$.
Then the processor $\Proc_{\SDP_G}$ that maps a DP problem
$(\P,\Rules,m,f)$ to $\{(\SDP(\Rules)_G,\Rules,
\static_\Rules,\formative)\}$ if
$\P \subseteq \DDP(\Rules)$ and $\Rules$ is AFP, is sound;
it is complete if also $\SDP(\Rules)_G \subseteq \P$.
\end{theorem}
}
\toSDPTheProc

\begin{proof}[Proof sketch]
Soundness holds by a combination of \refThm{thm:sschain} and
\ref{def:graphproc}, completeness
since the processor only
removes DPs.
\end{proof}

\refThm{thm:sdpproc} can be applied at any time in the framework.
Although a typical application would be to do this at the start, doing
it later might be useful if we can first apply a processor to, for
instance, remove some non-AFP rules.
This gives us ``the best of both worlds'': now reasoning
based on dynamic \emph{and} static DPs can be combined
within the \emph{same} termination proof.

\begin{example}\label{ex:staticgraph}
Let $\Rules$ consist of the rules for $\symb{map}$ from
\refEx{ex:mapintro} along with $\symb{f}\ L \arrz \symb{map}\ (\abs{x}{
\symb{g}\ x})\ L$.  Then $\DDP(\Rules)$ consists
of:
\[
\begin{array}{lrcll}
(1)\!\! & \symb{map}^\sharp\ (\abs{x}{\meta{Z}{x}})\ (\symb{cons}\ H\ T) &
\arrdp & \meta{Z}{H} & (\emptyset) \\
(2)\!\! & \symb{map}^\sharp\ (\abs{x}{\meta{Z}{x}})\ (\symb{cons}\ H\ T) &
\arrdp & \symb{map}^\sharp\ (\abs{x}{\meta{Z}{x}})\ T & (\emptyset) \\
(3)\!\! & \symb{f}^\sharp\ L & \arrdp & \symb{map}^\sharp\ (\abs{x}{\symb{g}\ x})\ L & (\emptyset) \\
(4)\!\! & \symb{f}^\sharp\ L & \arrdp & \symb{g}^\sharp\ x & (\emptyset) \\
\end{array}
\]
And $\SDP(\Rules) = \{ (2),(3),(4') \}$, where $(4')$ is
$\symb{f}^\sharp\ L \arrdp \symb{g}^\sharp\ X\ (\emptyset)$.
  %
  %
  %
  We can clearly choose a graph approximation $G_{\mathtt{id}}$ where
  $(3)$ and $(4')$ have no incoming edges.
  Thus, only $(2')$ is on a cycle;\linebreak
$\Proc_{\SDP_{G_\theta}}(
\DDP(\Rules),\Rules,\minimal,\formative) =
\{\ (\{(2)\},\Rules,\linebreak
\mathtt{computable}_\Rules,
\formative)\ \}$,
and 
we have not lost completeness.
\end{example}

  %
  %
  \edef\procUseless{\number\value{theorem}}
  \edef\procUselessSec{\number\value{section}}
  \newcommand{\removeUselessTheProc}{
  \begin{theorem}\label{thm:removeuseless}
  The processor $\Proc_{\mathtt{useless}}$ that maps a DP problem
  $(\P_1 \uplus \P_2,\Rules,m,f)$ with $m \succeq \minimal$ and
  $\P_1 \uplus \P_2 \subseteq \DDP(\Rules)$
  to $\{ (\P_2,\Rules,m,f) \}$ if for all DPs $\rho\in\P_1$ there is $\ell^\sharp \arrdp p^\sharp\ (A) \in \P_2$ such that
  $\rho$ has the form $(\apps{\ell}{Z_1}{Z_i})^\sharp \arrdp p^\sharp\ (A)$,
  is
  sound and complete.
  \end{theorem}
  }

\subsection{Modifying collapsing dependency pairs}\label{subsec:modify}

The following two processors, aimed at collapsing DPs, have no
counterpart in any first-order framework.

We start with a simple transformation of collapsing DPs that is
particularly relevant for \emph{AFSs}, where meta-variables have
arity $0$:

\edef\procExtend{\number\value{theorem}}
\edef\procExtendSec{\number\value{section}}
\newcommand{\extendTheProc}{
\begin{theorem}[Extended meta-application processor]
Let us define
$\mathtt{extend}(s) := \meta{Z}{t_1,\dots,t_\mac,s_1,\dots,s_n}$ if
$s = \apps{\meta{Z}{t_1,\dots,t_\mac}}{\linebreak s_1}{s_n}$ and
$\mathtt{extend}(s) := s$ if $s$ has any other form.
The processor $\Proc_{\mathtt{extend}}$ that maps a DP problem
$(\P,\Rules,m,f)$ to the singleton set
$\{ (\P', \Rules, m, f) \}$ with $\P' = \{ \ell \arrdp
\mathtt{extend}(p)\ (A) \mid \ell \arrdp p\ (A) \in \P \}$,
is sound.
It is complete if $\P \subseteq \DDP(\Rules)$.
\end{theorem}
}
\extendTheProc

\begin{proof}[Proof sketch]
By the definition of chains, for every occurrence of $\ell \arrdp
\apps{\meta{Z}{t_1,\dots,t_\mac}}{s_1}{s_n}\ (A)$ in a $(\P,
\Rules)$-chain, if $\gamma(Z) =
\linebreak
\abs{x_1 \dots x_i}{u}$ with $\mac
< i \leq n$, then the occurrence is followed by $i-\mac$ steps with
\texttt{beta}; we might as well do the $\beta$-reduction directly.
Completeness follows because,
when $\P \subseteq \DDP(\Rules)$, the existence of an infinite
$(\P',\Rules)$-chain implies non-termination of $\arr{\Rules}$.
\end{proof}

The advantage of extending meta-variable applications is twofold.
First, it might enable other transformations, such as the one we
will discuss next.  Second, it makes some reduction triples easier to
apply because the right-hand side is made smaller
(see \refSec{subsec:triple}).

\begin{example}\label{ex:extend}
Consider the AFSM with $\F \supseteq \{ \symb{f} : \nat \arrtype \nat,\ 
\symb{g} \linebreak
: (\nat \arrtype \nat) \arrtype \nat \arrtype \nat \}$ and
$\Rules = \{ \symb{f}\ \nul \arrz \symb{g}\ (\abs{x}{\nul})\ \one,\ 
\symb{g}\ F\ X \arrz F\ (\symb{f}\ X) \}$.
Then $\DDP(\Rules)$ consists of
(1) $\symb{f}^\sharp\ \nul \arrdp \symb{g}^\sharp\ (\abs{x}{\nul})\ \one\ 
  (\emptyset)$,
(2) $\symb{g}^\sharp\ F\ X \arrdp F\ (\symb{f}\ X)\ (\emptyset)$ and
(3) $\symb{g}^\sharp\ F\ X \arrdp \symb{f}^\sharp\ X\ (\emptyset)$.
We have $\Proc_{\mathtt{extend}}(\DDP(\Rules),\Rules,\minimal,
\formative) = \{ (\P,\Rules, \minimal,\linebreak
\formative) \}$ where $\P = \{ (1), (2')\ \symb{g}^\sharp\ F\ X \arrdp
\meta{F}{\symb{f}\ X}\ (\emptyset), (3) \}$.

If we use reduction triples with an extension of polynomial interpretations
to AFSMs \ext{(see \refEx{ex:staticbad:finish})},
the DP $(2)$
gives a requirement $\llbracket
\symb{g}^\sharp\ F\ X\rrbracket_{\mathcal{J},\alpha}\ \pge \max(
\alpha(F)(\llbracket \symb{f}\ X \rrbracket_{\mathcal{J},\alpha}),\ 
\llbracket \symb{f}\ X \rrbracket_{\mathcal{J},\alpha})$.
The DP $(2')$ would instead generate $\llbracket
\symb{g}^\sharp\ F\ X\rrbracket_{\mathcal{J},\alpha}\ \pge
\alpha(F)(\llbracket \symb{f}\ X \rrbracket_{\mathcal{J},\alpha})$.
\end{example}

\begin{example}
To see why the processor is not necessarily complete if $\P \not\subseteq
\DDP(\Rules)$, consider a DP problem $(\P,\emptyset,m,f)$ with $\P =
\{ \text{(1)}\ \symb{f}^\sharp\ F\ X \arrdp F\ \symb{a}\ X\ (\emptyset),
\ \text{(2)}\ \symb{g}^\sharp\ \symb{a} \arrdp \symb{f}^\sharp\ 
(\abs{xy}{\symb{g}\ x})\ \symb{a}\ (\emptyset) \}$.  This DP problem is
\emph{finite}: in an infinite chain $[(\rho_i,s_i,t_i) \mid i \in \N]$,
if $s_i = \symb{g}^{\sharp}\ \symb{a}$ then $s_{i+3} =
\app{(\abs{y}{\symb{g}\ \symb{a}})}{\symb{a}}$ and neither case of
\refDef{def:chain}(\ref{depchain:beta}) applies.  However, the DP problem
$(\{ (1')\ \symb{f}^\sharp\ F\ X \arrdp \meta{F}{\symb{a},X}\ 
(\emptyset),\ \text{(2)} \},\Rules,m,f)$ is infinite, as demonstrated by
the infinite dependency chain with $\rho_i = (1')$ for even $i$
and $\rho_i = \text{(2)}$ for odd $i$.
\end{example}

In \refSec{subsec:graph} we have seen how we can utilise
meta-variable conditions.
Our next
processor seeks to
enable this by adding conditions.

\edef\procAddCond{\number\value{theorem}}
\edef\procAddCondSec{\number\value{section}}
\newcommand{\addCondTheProc}{
\begin{theorem}[Condition-adding processor]\label{def:addcond:mean}
A processor\linebreak $\Proc_{\mathtt{addcond}}$ that maps a DP problem
$(\P,\Rules,m,f)$ with $m \succeq \minimal$ to $\{ (\P',\Rules,m,
f) \}$ if
the following conditions are satisfied is both\linebreak
sound and complete:
\begin{enumerate}
\item\label{thm:addcond:mean}
  for any term $s$ that is terminating under $\arr{\Rules}$: there is
  no minimal $(\P,\Rules)$-chain that starts in $s^\sharp$ or any of
  its subterms;
\item $\P = \P_1 \uplus \P_2$ where $\P_2$ contains only dependency
  pairs of the form $\ell \arrdp \meta{Z}{p_1,\dots,p_\mac}\ (A)$
  with $Z \in \FMV(\ell)$;
\item $\P' = \P_1 \cup \{ \ell \arrdp \meta{Z}{p_1,\dots,p_\mac}\ 
  (A \cup \{ Z : i \}) \mid \ell \arrdp \meta{Z}{p_1,\dots,\linebreak
  p_\mac}\ (A) \in \P_2 \wedge 1 \leq i \leq \mac \}$.
\end{enumerate}
\end{theorem}
}
\addCondTheProc

\begin{proof}[Proof sketch]
If $\gamma(Z)$ disregards its first $\mac$ arguments,
yet \linebreak
$\meta{Z}{p_1,\dots,p_\mac}\gamma$
\ext{starts} an infinite
chain, then $\gamma(Z)$ is non-termina\-ting by condition
(\ref{thm:addcond:mean}), contradicting minimality.
For completeness note that
a substitution that respects $A \cup \{
Z : i \}$ also respects $A$.
\end{proof}

\ext{Intuitively, if a DP collapses to
$\meta{Z}{p_1,\dots,p_\mac}$,
then $\gamma(Z)$ must regard
\emph{some} 
$p_i$ in
a minimal chain.}
Requirement (\ref{thm:addcond:mean}) is satisfied if $\P \subseteq
\DDP(\Rules)$, but also if $\P \subseteq \DDP(\Rules)'$ where
$\DDP(\Rules)'$ is obtained by applying the extended meta-variable
processor to $\DDP(\Rules)$.

\begin{example}\label{ex:extendadd}
We continue
\refEx{ex:extend}.
The condition-adding processor with
$\P_2 = \{ (2) \}$
maps $(\P,\Rules,\minimal,\formative)$ to $\{ (
\P',\linebreak
\Rules, \minimal,\formative) \}$ where $\P' = \{ (1),
(2'')\ \symb{g}^\sharp\ F\ X \arrdp \meta{F}{\symb{f}\ X}\ 
\linebreak
(\{ F : 1 \}), (3) \}$.
The dependency graph $G$ for the remaining problem,
\begin{center}
\includegraphics[scale=0.4]{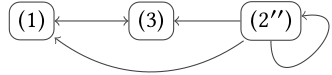}
%
\end{center}
does \emph{not} have an edge from $(1)$ to $(2'')$ due to the condition.
Thus, $\Proc_G(\P',\Rules,\minimal,\formative) = \{ \: (\{(1),(3)\},
\Rules,\minimal,\linebreak
\formative),\: (\{ (2'') \}, \Rules, \minimal,
\formative) \: \}$.
\end{example}

If duplication of dependency pairs is undesirable, we can take for
$\P_2$ the set of DPs in $\P$ whose right-hand side has the form
$\meta{Z}{s}$.
Then, the condition-adding processor merely adds a condition.

With the processors in this paper, there is no benefit in applying the
processors of 
  \ext{\refSec{subsec:modify}}
more than once.  Thus, it is a reasonable strategy to apply these
processors at the start of the algorithm for the \ext{framework,}
and then \ext{ignore} them in the rest of the process.



\subsection{Rule removal without search for orderings}\label{subsec:ruleremove}

While processors often simplify only $\P$,
they can also simplify $\Rules$.
One such processor uses the notion of \emph{formative rules}: the rules
that suffice for formative reductions.
As in the first-order case \cite{fuh:kop:14}, we use a semantic
characterisation of formative rules. In practice, we
then work with over-approximations of this characterisation,
analogous to the 
dependency graph approximations in
\refThm{def:graphproc}.

\begin{definition}
\label{def:formativerules}
A function $\FR$ that maps a pattern $\ell$ and a set of rules $\Rules$
to a set $\FR(\ell,\Rules) \subseteq \Rules$ is a \emph{formative
  rules approximation} if for all $s$ and $\gamma$: if
$s \arrr{\Rules} \ell\gamma$ by an $\ell$-formative reduction,
then this reduction can be done using only rules in $\FR(\ell,\Rules)$.
We
write
$\FR(\P,\Rules) = \bigcup \{ \FR(\ell_i,
\Rules) \mid \apps{\afun}{\ell_1}{\ell_n} \arrdp p\, (A)\: \in\, \P
\wedge 1 \leq i \leq n \}$.
\end{definition}

A formative rules approximation for AFSs is provided in
\cite[Def.~6.10]{kop:raa:12}; an approximation for AFSMs is
available in \refApp{app:processors}.
The following result follows trivially from \refDef{def:formativerules}:


\edef\procFormative{\number\value{theorem}}
\edef\procFormativeSec{\number\value{section}}
\newcommand{\formativeTheProc}{
\begin{theorem}[Formative rules processor]
\label{thm:formativeproc}
For a formative rules approximation $\FR$, the processor $\Proc_{\FR}$
that maps a DP problem $(\P,\Rules,m,\formative)$ to $\{ (\P,
\FR(\P,
\Rules),m,\formative) \}$ is both
sound and complete.
\end{theorem}
}
\formativeTheProc


\begin{example}\label{ex:extendshorten}
Let $\FR$ be the trivial formative rules approximation
:
$\FR(Z,\Rules) = \emptyset$ for $Z \in \M$ and $\FR(t,\Rules) =
\Rules$ otherwise.
Continuing \refEx{ex:extendadd}, we have
$\Proc_\FR(\{ (2'') \},\Rules,
\minimal,\formative)\linebreak =
\{ (\{ (2'') \}, \emptyset, \minimal,\formative) \}$.
This problem is 
mapped
to $\{ (\emptyset,\emptyset,\minimal,\formative) \}$ by the reduction
triple processor, 
and this in turn
is mapped to $\emptyset$ by the dependency graph processor.
\end{example}

\smallskip
Where \emph{formative} rules are generated from the left-hand sides of
DPs and rules, \emph{usable} rules are generated from the right.
Origina\-ting in the first-order setting, this processor can eliminate
potentially many rules at once.  However, it is only applicable to
non-collapsing $\P$.  It also imposes a heavy price on the flags.

\begin{definition}
A function $\UR$ that takes a pair $(\P,\Rules)$ with $\P$
non-collapsing and returns a set of rules is a \emph{usable
rules approxi\-mation} if 
a function $\varphi$ from
terminating terms to terms 
exists
with:
\begin{itemize}
\item for all meta-terms $s$ that occur as the direct argument of a
  left- or right-hand side in $\P$, and all substitutions $\gamma$
  such that $s\gamma$ is a terminating term:
  $\varphi(s\gamma) = s\gamma^\varphi$;
  here, $\gamma^\varphi$ is the substitution mapping each $x \in
  \domain(\gamma)$ to $\varphi(\ext{\gamma(}x\ext{)})$
\item if $s \arrr{\Rules} t$ and $s$ is terminating, then $\varphi(s)
  \arrr{\UR(\P,\Rules)} \varphi(t)$.
\end{itemize}
\end{definition}

\edef\procUsable{\number\value{theorem}}
\edef\procUsableSec{\number\value{section}}
\newcommand{\usableTheProc}{
\begin{theorem}[Usable rules processor]
\label{thm:usable}
For a usable rules approximation $\UR$, the processor $\Proc_\UR$
that maps a DP problem $(\P,\Rules,m,f)$ with $\P$ 
non-collapsing 
and $m \succeq \minimal$ to $\{ (\P,
\linebreak
\UR(\P,\Rules),
\arbitrary,\nonformative)\}$ is sound.
\end{theorem}
}
\usableTheProc

\begin{proof}[Proof sketch]
Since $\P$ is non-collapsing, we can assume that a
$(\P,\Rules)$-chain
$[(\rho_i,s_i,t_i) \mid i \in \N]$ has no $\mathtt{beta}$ entries;
thus,
with $\varphi'(\apps{\afun}{s_1
}{s_n}) := \apps{\afun}{
\varphi(s_1)}{\varphi(s_n)}$, the sequence $[(\rho_i,\varphi'(s_i),
\varphi'(t_i)) \linebreak
\mid i \in \N]$ is well-defined and an infinite
$(\P,\UR(\P,\Rules))$-chain.
\end{proof}


This semantic notion of usable rules approximation differs wild\-ly
from first-order counterparts (and the higher-order definition in
\cite{suz:kus:bla:11}) by leaving the choice of $\UR$ completely
free; it is \emph{not}, for instance, required that $\UR(\P,\Rules)
\subseteq \Rules$.
This also differs from the semantic notion of usable rules in
\cite{thi:07}, which however considers the classical notion of usable
rules defined for \emph{innermost} termination.
Usable rules for full termination, called ``needed rules'' in
\cite{thi:07}, with a transformation $\varphi$ are still defined
syntactically in \cite{thi:07}.

An example approximation is the union of
usable rules following \cite{suz:kus:bla:11} (adapted to AFSMs) with
$\Ce := \{ \symb{p}_\atype\ X\ Y \arrz X,\ \symb{p}_\atype\ X\ Y
\arrz Y \mid$ any type $\atype \}$, where $\symb{p}_\atype : \atype
\arrtype \atype \arrtype \atype$.
However, this
leaves \emph{all} rules usable if any DP or
usable rule of $\P$ is collapsing -- which is
common,
even when $\P \subseteq \SDP(\Rules)$.
It would be worthwhile to research alternative approximations
not subject to this weakness.

\smallskip
As in the first-order case, we can use usable rules
inside
reduction triple
processors 
without losing the
$\minimal$ and $\formative$ flags.
For example, both in \refThm{thm:basetriple} and
\refThm{thm:tagredpair} we could
use
$\UR(\P,\Rules)$ in
case (\ref{it:redpair:rules}).
%
%
Even stronger results could likely be obtained by considering formative and
usable rules with respect to an argument filtering
\cite{fuh:kop:14,gie:thi:sch:fal:06}
(which is beyond the scope of this paper).



\subsection{Subterm criterion processors}

Reduction triple processors are powerful, but
they exert a computational
price:
we must orient all rules in $\Rules$.  The subterm criterion processor
allows us to remove DPs without considering
$\Rules$ at all.
It is based on a \emph{projection function} \cite{hir:mid:07}.
In our higher-order setting:

\begin{definition}
For $\P$ non-collapsing, let
$\mathtt{heads}(\P)$ be the set of all symbols $\afun$ that occur as
the head of a left- or right-hand side of a DP in $\P$.
A \emph{projection function} for $\P$ is a function $\nu :
\mathtt{heads}(\P) \to \N$ such that for all dependency pairs
$\ell \arrdp p\ (A) \in \P$, the function $\project$ with
$\project(\apps{\afun}{s_1}{s_n}) = s_{\nu(\afun)}$ is well-defined
both for $\ell$ and for $p$.
\end{definition}

Note the limitation to non-collapsing dependency pairs, which
is essential for the subterm criterion
\cite{hir:mid:07,kop:raa:12,kus:iso:sak:bla:09,suz:kus:bla:11} to work.

\pagebreak

\edef\procSubtermCriterion{\number\value{theorem}}
\edef\procSubtermCriterionSec{\number\value{section}}
\newcommand{\subtermCriterionTheProc}{
\begin{theorem}[Subterm criterion processor]\label{thm:subtermproc}
The processor
$\Proc_{\mathtt{subcrit}}$ that maps a DP problem
$(\P,\Rules,m,f)$ with $m \succeq \minimal$ to $\{(\P_2,\Rules,
m,f)\}$
if the following holds is both sound and complete:
\begin{itemize}
\item $\P = \P_1 \uplus \P_2$ is non-collapsing;
\item a projection function $\nu$ exists such that
  $\project(\ell) \supterm \project(p)$ for $\ell \arrdp p\ (A) \in
  \P_1$ and
  $\project(\ell) = \project(p)$ for $\ell \arrdp p\ (A) \in \P_2$.
\end{itemize}
\end{theorem}
}
\subtermCriterionTheProc

\begin{proof}[Proof sketch]
If the given conditions are satisfied, every infinite $(\P,
\Rules)$-chain induces an infinite $\suptermeq \mathop{\cdot} \arrr{\beta}$
sequence which starts in a strict subterm of $t_1$, contradicting
minimality unless all but finitely many steps are equality.  Since
every occurrence of a pair in $\P_1$ results in a strict $\supterm$
step, a tail of the chain lies in $\P_2$.
\end{proof}

\begin{example}\label{ex:mapfinish}
Using 
$\nu(\symb{map}^\sharp) = 2$,
$\Proc_{\mathtt{subcrit}}$ maps the DP
\linebreak
problem $(\{(2)\},\Rules,
\static_\Rules,\formative)$ from \refEx{ex:staticgraph} to
\linebreak
$\{(\emptyset,\Rules,\static_\Rules,\formative)\}$.
\end{example}

The subterm criterion can be strengthened, following
\cite{kus:iso:sak:bla:09,suz:kus:bla:11}, to also for instance handle
DPs like the one in \refEx{ex:derivsdp}.
Here, we focus on a new idea.
When considering \emph{computable} chains, we can build on the idea of
the subterm criterion to get something more.

\edef\procStaticSubtermCriterion{\number\value{theorem}}
\edef\procStaticSubtermCriterionSec{\number\value{section}}
\newcommand{\staticSubtermCriterionTheProc}{
\begin{theorem}[Static subterm criterion processor]\label{thm:staticsubtermproc}
The\linebreak processor $\Proc_{\mathtt{statcrit}}$  that maps a DP problem $(\P,
\Rules,\static_\AlterRules,\linebreak f)$ to $\{(\P_2,\Rules,\static_\AlterRules,
f)\}$ provided the following conditions hold is both
sound and complete:
\begin{itemize}
\item $\P = \P_1 \uplus \P_2$ is non-collapsing
\item a projection function $\nu$ exists such that $\project(\ell)
  \sqsupset \project(p)$ for all $\ell \arrdp p\ (A) \in \P_1$ and
  $\project(\ell) = \project(p)$ for all $\ell \arrdp p\ (A) \in
  \P_2$; here, $\sqsupset$ is the relation on base-type terms such
  that $s \sqsupset t$ if \ext{$s \neq t$ and (a)} $s \gracc t$ or \ext{(b)}
  there exists a meta-variable
  $Z$ such that $s \gracc \meta{Z}{x_1,\dots,x_\mia}$ and $t =
  \apps{\meta{Z}{t_1,\dots,t_\mac}}{s_1}{s_n}$
\end{itemize}
\CFold{The processor does not use the set $S$ that parametrises
the flag $\static_S$ at all. Thus, would this work for
\emph{arbitrary} $S$?}
\CKold{It is used in the definition of a finite chain.  Note that
  $S$-computability implies that $\arrr{S} \supseteq \arrr{\Rules}$.}

\end{theorem}
}
\staticSubtermCriterionTheProc

\begin{proof}[Proof sketch]
If the given conditions are satisfied, every infinite $(\P,
\Rules)$-chain induces an infinite $(\accreduce{C_\AlterRules} \cup
\arr{\beta})^* \cdot \arrr{\Rules}$ sequence (where $C_\AlterRules$
is a computability set for $\arr{\AlterRules}$ following
\refThm{thm:defC}) that starts in an immediate subterm of $t_1$,
contradicting computability unless all but finitely many steps are
equality, while every occurrence of a pair in $\P_1$ results in a
strict $(\accreduce{C} \cup \arr{\beta})^+$ step.
\end{proof}

\begin{example}\label{ex:ordrecdone}
By
 \refEx{ex:ordrecstatic}, the AFSM $(\F,\Rules)$ from
\refEx{ex:ordrec} is terminating if $(\P,\Rules,\static_\Rules,
\formative)$ is finite, where $\P$ is:
\[
\begin{array}{rcll}
\symb{rec}^\sharp\ (\suc\ X)\ K\ F\ G & \arrdp & \symb{rec}^\sharp\ 
  X\ K\ F\ G\ (\emptyset) \\
\symb{rec}^\sharp\ (\symb{lim}\ H)\ K\ F\ G & \arrdp &
  \symb{rec}^\sharp\ (H\ M)\ K\ F\ G\ (\emptyset) \\
\end{array}
\]
Consider the projection function $\nu$ with $\nu(\symb{rec}^\sharp) =
1$.
As $\suc\ X \gracc X$ and $\symb{lim}\ H \gracc H$, we have
both $\suc\ X \sqsupset X$ and $\symb{lim}\ H \sqsupset H\ M$.
Thus $\Proc_{\mathtt{statc}}(\P,\Rules,\static_\Rules,\formative) =
\{ (\emptyset,\Rules,\linebreak
\static_\Rules,\formative) \}$.
We obtain termination of $\Rules$ by using the dependency graph
processor on the remaining DP problem.
\end{example}


The static subterm
\ext{criterion} fundamentally relies on the new
$\static_\Rules$ flag, so it has no counterpart in the literature
so far.

\subsection{Non-termination}

While (most of) the processors
presented thus far are complete,
none of them can actually return \no.
We have not yet implemented any dedicated automation to this end;
however, we can already
give a general specification of such a
non-termination processor.

\edef\procInfinite{\number\value{theorem}}
\edef\procInfiniteSec{\number\value{section}}
\newcommand{\infiniteTheProc}{
\begin{theorem}[Non-termination processor]\label{def:nontermproc}
Let $M = (\P,\Rules,\linebreak m,f)$ be a DP problem.
The processor $\Proc_{\mathtt{infinite}}$ that
maps $M$ to
\no\ if it
determines that a sufficient criterion for non-termination of $\Rules$
or for existence of an infinite
$(\P,\Rules)$-dependency chain according to the flags
$m$ and $f$ holds is sound and complete.
\end{theorem}
}
\infiniteTheProc


\begin{example}\label{ex:nontermproc}
The AFSM from \refEx{ex:lambdadynamic} is non-terminating
if $(\P,\Rules,\minimal,\formative)$ is infinite, where $\P :=
\DDP(\Rules) = \{ \symb{ap}\ (\symb{lm}\ F)\ X \arrdp F\ X\ 
(\emptyset) \}$.  Write $s \arrdp_{\P} t$ if there exist a pair
$\ell \arrdp p\ (A) \in \P$ and a substitution $\gamma$ respecting
$A$ such that $s = \ell\gamma$ and $t = p\gamma$. Then the existence of
a loop $s_1\ (\arrdp_\P \cup \arr{\Rules}) \dots (\arrdp_\P \cup
\arr{\Rules}) \ s_n\ (\arrdp_\P \cup \arr{\Rules})\ s_1$ such that the
strict subterms of each $s_i$ are not instances of any rule
and such that every step is formative
is
certainly a sufficient criterion.
Thus, since such a sequence is given in \refEx{ex:lambdadynamic},
$\Proc_{\mathtt{infinite}}(\DDP(\Rules),\Rules,\minimal,\formative)
= \no$.
\end{example}


Aside from dedicated higher-order criteria (as used in \refEx{ex:lambdadynamic}),
we can also borrow non-termination criteria from first-order rewriting
\cite{gie:thi:sch:05:1,pay:08,emm:eng:gie:12}, with minor adaptions to
the typed setting.

\section{Conclusions and Future Work}\label{sec:conclusions}

We have presented 
a full integration of static and dynamic
higher-order dependency pairs in a unified DP framework
for termination and non-termination.
Our formulation is based on AFSMs, which makes it applicable to
various higher-order rewriting formalisms.

This framework considers not only arbitrary and minimally\linebreak non-terminating
dependency chains, but also minimally non-com-\linebreak{}putable 
chains.
This
lets us
efficiently
handle rules like \refEx{ex:ordrec} even
as part of a larger system.
We
use
formative reductions to isolate the
use of ``tags'' in \cite{kop:raa:12} to a 
processor and presented 
several new processors,
many of
which have no counterpart in 
the
literature.

To provide a strong formal groundwork for the higher-order
DP framework, we have presented many of our processors in a general
way, using semantic definitions of, e.g.,
the dependency
graph approximation as well as formative and usable rules rather than
syntactic definitions using functions like
$\mathit{TCap}$ \cite{gie:thi:sch:05:1}.
Even so,
most parts of the DP framework for AFSMs have been
implemented in the open-source termination prover
\wanda~\cite{wanda},
which has won the higher-order category of the
International Termination Competition~\cite{termcomp}
(analysing AFSs) four times.
In the International Confluence Competition~\cite{coco},
the 
tools \acph~\cite{acph16} and \csiho~\cite{csiho16}
also 
use \wanda\ as their ``oracle'' of choice for
termination proofs on HRSs.
This highlights the versatility of AFSMs and the
amena\-bility of the DP framework as a theoretical foundation
for an automatic termination analysis tool for higher-order rewriting.

\paragraph{Future work.}
DP frameworks for first-order rewriting exist also
specialised to
particular rewrite strategies, such as
innermost \cite{gie:thi:sch:fal:06} and
context-sensitive \cite{ala:gut:luc:10} rewriting. In future work,
we plan to extend the 
higher-order DP
framework to
rewrite
strategies, most importantly
implicit $\beta$-normalisation,
to have a complete analysis also for pattern HRSs of third order
and above.  Strategies inspired by the evaluation strategies of
functional programming languages like OCaml or Haskell are
natural additions as well.
A different direction for extensions would be to reduce the number of
term constraints
solved by the reduction triple
processor via a tighter integration with usable and formative rules
with respect to argument filterings, and
dedicated automation for detecting non-termination.

\bibliography{references}

\newpage
\appendix
\edef\savedcounter{\number0}
\edef\savedcounterSec{\number0}
%
\newcommand{\startappendixcounters}{
\setcounter{theorem}{\savedcounter}
\setcounter{section}{\savedcounterSec}
\renewcommand{\thetheorem}{\Alph{section}.\arabic{theorem}}
}
\newcommand{\oldcounter}[2]{
\edef\savedcounter{\number\value{theorem}}
\edef\savedcounterSec{\number\value{section}}
\setcounter{theorem}{#1}
\setcounter{section}{#2}
\renewcommand{\thetheorem}{\arabic{section}.\arabic{theorem}}
}
\startappendixcounters

\newcommand{\comprel}{\sqsupset}


\section{Arity functions}\label{app:arity}

In this appendix, we will prove the claim made in the text that if
$\Rules$ respects a certain arity function $\arity$ -- that is, if for
all $\ell \arrz r \in \Rules$ both $\ell$ and $r$ respect $\arity$ --
then for the sake of termination we only need to consider terms that
respect $\arity$.

To this end, we will use the following transformation:

\newcommand{\increasear}{\mathit{inc}_\arity}
\begin{definition}
Given an arity function $\arity$ and term $s$, let $\increasear(s)$
be given by:
\[
\begin{array}{rcl}
\increasear(\apps{x}{t_1}{t_n}) & \!\!\!\!=\!\!\!\! &
  \apps{x}{\increasear(t_1)}{\increasear(t_n)} \hfill
  (\text{if}\ x \in \V) \\
\increasear(\apps{\identifier{f}}{t_1}{t_n}) & \!\!\!\!=\!\!\!\! &
  \abs{x_{n+1}\dots x_\mia}{\apps{
  \identifier{f}}{\increasear(t_1)}{\increasear(t_n)}}\phantom{A} \\
  & & x_{n+1} \cdots x_\mia
  \hfill (\text{if}\ \identifier{f} \in \F\ \text{and}\ 
  \mia = \arity(\identifier{f})) \\
\!\!\increasear(\apps{(\abs{x}{s})}{t_1}{t_n}) & \!\!\!\!=\!\!\!\! &
  \apps{(\abs{x}{\increasear(s)})}{\increasear(t_1)}{\increasear(t_n)} \\
\end{array}
\]
Here, $\abs{x_{n+1}\dots x_\mia}{t}$ should be read as just $t$ if
$n \geq \mia$.
\end{definition}

To prove results on $\increasear$, the following observation is
useful:

\begin{lemma}\label{lem:increasearapply}
$\apps{\increasear(s)}{\increasear(t_1)}{\increasear(t_n)}
\arrr{\beta} \increasear(\apps{s}{t_1}{t_n})$ for all $s,\vec{t}$.
\end{lemma}

\begin{proof}
Consider the form of $s$.  If $s$ has any form other than
$\apps{\identifier{f}}{u_1}{u_i}$ with $i < \arity(\identifier{f})$,
we immediately obtain
$\increasear(\apps{s}{t_1}{t_n}) =
\apps{\increasear(s)}{\increasear(t_1)}{\increasear(t_n)}$.  If $s$
does have this form,
on the one hand
$\increasear(s) = \abs{x_{i+1} \dots x_\mia}{
\apps{\apps{\identifier{f}}{\increasear(u_1)}{\increasear(u_i)}}{
x_{i+1}}{x_\mia}}$, where $\mia := \arity(\identifier{f})$, and
on the other hand
$\increasear(\apps{s}{t_1}{t_n}) =\linebreak \abs{x_{i+n+1}
\ldots x_\mia}{
\apps{\apps{\apps{\identifier{f}}{\increasear(u_1)}{\increasear(u_i)}}{
\increasear(t_1)}{\increasear(t_n)}}{\linebreak
x_{i+n+1}}{x_\mia}}$.  It is clear
that $\apps{\increasear(s)}{\increasear(t_1)}{\increasear(t_n)}$
reduces in either $n$ steps (if $n+i \leq \mia$) or $\mia-i$ steps (if
$n + i > \mia$) to $\increasear(\apps{s}{t_1}{t_n})$.
\end{proof}

Having \refLemma{lem:increasearapply} as an aid, we easily see that
$\increasear$ is well-behaved for substitution, both on terms and
meta-terms.  In the following lemmas, $\gamma$ is a substitution and
$\gamma^{\increasear}$ is the substitution on domain $\domain(\gamma)$
that maps each (meta-)variable $x$ to $\increasear(\gamma(x))$.

\begin{lemma}\label{lem:increasearterm}
If $s$ is a term, then $\increasear(s)\gamma^{\increasear} \arrr{\beta}
\increasear(s\gamma)$.
\end{lemma}

\begin{proof}
By a straightforward induction on the form of $s$; the only case where
$\arrr{\beta}$ is explicitly needed is when $s = \apps{x}{t_1}{t_n}$
and $\gamma(x) = \apps{\identifier{f}}{u_1}{u_i}$ with $i < \arity(
\identifier{f})$, in which case we apply \refLemma{lem:increasearapply}.
We will set out each of the cases in detail.

If $s = \apps{\identifier{f}}{t_1}{t_n}$, then $\increasear(s) =
(\abs{x_{n+1}\dots x_\mia}{\apps{\apps{\identifier{f}}{\increasear(t_1)}{
\linebreak
\increasear(t_n)}}{x_{n+1}}{x_\mia}})$ for $\mia := \arity(\identifier{f})$, 
so $\increasear(s)\gamma^{\increasear} =
\abs{x_{n+1}\dots \linebreak
x_\mia}{\apps{\apps{\identifier{f}}{
(\increasear(t_1)\gamma^{\increasear})}{
(\increasear(t_n)\gamma^{\increasear})}}{x_{n+1}}{x_\mia}}$,
which by the induction hypothesis $\arrr{\beta}
\abs{x_{n+1} \dots x_\mia}{\apps{\apps{\identifier{f}}{
\increasear(t_1\gamma)}{\increasear(t_n\gamma)}}}{\linebreak
x_{n+1}}{x_\mia} =
\increasear(\apps{\identifier{f}}{(t_1\gamma)}{(t_n\gamma)}) =
\increasear(s\gamma)$.

If $s = \apps{(\abs{x}{t_0})}{t_1}{t_n}$ we also easily complete by
induction hypothesis:
$\increasear(s)\gamma^{\increasear} =
\apps{(\abs{x}{(\increasear(t_0)\gamma^{\increasear})})}{
(\increasear(t_1)\gamma^{\increasear})}{\linebreak
(\increasear(t_n)\gamma^{\increasear})}$, which by the induction
hypothesis reduces to $\apps{(\abs{x}{\linebreak
\increasear(t_0\gamma)})}{
\increasear(t_1\gamma)}{\increasear(t_n\gamma)} =
\increasear(s\gamma)$.

If $s = \apps{x}{t_1}{t_n}$ then 
$\increasear(s)\gamma^{\increasear} =
\apps{\gamma^{\increasear}(x)}{(\increasear(t_1)\gamma^{\increasear})
\linebreak
}{(\increasear(t_n)\gamma^{\increasear})}$,
which by the induction hypothesis reduces to
$\apps{\gamma^{\increasear}(x)}{\increasear(t_1\gamma)}{
\increasear(t_n\gamma)}$, and by
\refLemma{lem:increasearapply} this reduces to
$\increasear(\apps{\gamma(x)}{(t_1\gamma)}{(t_n\gamma)}) =
\increasear(s\gamma)$.
\end{proof}

\begin{lemma}\label{lem:increasearmetasubst}
If $\gamma(Z) \approxp \abs{x_1 \dots x_\mac}{u}$, then also
$\gamma^{\increasear}(Z) \approxp
\linebreak
\abs{x_1 \dots
x_\mac}{\increasear(u)}$.
\end{lemma}

\begin{proof}
If $\gamma(Z) = \abs{x_1 \dots x_\mac}{u}$ then clearly
$\gamma^{\increasear}(Z) = \abs{x_1 \dots
\linebreak
x_\mac}{\increasear(u)}$, so
also $\gamma(Z) \approxp \abs{x_1 \dots x_\mac}{\increasear(u)}$.

Otherwise, $\gamma(Z) = \abs{x_1 \dots x_i}{q}$ with $i < \mac$ and $u =
\apps{q}{x_{i+1}}{x_\mac}$, and $q$ not an abstraction.  Write $q =
\apps{q'}{q_1}{q_n}$ with $q'$ not an application (so either $q'$ is
an abstraction and $n > 0$, or $q' \in \V \cup \F$).

If $q' \notin \F$, then $\increasear(q) = \apps{\increasear(q')}{
\increasear(q_1)}{\increasear(q_n)}$ and $\increasear(\apps{q}{x_{i+1}}{
x_\mac}) = \apps{\apps{\increasear(q')}{\increasear(q_1)}{\increasear(q_n)}
}{x_{i+1}}{x_\mac}$.  Also if $q' \in \F$ but $i \geq \arity(q')$ we have
$\increasear(q) = \apps{q'}{\increasear(q_1)}{\linebreak
\increasear(q_n)}$ and
$\increasear(\apps{q}{x_{i+1}}{x_\mac}) = \apps{\apps{q'}{\increasear(q_1)
}{\increasear(q_n)}}{x_{i+1}\linebreak
}{x_\mac}$.
Either way, $\gamma^{\increasear}(Z) \approxp \abs{x_1 \dots x_\mac}{
\apps{\increasear(q)}{x_{i+1}}{x_\mac}} = \abs{x_1 \dots x_\mac}{
\increasear(u)}$.

If $q' = \identifier{f} \in \F$ and $n < \arity(\identifier{f}) =: \mia$,
then $\increasear(q) = \abs{y_{n+1} \dots y_\mia}{\apps{\apps{
\identifier{f}}{\linebreak
\increasear(q_1)}{\increasear(q_n)}}{y_{n+1}}{y_\mia}}$,
Using $\alpha$-conversion (and introducing new variables
$x_{\mac+1}\dots x_{i+\mia-n}$ if $\mia-n > \mac-i$) this is equal to
$\abs{x_{i+1} \dots x_{i+\mia-n}}{\apps{\apps{\identifier{f}}{
\increasear(q_1)}{\increasear(q_n)}}{x_{i+1}}{x_{i+\mia-n}}}$.  Thus
we have:
$\gamma^{\increasear}(Z) = \abs{x_1 \dots x_{i+\mia-n}}{\apps{\apps{
\identifier{f}}{\increasear(q_1)}{\increasear(q_n)}}{x_{i+1}}{\linebreak
x_{i+\mia-n}}}$.  There are two possibilities:
\begin{itemize}
\item $\mia - n \leq \mac - i$: then $i+\mia-n \leq \mac$ and
  $\gamma^{\increasear}(Z) \approxp \abs{x_1 \dots x_\mac}{\linebreak
  \apps{\apps{\identifier{f}}{\increasear(q_1)}{\increasear(q_n)}}{
  x_{i+1}}{x_{i+\mia-n}} \cdots x_\mac} = \abs{x_1 \dots x_\mac}{\linebreak
  \increasear(u)}$ because in this case
  $\increasear(u) = \increasear(\apps{\apps{\identifier{f}}{q_1}{q_n}}{
  \linebreak
  x_{i+1}}{x_\mac}) = \apps{\apps{\identifier{f}}{\increasear(q_1)}{
  \increasear(q_n)}}{x_{i+1}}{x_\mac}$.
\item $\mia - n > \mac - i$: then $\gamma^{\increasear}(Z) \approxp
  \gamma^{\increasear}(Z) = \abs{x_1 \dots x_\mac}{\abs{x_{\mac+1}\linebreak
  \dots x_{i+\mia-n}}{\apps{\apps{\identifier{f}}{\increasear(q_1)}{
  \increasear(q_n)}}{x_{i+1}}{x_\mac} \cdots x_{i+\mia-n}}} = \abs{x_1
  \dots x_\mac}{\increasear(u)}$ because in this case $\increasear(u) =
  \increasear(\apps{\apps{\identifier{f}}{q_1\linebreak
  }{q_n}}{x_{i+1}}{x_\mac}) =
  \abs{x_{\mac+1} \dots x_{i+a-n}}{\apps{\apps{\identifier{f}}{
  \increasear(q_1)}{\linebreak
  \increasear(q_n)}}{x_{i+1}}{x_{i+\mia-n}}}$.
\qedhere
\end{itemize}
\end{proof}

\begin{lemma}\label{lem:increasearmetaterm}
If $s$ is a meta-term that respects $\arity$ and whose domain includes
all meta-variables in $s$, then $s\gamma^{\increasear} \arrr{\beta}
\increasear(s\gamma)$.
\end{lemma}

\begin{proof}
We use an induction similar to the one used in the proof of
\refLemma{lem:increasearterm}.  This gives the same cases and
reasoning as before, except that in the first case ($s =
\apps{\identifier{f}}{t_1}{t_n}$) we have $n \geq k$ because $s$
respects $\arity$, so $\increasear(s\gamma) = \apps{\identifier{f}}{
\increasear(t_1\gamma)}{\increasear(t_n\gamma)}$, allowing the induction
to proceed.  We also have an additional case
if $s = \apps{\meta{Z}{s_1,\dots,s_\mac}}{t_1}{t_n}$ and
$\gamma \approxp \abs{x_1 \dots x_\mac}{u}$, then (by
\refLemma{lem:increasearmetasubst}) $s\gamma^{\increasear} =
\apps{\increasear(u)[x_1:=s_1\gamma^{\increasear},\dots,x_\mac:=
s_\mac\gamma^{\increasear}]}{(t_1\gamma^{\increasear})}{\linebreak
(t_n\gamma^{\increasear})}$, which by the induction hypothesis reduces
to $\apps{\increasear(u)\linebreak{}[x_1
:=\increasear(s_1\gamma),\dots,x_\mac:=
\increasear(s_\mac\gamma)]}{
\increasear(t_1\gamma)}{\increasear(t_n\gamma)}$, and by
\refLemma{lem:increasearterm} to
$\apps{\increasear(u[x_1:=s_1\gamma,\dots,x_\mac:=s_\mac\gamma])}{
\increasear(t_1\gamma)\linebreak}{\increasear(t_n\gamma)}$.
Now we complete with \refLemma{lem:increasearapply}.
\end{proof}

\begin{lemma}\label{lem:increasearpattern}
If $s$ is a pattern that respects $\arity$ and $\domain(\gamma)$
contains only meta-variables, and contains all meta-variables in $s$,
then $s\gamma^{\increasear} = \increasear(s\gamma)$.
\end{lemma}

\begin{proof}
We again proceed by induction.  If $s = \apps{\identifier{f}}{t_1}{t_n}$
(with $n \geq \arity(\identifier{f})$) we immediately complete with
the induction hypothesis;, 
also
if $s = \abs{x}{s'}$ or
$\apps{x}{t_1}{t_n}$ with $x \in \V$ (since $x \notin \domain(\gamma)$).
The only remaining case is $\meta{Z}{x_1,\dots,x_\mia}$, in which case
we can write $\gamma(Z) = \abs{x_1 \dots x_\mia}{t}$ and have
$s\gamma^{\increasear} = \increasear(t) = \increasear(s\gamma)$.
\end{proof}

The results on substitution make it easy to prove that reduction is
preserved under $\increasear$:

\begin{lemma}\label{lem:arityrespect}
Suppose that all rules in $\Rules$ respect $\arity$.
Then whenever $s \arr{\Rules} t$ also $\increasear(s) \arr{\Rules}
\increasear(t)$.
\end{lemma}

\begin{proof}
By induction on the form of $s$.
If $s$ has the form $\apps{x}{s_1}{s_n}$ or $\apps{\identifier{f}}{
s_1}{s_n}$ or $\apps{(\abs{x}{s_0})}{s_1}{s_n}$ and the reduction takes
place in one of the $s_i$, we immediately complete by the induction
hypothesis.
Otherwise the reduction takes place at the head; it could be either by
a rule step or a $\beta$ step.

In the first case, $s = \apps{\identifier{f}}{s_1}{s_n}$ and there exist
a rule $\apps{\identifier{f}}{\ell_1}{\ell_i} \arrz r$ and substitution
$\gamma$ such that $n \geq i$ and $s_j = \ell_j\gamma$ for $1 \leq j
\leq i$ and $t = \apps{(r\gamma)}{s_{i+1}}{s_n}$.  Since $\Rules$
respects $\arity$, in particular $\apps{\identifier{f}}{\ell_1}{\ell_i}$
and $r$ do so, and $n \geq i \geq \arity(\identifier{f})$.  Thus,
\begin{align*}
&\increasear(s)\\
=\ \ \ \ & \apps{\identifier{f}}{\increasear(s_1)}{\increasear(s_n)} \\
=\ \ \ \ & \apps{\identifier{f}}{\increasear(\ell_1\gamma)}{\increasear(
  \ell_i \gamma)} \cdots \increasear(s_n) \\
=\ \ \ \ & \apps{\apps{\identifier{f}}{(\ell_1\gamma^{\increasear})}{
  (\ell_i\gamma^{\increasear})}}{\increasear(s_{i+1})}{\increasear(s_n)} &
  \tag{\text{by \refLemma{lem:increasearpattern}}} \\
\arr{\Rules}\ & \apps{(r\gamma^{\increasear})}{\increasear(s_{i+1})}{
  \increasear(s_n)} \\
\arrr{\beta}\ & \apps{\increasear(r\gamma)}{\increasear(s_{i+1})}{
  \increasear(s_n)}
  \tag{\text{by \refLemma{lem:increasearmetaterm}}} \\
\arrr{\beta}\ & \increasear(\apps{(r\gamma)}{s_{i+1}}{s_n})
  \tag{\text{by \refLemma{lem:increasearapply}}} \\
=\ \ \ \ & \increasear(t)
\end{align*}

In the second case, $s = \apps{(\abs{x}{u})}{s_0}{s_n}$ and
thus
$t = \apps{u[x:=s_0]}{s_1}{s_n}$.  We have
\begin{align*}
&\increasear(s)\\
=\ \ \ \ & \apps{(\abs{x}{\increasear(u)})}{\increasear(s_0)}{
  \increasear(s_n)} \\
\arr{\beta}\ \ & \apps{\increasear(u)[x:=\increasear(s_0)]}{
  \increasear(s_1)}{\increasear(s_n)} \\
\arrr{\beta}\ \ & \apps{\increasear(u[x:=s_0])}{\increasear(s_1)}{
  \increasear(s_n)}
   \tag{\text{by \refLemma{lem:increasearterm}}} \\
\arrr{\beta}\ \ & \increasear(\apps{u[x:=s_0]}{s_1}{s_n})
  \tag{\text{by \refLemma{lem:increasearapply}}} \\
=\ \ \ \ & \increasear(t)
\end{align*}
\end{proof}

With reduction steps being preserved, the statement from the text
now follows immediately:

\begin{theorem}\label{thm:arityrespect}
If for all rules $\ell \arrz r$ in $\Rules$ both $\ell$ and $r$ respect
$\arity$, then $\arr{\Rules}$ is terminating if and only if all
arity-respecting terms are terminating under $\arr{\Rules}$.
\end{theorem}

\begin{proof}
Termination of $\arr{\Rules}$ is by definition equivalent to
termination of all terms under $\arr{\Rules}$.
Obviously if all terms are terminating, then so are all arity-respecting
terms.  This provides one direction.
For the other, suppose that there is a non-terminating term, so
we can construct a sequence $s_0 \arr{\Rules} s_1 \arr{\Rules} s_2
\arr{\Rules} \dots$.  Then
\refLemma{lem:arityrespect} implies that there is
a non-terminating arity-respecting term as well: $\increasear(s_0)
\arr{\Rules}^+ \increasear(s_1) \arr{\Rules}^+ \increasear(s_2)
\arr{\Rules}^+ \dots$.
\end{proof}

\section{Computability: the set $C$}\label{app:computability}

In this appendix, we prove \refThm{thm:defC}: the existence of an
RC-set $C$ that provides an accessibility relation $\gracc$ that
preserves computability, and a base-type accessibility step
$\accreduce{C}$ that preserves both computability and termination.

As we have said before, $\V$ and $\F$ contain infinitely many symbols
of all types.  We will use this to select variables or constructor
symbol of any given type without further explanation.

These proofs \emph{do not require} that computability is considered
with respect to a rewrite relation: other relations (such as recursive
path orderings) may be used as well.  To make this explicit, we will
use an alternative relation symbol, $\comprel$.

Note: a more extensive discussion of complexity can be found
in~\cite{bla:16}.  Our notion of \ext{accessibility} largely corresponds
\ext{to}
membership of the computability closure defined there (although not
completely).

\subsection{Definitions and main computability result}

\begin{definition}\label{def:comprel}
In \refApp{app:computability}, $\comprel$ is assumed to be a given
relation on terms of the same type, with respect to which we consider
computability.  We require that:
\begin{itemize}
\item $\comprel$ is monotonic (that is, $s \comprel t$ implies that
  $s\ u \comprel s'\ u$ and $u\ s \comprel u\ s'$ and $\abs{x}{s}
  \comprel \abs{x}{s'}$);
\item for all variables $x$: $\apps{x}{s_1}{s_n} \comprel t$ implies
  that $t$ has the form $\apps{x}{s_1}{s_i'} \cdots s_n$ with
  $s_i \comprel s_i'$;
\item if $s \arrr{\mathtt{head}\beta} u$ and $s \comprel t$, then
  there exists $v$ such that $u \comprel^* v$ and $t \arrr{\mathtt{
  head}\beta} v$;
  here, $\arr{\mathtt{head}\beta}$ is the relation generated by the
  head-step $\apps{(\abs{x}{u})\ v}{w_1}{w_n}
  \arr{\mathtt{head}\beta} \apps{u[x:=v]}{w_1}{w_n}$;
\item if $t$ is the $\mathtt{head}\beta$-normal form of $s$, then
  $s \comprel^* t$.
\end{itemize}
We call a term \emph{neutral} if it has the form $\apps{x}{s_1}{s_n}$
or $\apps{(\abs{x}{u})}{s_0}{s_n}$.
\end{definition}

The generality obtained by imposing only the minimal requirements needed
on $\comprel$ is not needed in the current paper (where we only consider
computability with respect to a rewrite relation), but could be used to
extend the method to other domains.
First note:

\begin{lemma}\label{lem:rewriterelationsuffices}
A rewrite relation $\arr{\Rules}$ satisfies the requirements of
$\comprel$ stated in \refDef{def:comprel}.
\end{lemma}

\begin{proof}
Clearly $\arr{\Rules}$ is monotonic, applications with a variable at
the head cannot be reduced at the head, and
moreover
$\arr{\Rules}$ includes
$\arr{\mathtt{head}\beta}$.

The third property we prove by induction
on $s$ with $\arr{\beta}$, using $\arrr{\Rules}$ instead of
$\arr{\Rules}$ for a stronger induction hypothesis.  If $s = u$, then we
are done choosing $v := t$.  Otherwise we can write
$s = \apps{(\abs{x}{q})\ w_0}{w_1}{w_n}$
and $s \arr{\mathtt{head}\beta} s' := \apps{q[x:=w_0]}{w_1}{w_n}$,
and $s' \arrr{\mathtt{head}\beta} u$.
If the reduction $s \arrr{\Rules} t$ does not take any head steps, then
\[
t = \apps{(\abs{x}{q'})\ w_0'}{w_1'}{w_n'} \arrr{\mathtt{head}\beta}
\apps{q'[x:=w_0']}{w_1'}{w_n'} =: v
\]
and indeed $u \arrr{\Rules} v$ by
monotonicity.  Otherwise, by the same argument we can safely assume
that the head step is done first, so $s' \arrr{\Rules} t$; we complete
by the induction hypothesis.
\end{proof}

\medskip
Recall \refDef{def:RCset} from the text.

\oldcounter{\defRCset}{\defRCsetSec}
\begin{definition}[with $\comprel$ rather than $\arr{\Rules}$]
A \emph{set of reducibility candidates}, or \emph{RC-set}, for a
relation $\comprel$ as in \refDef{def:comprel} is a set $I$ of
base-type terms $s : \asort$ such that:
\begin{itemize}
\item every term in $I$ is terminating under $\comprel$
\item $I$ is closed under $\comprel$ (so if $s \in I$ and $s \comprel t$
  then $t \in I$)
\item if $s$ is neutral, and for all $t$ with $s \comprel t$ we have
  $t \in I$, then $s \in I$
\end{itemize}
We define $I$-computability for an RC-set $I$ by induction on types:
\begin{itemize}
\item $s : \asort$ is $I$-computable if $s \in I$ ($\asort \in \Sorts$)
\item $s : \atype \arrtype \btype$ is $I$-computable if for all terms
  $t : \atype$ that are $I$-computable, $\app{s}{t}$ is $I$-computable
\end{itemize}
\end{definition}
\startappendixcounters

For $\asort$ a sort and $I$ an RC-set, we will write $I(\asort) = \{
s \in I \mid s : \asort \}$.

Let us illustrate \refDef{def:RCset} with two examples:

\begin{lemma}\label{lem:MINSN}
The set \textsf{SN} of all terminating base-type terms is an RC-set.
The set \textsf{MIN} of all terminating base-type terms whose
$\mathtt{head}\beta$-normal form can be written $\apps{x}{s_1}{
s_\maa}$ with $x \in \V$ is also an RC-set.
\end{lemma}

\begin{proof}
It is easy to verify that the requirements hold for \textsf{SN}.  For
\textsf{MIN}, clearly termination holds.  If $s \in \textsf{MIN}$,
then $s \arrr{\mathtt{head}\beta} \apps{x}{s_1}{s_\maa} =: s'$, so for
any $t$ with $s \comprel^* t$ the assumptions on $\comprel$ provide
that $t \arrr{\mathtt{head}\beta} v$ for some $\comprel^*$-reduct of
$s'$, which can only have the form $\apps{x}{t_1}{t_\maa}$.
Finally, we prove that a neutral term $s$ is in \textsf{MIN} if all
its $\comprel^+$-reducts are, by induction on $s$ with $\arr{\beta}$
(this suffices because we have already seen that \textsf{MIN} is
closed under $\comprel$).
If $s = \apps{x}{s_1}{s_\maa}$ then it is included in \textsf{MIN} if
it is terminating, which is the case if all its reducts are
terminating, which is certainly the case if they are in \textsf{MIN}.
If $s = \apps{(\abs{x}{u})\ v}{w_1}{w_\maa}$ then it is included if
(a) all its reducts are terminating (which is satisfied if they are in
\textsf{MIN}), and (b) the $\mathtt{head}\beta$-normal form $s'$ of
$s$ has the right form, which holds because $s \comprel^+ s'$
(as $\arr{\mathtt{head}\beta}$ is included in $\comprel$) and
therefore $s' \in \textsf{MIN}$ by assumption.
\end{proof}

In fact, we have that \textsf{MIN} $\subseteq I \subseteq$ \textsf{SN}
for all RC-sets $I$.  The latter inclusion is obvious by the
termination requirement in the definition of RC-sets.  The former
inclusion follows easily:

\begin{lemma}\label{lem:mincom}
For all RC-sets $I$, \textsf{MIN} $\subseteq I$.
\end{lemma}

\begin{proof}
We prove by induction on $\comprel$ that all elements of \textsf{MIN}
are also in $I$.  It is easy to see that if $s \in \textsf{MIN}$ then
$s$ is neutral.  Therefore, $s \in I$ if $t \in I$ whenever $s
\comprel t$.  But since \textsf{MIN} is closed by
\refLemma{lem:MINSN}, each such $t$ is in \textsf{MIN}, so also in
$I$ by the induction hypothesis.
\end{proof}

Aside from minimality of \textsf{MIN}, \refLemma{lem:mincom}
actually provides $I$-computability of all variables, regardless of $I$.
We prove this alongside termination of all $I$-computable terms.

\begin{lemma}\label{lem:compresults}
Let $I$ be an RC-set.  The following statements hold for all types
$\atype$:
\begin{enumerate}
\item\label{lem:compresults:vars}
  all variables $x : \atype$ are $I$-computable
\item\label{lem:compresults:term}
  all $I$-computable terms $s : \atype$ are terminating (wrt. $\comprel$)
\end{enumerate}
\end{lemma}

\begin{proof}
By a mutual induction on the form of $\atype$, which we may safely
write $\atype_1 \arrtype \dots \arrtype \atype_\maa \arrtype \asort$
(with $\maa \geq 0$ and $\asort \in \Sorts$).

(\ref{lem:compresults:vars}) By definition of $I$-computability,
$x : \atype$ is computable if and only if $\apps{x}{s_1}{s_\maa} \in I$
for all $I$-computable terms $s_1 : \atype_1,\dots,s_\maa : \atype_\maa$.
However, as all $\atype_i$ are smaller types, we know that such terms
$s_i$ are terminating, so \refLemma{lem:mincom} gives the required
result.

(\ref{lem:compresults:term}) Let $x_1 : \atype_1,\dots,x_\maa :
\atype_m$ be variables; by the induction hypothesis they are
computable, and therefore $\apps{s}{x_1}{x_\maa}$ is in $I$ and
therefore terminating. Then the head, $s$, cannot itself be
non-terminating (by monotonicity of $\comprel$).
\end{proof}

While \textsf{SN} is indisputably the easiest RC-set to define and
work with, it will be beneficial for the strength of the method to
consider a set strictly between \textsf{MIN} and \textsf{SN}.  To this
end, we assume given an ordering on types, and a function mapping each
function symbol $\afun$ to a set $\Acc(\afun)$ of arguments positions.
Here, we deviate from the text by not fixing $\Acc$; again, this
generality is not needed for the current paper, but is done with an
eye on future extensions.

\begin{definition}
Assume given a quasi-ordering $\greqsort$ on $\Sorts$ whose strict part
$\grsort\ :=\ \greqsort \setminus \leqsort$ is well-founded.  Let $\eqsort$
denote the corresponding equivalence relation $\eqsort\ :=\ \greqsort \cap
\leqsort$.

For a type $\atype\ \equiv\ \atype_1 \arrtype \dots \arrtype \atype_\maa
\arrtype \bsort$ (with $\bsort \in \Sorts$) and sort $\asort$, we
write $\ext{\asort \gracsortup \atype}$ if $\asort \greqsort \bsort$ and
\ext{$\asort \gracsortdown \atype_i$ for each $i$}, and we
write $\ext{\asort
\gracsortdown \atype}$ if $\asort \grsort \bsort$ and \ext{$\asort
\gracsortup \atype_i$ for each $i$.}

For $\afun : \atype_1 \arrtype \dots \arrtype \atype_\maa \arrtype
\asort$ we assume given a set $\Acc(\afun) \subseteq \{ i \mid 1
\leq i \leq \maa \wedge \ext{\asort \gracsortup \atype_i} \}$.
For $x : \atype_1 \arrtype \dots \arrtype \atype_\maa \arrtype \asort
\in \V$, we write $\Acc(x) = \{ i \mid 1 \leq i \leq \maa \wedge
\atype_i = \btype_1 \arrtype \dots \arrtype \btype_n \arrtype \bsort
\wedge \asort \greqsort \bsort \}$.
We
write $s \gracc t$ if either $s = t$, or $s = \abs{x}{s'}$ and $s'
\gracc t$, or $s = \apps{a}{s_1}{s_n}$ with $a \in \F \cup \V$ and $s_i
\gracc t$ for some $i \in \Acc(a)$.
\end{definition}

\emph{Remark:} This definition of the accessibility relations deviates
from, e.g., \cite{bla:jou:rub:15} by using a pair of relations ($\ext{\gracsortup}$ and
$\ext{\gracsortdown}$) rather than positive and negative positions.  This is not
an important difference, but simply a matter of personal preference;
using a pair of relations avoids the need to discuss type positions in
the text, allowing for a shorter presentation.  It is also not common
to allow a choice in
$\Acc(\afun)$, but rather to fix $\Acc(\afun) = \{ \atype_i \mid 1 \leq i
\leq \maa \wedge \ext{\asort \gracsortup \atype} \}$ for \emph{some} symbols
(for instance constructors) and $\Acc(\afun) = \emptyset$ for the
rest.  We elected to leave the choice open for greater generality.

\medskip
The interplay of the positive and negative relations \ext{$\gracsortup$} and
\ext{$\gracsortdown$} leads to an important result on RC-sets.

\begin{lemma}\label{lem:gracsortinterplay}
Fix a sort $\asort \in \Sorts$.
Suppose $I,J$ are RC-sets such that $I(\bsort) = J(\bsort)$ for all
$\bsort$ with $\asort \grsort \bsort$ and $I(\bsort) \subseteq
J(\bsort)$ if $\asort \eqsort \bsort$.  Let $s : \atype$.  Then we have:
\begin{itemize}
\item If $\ext{\asort \gracsortup \atype}$, then if $s$ is $I$-computable
  also $s$ is $J$-computable.
\item If $\ext{\asort \gracsortdown \atype}$, then if $s$ is $J$-computable
  also $s$ is $I$-computable.
\end{itemize}
\end{lemma}

\begin{proof}
We prove both statements together by a shared induction on the form of
$\atype$.  We can always
write $\atype\ \equiv\ \atype_1 \arrtype
\dots \arrtype \atype_\maa \arrtype \bsort$ with $\bsort \in \Sorts$.

First suppose $\ext{\asort \gracsortup \atype}$; then $\asort \greqsort
\bsort$ -- so $I(\bsort) \subseteq J(\bsort)$ -- and each $\ext{\asort
\gracsortdown \atype_i}$.  Assume that $s$ is $I$-computable; we must show
that it is $J$-computable, so that for all $J$-computable $t_1 :
\atype_1,\dots,t_\maa : \atype_\maa$ we have:
$\apps{s}{t_1}{t_\maa} \in J$.  However, by the induction hypothesis each
$t_i$ is also $I$-computable, so $\apps{s}{t_1}{t_\maa} \in I(\bsort)
\subseteq J(\bsort)$ by the assumption.

For the second statement, suppose $\ext{\asort \gracsortdown \atype}$; then
$\asort \grsort \bsort$, so $I(\bsort) = J(\bsort)$.  Assume that $s$
is $J$-computable; $I$-computability follows if $\apps{s}{t_1}{t_\maa}
\in I(\bsort) = J(\bsort)$ whenever $t_1,\dots,t_\maa$ are\linebreak
$I$-computable.
By the induction hypothesis they are $J$-computable\ext{,} so this holds
by assumption.
\end{proof}

The RC-set $C$ whose existence is asserted below offers computability
with a notion of accessibility.  It is worth noting that this is
\emph{not} a standard definition, but is designed to provide an
additional relationship $\accreduce{I}$ that is terminating on computable
terms.  This re\-lation will be useful in termination proofs using static
DPs.

\begin{theorem}\label{thm:defC2}
Let $\accreduce{I}$ be the relation on base-type terms where $
\apps{\afun}{s_1}{s_\maa} \accreduce{I} \apps{s_i}{t_1}{t_n}$ whenever $
i \in \Acc(\afun)$ and $s_i : \atype_1 \arrtype \dots \arrtype \atype_n
\arrtype \asort $ and  each $t_j$ is $I$-computable.

There exists an RC-set $C$ such that $C = \{ s : \asort \mid \asort \in
\Sorts \wedge s$ is terminating under $\comprel \cup \accreduce{C}$ and
if $s \comprel^* \apps{\afun}{s_1}{s_\maa} : \asort$ then $s_i$ is
$C$-computable for all $i \in \Acc(\afun) \}$.
\end{theorem}

\begin{proof}
We will define, by well-founded induction on $\asort$ using $\greqsort$,
a set $A_\asort$ of terms as follows.

Assume $A_\bsort$ has already been defined for all $\bsort$ with
$\asort \grsort \bsort$, and let $X_\asort$ be the set of RC-sets $I$
such that $I(\bsort) = A_\bsort$ whenever $\asort \grsort \bsort$.
We observe that $X_\asort$ is a complete lattice with respect to
$\subseteq$: defining the bottom element $\sqcup \emptyset := \bigcup \{
A_\bsort \mid \asort \grsort \bsort \} \cup \textsf{MIN}$ and the
top element $\sqcap \emptyset := \bigcup \{ A_\bsort \mid \asort
\grsort \bsort \} \cup \bigcup \{$ \textsf{SN}$(\bsort) \mid \neg
(\asort \grsort \bsort) \}$, and letting
$\sqcup Z := \bigcup Z,\ \sqcap Z := \bigcap Z$ for non-empty $Z$, it
is easily checked that $\sqcap$ and $\sqcup$ give a greatest lower and
least upper bound within $X_\asort$ respectively.
Now for an RC-set $I \in X_\asort$, we let:
\[
\begin{array}{rcl}
F_\asort(I) & = & \{ s \in I \mid s : \bsort \not\eqsort \asort \} \\
& \cup & \{s \in \Terms(\F,\V,\arity) \mid s : \bsort \eqsort \asort
   \wedge s\ \text{is terminating} \\
& & \ \ \text{under}\ \comprel \cup \accreduce{I}
   \wedge\ \text{if}\ s \comprel^* \apps{\afun}{s_1}{s_\maa}\ 
   \text{for a symbol} \\
& & \ \ \afun \in \F\ \text{then}\ 
   \forall i \in \Acc(\afun)\ [t_i\ \text{is}\ I\text{-computable}]
   \} \\
\end{array}
\]

Clearly, $F_\asort$ maps elements of $X_\asort$ to $X_\asort$: terms of
type $\bsort \not\eqsort \asort$ are left alone, and $F_\asort(I)$
satisfies the properties to be an RC-set.  Moreover, $F_\asort$ is
monotone.  To see
this, let $I,J \in X_\asort$ such that $I \subseteq J$; we must see that
$F_\asort(I) \subseteq F_\asort(J)$.  To this end, let $s \in
F_\asort(I)$; we will see that also $s \in F_\asort(J)$.  This is
immediate if $s : \bsort \not\eqsort \asort$, as membership in $X_\asort$
guarantees that $F_\asort(I)(\bsort) = I(\bsort) \subseteq J(\bsort) =
F_\asort(J)(\bsort)$.  So assume $s : \bsort \eqsort \asort$.  We must see
two things:
\begin{itemize}
\item $s$ is terminating under $\comprel \cup \accreduce{J}$.  We
  show that $\comprel \cup \accreduce{J}\ \subseteq\ 
  \comprel \cup \accreduce{I}$; as $s$ is terminating in the
  latter, the requirement follows.  Clearly $\comprel\ \subseteq\ 
  \comprel \cup \accreduce{I}$, so assume $s \accreduce{J} s'$.
  Then $s = \apps{\identifier{f}}{t_1}{t_\maa}$ and $s' =
  \apps{t_i}{u_1}{u_n}$ for $i \in \Acc(\identifier{f})$ and
  $J$-computable $u_1,\dots,u_n$.
  We can write $t_i : \ext{\atype_1} \arrtype \dots \arrtype \atype_n \arrtype
  \bsort$ and since $i \in \Acc(\identifier{f})$ we have $\asort
  \greqsort \bsort$ and \ext{$\asort \gracsortdown \atype_j$ for each $j$}.
  By \refLemma{lem:gracsortinterplay} then each $u_j$ is also
  $I$-computable, so also $s \accreduce{I} s_1$.
\item If $s \comprel^* \apps{\identifier{f}}{s_1}{s_\maa}$ for some
  symbol $\identifier{f}$ then for all $i \in \Acc(\identifier{f})$:
  $t_i$ is $J$-computable.  But this is obvious: as $s \in F_\asort(I)$,
  we know that such $t_i$ are $I$-computable, and since $\ext{\asort
  \gracsortup \atype_i}$ for $i \in \Acc(\identifier{f})$,
  \refLemma{lem:gracsortinterplay} provides $J$-computability.
\end{itemize}
Thus, $F$ is a monotone function on a lattice; by Tarski's fixpoint
theorem there is a fixpoint, so an RC-set $I$ such that for all sorts
$\bsort$:
\begin{itemize}
\item if $\asort \grsort \bsort$ then $I(\bsort) = A_\bsort$;
\item if $\asort \eqsort \bsort$ then $I(\bsort) = \{ s \in
  \Terms(\F,\V,\arity) \mid s : \bsort \wedge s$ is terminating under
  $\comprel \cup \accreduce{I} \wedge$ if $s \comprel^*
  \apps{\identifier{f}}{s_1}{s_\maa}$ for a symbol $\identifier{f}$ then
  $\forall i \in \Acc(\identifier{f})\ [t_i$ is $I$-computable$] \}$
\end{itemize}
We define $A_\bsort := I(\bsort)$ for all $\bsort \eqsort \asort$.

Now we let $C := \bigcup_{\asort \in \Sorts} A_\asort$.  Clearly, $C$
satisfies the requirement in the theorem.
\end{proof}

\refThm{thm:defC2} easily gives the proof of \refThm{thm:defC}
in the text:

\begin{proof}[Proof of \refThm{thm:defC}]
\refThm{thm:defC} follows by taking $\arr{\Rules}$ for $\comprel$\linebreak
(which satisfies the requirements by \refLemma{lem:rewriterelationsuffices}) 
and taking for each
$\Acc(\afun)$ the maximum set $\{ i \mid 1 \leq i \leq \maa \wedge
 \ext{\asort \gracsortup \atype_i} \}$.
\end{proof}

\subsection{Additional properties of computable terms}\label{subsec:compprop}

For reasoning about computable terms (as we will do when defining
static DPs and reasoning about computable chains), there are a number
of properties besides those in \refLemma{lem:compresults} that will
prove very useful to have.  In the following, we fix the RC-set $C$ as
obtained from \refThm{thm:defC2}.

\begin{lemma}\label{lem:preservecomp}
If $s$ is $C$-computable and $s \comprel t$, then $t$ is
also
$C$-computable.
\end{lemma}

\begin{proof}
By induction on the type of $s$.  If $s$ has base type, then
$C$-computability implies that $s \in C$, and following the definition
in \refThm{thm:defC2} all reducts of $s$ are also in $C$.
Otherwise, $s : \atype \arrtype \btype$ and computability of $s$ implies
computability of $s\ u$ for all computable $u : \atype$.  By the
induction hypothesis, the fact that $s\ u \comprel t\ u$ by
monotonicity of $\comprel$ implies that $t\ u$ is computable for
all computable $u$, and therefore by definition $t$ is computable.
\end{proof}

\ext{Thus, computability is preserved under $\comprel$; the following result
shows that it is also preserved under $\accreduce{C}$.}

\begin{lemma}\label{lem:preservecompaccreduce}
\ext{If $s$ is $C$-computable and $s \accreduce{C} t$, then $t$ is also
$C$-computable.}
\end{lemma}

\begin{proof}
\ext{If $s \accreduce{C} t$, then both terms have base type, so
$C$-compu\-tability is simply membership in $C$.  We have
$s = \apps{\afun}{s_1}{s_\maa}$ and $t = \apps{s_i}{t_1}{t_n}$
with each $t_j$ $C$-computable.  Since, by definition of $C$, also
$s_i$ is $C$-computable, $C$-computability of $t$ immediately follows.}
\end{proof}

\ext{Finally, we will see that $C$-computability is also preserved under
$\gracc$.  For this, we first make a more general statement,
which will also handle variables below binders (which are freed in
subterms).}

\begin{lemma}\label{lem:preservecompacchelper}
\ext{Let $s : \atype_1 \arrtype \dots \arrtype \atype_\maa \arrtype \asort$
and $t : \btype_1 \arrtype \dots \arrtype \btype_n \arrtype \bsort$ be
meta-terms, such that $s \gracc t$.
Let $\gamma$ be a substitution with $\FMV(s) \subseteq \domain(\gamma)
\subseteq \M$.}

\ext{Let $u_1 : \btype_1,\dots,u_n : \btype_n$ be $C$-computable terms, and
$\delta$ a substitution with $\domain(\delta) \subseteq \V$ such that
each $\delta(x)$ is $C$-computable, and for $t' := \apps{(t(\gamma \cup
\delta))}{u_1}{u_n}$ there is no overlap between $\FV(t')$ and the
variables bound in $s$.}

\ext{Then there exists a $C$-computable substitution $\xi$ with
$\domain(\xi) \subseteq \V$ and $C$-computable terms $v_1 : \atype_1,
\dots,v_\maa : \atype_\maa$ such that $\apps{(s(\gamma \cup
\xi))}{v_1}{v_\maa}\ (\accreduce{C} \cup \arr{\mathtt{head}\beta})^*\ 
t'$.}
\end{lemma}

\begin{proof}
\ext{We prove the lemma by induction on the derivation of $s \gracc t$.}

\ext{If $s = t$, then we are done choosing $\xi$ and $\vec{v}$
equal to $\delta$ and $\vec{u}$.}

\ext{If $s = \abs{x}{s'}$ with $x : \atype_1$ and $s' \gracc t$, then we can
safely assume $x$ to be fresh (so not occurring in $s$ or in the range
of $\gamma$).
By the induction hypothesis, there exist a computable substitution $\xi'$
and computable terms $v_2,\dots,v_\maa$ such that $\apps{(s'(\gamma
\cup \xi'))}{v_2}{v_\maa}\ (\accreduce{C} \cup \arr{\mathtt{head}
\beta})^* t'$.  We
can safely assume that $x$ does not occur in the range of $\xi'$, since
$x$ does not occur in $t'$ either.
Therefore, if we define $\xi := [x:=x] \cup [y:=\xi'(y) \mid y \in \V
\wedge y \neq x]$, we have $s'(\gamma \cup \xi') = (s'(\gamma \cup \xi))
[x:=\xi'(x)]$.  Choosing $v_1 := \xi'(x)$,
we get
$\apps{(s(\gamma \cup \xi))}{v_1}{v_\maa} \arr{\mathtt{head}\beta}
\apps{(s'(\gamma \cup \xi'))}{v_2}{v_\maa}\ (\accreduce{C} \cup
\arr{\mathtt{head}\beta})^*\linebreak t'$.}

\ext{If $s = \apps{x}{s_1}{s_j}$ for $s_i : \ctype_1 \arrtype \dots
\arrtype \ctype_{n'} \arrtype \bsort'$ with $\asort \greqsort \bsort'$
and $x \notin \FV(s_i)$ and $s_i \gracc t$, then the induction
hypothesis provides $C$-computable terms $w_1 : \ctype_1,\dots,w_{n'}
: \ctype_{n'}$ and a substitution $\xi'$ such that $\apps{(s_i(\gamma
\cup \xi'))}{w_1}{w_{n'}}\ (\accreduce{C} \cup \arr{\mathtt{head}\beta}
)^*\ t'$.  Since $x \notin \FV(s_i)$ we can safely assume that $x
\notin \domain(\xi')$.
Now recall that by assumption $\F$ contains infinitely many constructors
of all types; let $\symb{c} : \bsort \arrtype \asort$ be a symbol that
does not occur anywhere in $\Rules$.  We can safely assume that
$\Acc(\symb{c}) = \{ 1 \}$.  Then $w := \abs{y_1 \dots y_j z_1 \dots
z_\maa}{\symb{c}\ (\apps{y_i}{w_1}{w_{n'}})}$ is $C$-computable.
Now let $\xi := [x:=w] \cup [y:=\xi'(y) \mid y \in \V \wedge y \neq x]$,
and let $v_1,\dots,v_\maa$ be variables (which are $C$-computable by
Lemma~\ref{lem:compresults}(\ref{lem:compresults:vars})).
Then $\apps{(s(\gamma \cup \xi))}{v_1}{v_\maa} \arr{\mathtt{head}\beta}^{
j+\maa} \apps{s_i(\gamma \cup \xi')}{w_1}{w_n'}\ (\accreduce{C} \cup
\arr{\mathtt{head}\beta})^*\ t'$.}

\ext{Otherwise, $s = \apps{\identifier{f}}{s_1}{s_n}$ and $s_i \gracc t$
for some $i \in \Acc(\identifier{f})$; by the induction hypothesis
there exist $\xi$ and $C$-computable terms $w_1,\dots,w_{n'}$ such that
$s' := \apps{(s_i(\gamma \cup \xi))}{w_1}{w_{n'}}\ (\accreduce{C} \cup
\arr{\mathtt{head}\beta})^*\linebreak t'$.
We have $\apps{(s(\gamma \cup \xi))}{v_1}{v_\maa} \accreduce{C} s'$
for any $\vec{v}$ (e.g., variables).}
\end{proof}

\ext{From this we conclude:}

\begin{lemma}\label{lem:preservecompacc}
\ext{Let $s$ be a closed meta-term, $\gamma$ a substitution with
$\FMV(s) \subseteq \domain(\gamma) \subseteq \M$ and $t$ such that
$s \gracc t$ and $s\gamma$ is $C$-computable.  Then for all
substitutions $\delta$ mapping $\FV(t)$ to computable terms:
$t(\gamma \cup \delta)$ is $C$-computable.}
\end{lemma}

\begin{proof}
\ext{$t(\gamma \cup \delta)$ is $C$-computable if $\apps{(t(\gamma \cup
\delta))}{u_1}{u_n}$ is $C$-computable for all computable $u_1,\dots,
u_n$.  By \refLemma{lem:preservecompacchelper} and the fact that $s$
is closed, there exist $C$-computable terms $v_1,\dots,v_\maa$ such
that $\apps{(s\gamma)}{v_1}{v_\maa}\ (\accreduce{C} \cup \arr{\mathtt{
head}\beta})^* \apps{(t(\gamma \cup \delta))}{u_1}{u_n}$.  But
$s\gamma$ is $C$-computable, and therefore so is
$\apps{(s\gamma)}{v_1}{v_\maa}$.  Since $\accreduce{C}$ and
$\arr{\mathtt{head}\beta}$ are both computability-preserving by
Lemmas \ref{lem:preservecompaccreduce} and \ref{lem:preservecomp}
respectively (as $\arr{\mathtt{head}\beta}$ is included in $\comprel$)
we are done.}
\end{proof}

\begin{lemma}\label{lem:neutralcomp}
A neutral term is $C$-computable if and only if all its
$\comprel$-reducts are $C$-computable.
\end{lemma}

\begin{proof}
That $C$-computability of a term implies $C$-computability of its
reducts is given by \refLemma{lem:preservecomp}.  For the other
direction, let $s : \atype_1 \arrtype \dots \arrtype \atype_\maa \arrtype
\asort$ be neutral and suppose that all its reducts are
$C$-computable.  To prove that also $s$ is $C$-computable, we must see
that for all $C$-computable terms $t_1 : \atype_1,\dots,t_\maa :
\atype_\maa$ the term $u := \apps{s}{t_1}{t_\maa}$ is in $C$.  We prove
this by induction on $(t_1,\dots,t_\maa)$ ordered by
$\comprel_{\mathtt{prod}}$.
Clearly, since $s$ does not have the form $\apps{\afun}{s_1}{s_n}$ with
$\Acc(\afun) \neq \emptyset$, nor does $u$, so $u \in C$ if all its
reducts are in $C$.  But since $s$ is neutral, all reducts of
$u$ either have the form $\apps{s'}{t_1}{t_\maa}$ with $s \comprel s'$
-- which is in $C$ because all $t_i$ are $C$-computable and $s'$ is
computable as a reduct of $s$ -- or the form $\apps{s}{t_1}{t_i'} \cdots
t_\maa$ with $t_i \comprel t_i'$ -- which is in $C$ by the induction
hypothesis.
\end{proof}

Using the $\arr{\mathtt{head}\beta}$-restrictions on $\comprel$, we
obtain the following result:

\begin{lemma}\label{lem:abscomputable}
Let $x : \atype \in \V$.
A term $\abs{x}{s}$ is $C$-computable if and only if $s[x:=t]$ is
computable for all $C$-computable $t : \atype$.
\end{lemma}

\begin{proof}
If $\abs{x}{s}$ is $C$-computable, then by definition so is
$(\abs{x}{s})\ t$ for all $C$-computable $t$; by
\refLemma{lem:preservecomp} and inclusion of $\arr{\mathtt{head}\beta}$
in $\comprel$, this implies $C$-computability of the reducts $s[x:=t]$.

For the other direction, suppose $s[x:=t]$ is $C$-computable for all
$C$-computable $t : \atype$.
To obtain $C$-computability of $\abs{x}{s}$, we must see that
$(\abs{x}{s})\ t$ is $C$-computable for all $C$-computable $t : \atype$.
%
As $(\abs{x}{s})\ t$ is neutral, this holds if all its
$\comprel$-reducts $u$ are $C$-computable by
\refLemma{lem:neutralcomp}, and certainly if all its
$\comprel^+$-reducts are $C$-computable, which we prove by induction
on $u$ oriented with $\comprel$.  But by definition of $\comprel$ (and
induction on the derivation $(\abs{x}{s})\ t \comprel^+ u$) there
exists a term $v$ such that $s[x:=t] \comprel^* v$ and $u
\arrr{\mathtt{head}\beta} v$.  If $u = v$ we therefore obtain the
required property, and if $u \arr{\mathtt{head}\beta}^+ v$ then $u$
is neutral and therefore is $C$-computable if all its
$\comprel$-reducts are, which is the case by the induction hypothesis.
\end{proof}

\section{Dynamic dependency pairs: the main result}\label{app:ddp}

In this appendix, we prove \refThm{thm:chain}, which states that
an AFSM $(\F,\Rules)$ is terminating if and only if it admits no (minimal,
formative) infinite $(\DDP(\Rules),\Rules)$-dependency chains.
This proof follows the same reasoning as used in
\cite[Thm.~5.7]{kop:raa:12}, but is adapted to the new setting.
The new setting also allows the completeness result,
\refLemma{lem:candcomplete}, to be obtained without requiring
left-linearity.

\medskip
We first show that the existence of any infinite $(\DDP(\Rules),
\Rules)$-chain (minimal and formative or not) shows non-termination of
the relation $\arr{\Rules}$, through a number of lemmas exploring the
relation between dependency pairs and reduction steps.

\begin{lemma}\label{lem:candidatereduce}
Let $s,t$ be meta-terms and suppose $s \bsuptermeq{A} t$ for some set
$A$ of meta-variable conditions.  Then for any substitution $\gamma$
that respects $A$ and has a finite domain with $\FMV(s) \subseteq
\domain(\gamma) \subseteq \M$: $s\gamma\ (\supterm \mathop{\cup}
\arr{\beta})^*\ t\gamma$.
\end{lemma}

\begin{proof}
By induction on the definition of $\bsuptermeq{A}$.  Consider the last
step in its derivation.
\begin{itemize}
\item If $s = \apps{t}{s_1}{s_n}$ then we have
  $s\gamma = t\gamma$ if $n = 0$
  and $s\gamma = \apps{(t\gamma)}{(s_1\gamma)}{(s_n\gamma)} \supterm
  t\gamma$ if $n > 0$.
\item If $s = \abs{x}{u}$ and $u \bsuptermeq{A} t$, then by
  $\alpha$-conversion we can assume that $x \notin \FV(\gamma(Z))$ for
  any $Z \in \FMV(s)$. Thus, $s\gamma = \abs{x}{(u\gamma)} \supterm
  u\gamma\ (\supterm \mathop{\cup} \arr{\beta})^*\ t\gamma$ by the induction
  hypothesis.
\item If $s = \apps{(\abs{x}{u})}{s_0}{s_n}$ and $\apps{u[x:=s_0]}{
  s_1}{s_n} \bsuptermeq{A} t$, then by $\alpha$-conversion we can
  safely assume that $x$ is fresh wrt $\gamma$ as above; thus,
  $s\gamma = \apps{(\abs{x}{(u\gamma)})}{
  (s_0\gamma)}{(s_n\gamma)} \arr{\beta} \apps{(u\gamma[x:=s_0\gamma])
  }{(s_1\gamma)}{(s_n\gamma)} = (\apps{u[x:=s_0]}{s_1}{s_n})\gamma$,
  which reduces to $t\gamma$ by the induction hypothesis.
\item If $s = \apps{u}{s_1}{s_n}$ and $s_i \bsuptermeq{A} t$, then
  $s\gamma = \apps{(u\gamma)}{(s_1\gamma)}{(s_n\gamma)} \supterm
  s_i\gamma\ (\supterm \mathop{\cup} \arr{\beta})^*\ t\gamma$ by the
  induction hypothesis.
\item If $s = \apps{\meta{Z}{t_1,\dots,t_\mac}}{s_1}{s_n}$ and $t_i
  \bsuptermeq{A} t$ for some $1 \leq i \leq k$ with $(Z : i) \in A$,
  then because $\gamma$ respects $A$ we have $\gamma(Z) \approxp
  \abs{x_1 \dots x_\mac}{w}$ for some $w$ with $x_i \in \FV(w)$.  Thus,
  $s\gamma = w[x_1:=t_1\gamma,\dots,x_n:=t_n\gamma] \suptermeq
  t_i\gamma$.  We again complete by the induction hypothesis.
\qedhere
\end{itemize}
\end{proof}

\begin{lemma}\label{lem:dpreduce}
For $\ell^\sharp \arrdp p^\sharp\ (A) \in \DDP(\Rules)$ and substitution
$\gamma$ on domain $\FMV(\ell)$ such that $\gamma$ respects the
meta-variable conditions in $A$: both $\ell\gamma$ and $p\gamma$ are
terms and $\ell\gamma\ (\arr{\Rules} \mathop{\cup} \supterm)^+\ p\gamma$.
\end{lemma}

\begin{proof}
By definition of $\DDP$, there is a rule $\ell' \arrz r$ such that
$\ell = \apps{\ell'}{Z_1}{Z_i}$ and $p\ (A) \in \cand(\apps{r}{Z_1}{
Z_i})$, so $\apps{r}{Z_1}{Z_i} \bsuptermeq{A} p$.
Clearly, we have $\ell\gamma \arr{\Rules} (r\ \vec{Z})\gamma$ by that
rule (applied at the head), and $(\apps{r}{Z_1}{Z_i})\gamma\ (\supterm
\mathop{\cup} \arr{\Rules})^*\ p\gamma$ by \refLemma{lem:candidatereduce}
($\domain(\gamma)$ contains all meta-variables in $\apps{r}{Z_1}{Z_i}$
because $\FMV(r) \subseteq \FMV(\ell')$).  We are done because
$\arr{\beta}$ is included in $\arr{\Rules}$.
\end{proof}

\begin{lemma}\label{lem:chain:complete}
If there is an infinite
$(\DDP(\Rules),\Rules)$-dependency chain starting in
$s^\sharp$, then $s$ is non-terminating under $\arr{\Rules}$.
\end{lemma}

\begin{proof}
Let $\unsharp{s_i}, \unsharp{t_i}$ denote the terms $s_i,t_i$ with all $\sharp$ marks
removed.
An
infinite
$(\DDP(\Rules),\Rules)$-dependency chain provides a sequence $(s_i,
t_i)$ for $i \in \N$ such that for all $i$, $\unsharp{s_i}\ (\arr{\Rules} 
\mathop{\cup} \supterm)^+\ \unsharp{t_i}$ (either because $\arr{\beta}$ is included in
$\arr{\Rules}$ or by \refLemma{lem:dpreduce}), and $\unsharp{t_i}
\arrr{\Rules} \unsharp{s_{i+1}}$.  Thus, we obtain an infinite $\arr{\Rules}
\mathop{\cup} \supterm$ sequence, which provides an infinite $\arr{\Rules}$
sequence due to monotonicity of $\arr{\Rules}$.
\end{proof}

Thus, if we can prove the existence of an \ext{infinite} $(\DDP(\Rules),
\Rules)$-dependency chain, we obtain non-termination of $\arr{\Rules}$ by
Lemma \ref{lem:chain:complete}.
Now let us consider the other direction.
We start once more by considering the relation $\bsuptermeq{A}$.

\begin{lemma}\label{lem:candcomplete}
Let $s$ be a meta-term and $\gamma$ a substitution on a finite domain
with $\FMV(s) \subseteq \domain(\gamma) \subseteq \M$, such that all
$\gamma(Z)$ are terminating.
If $s\gamma$ is non-terminating, then there exists a pair $t\ (A) \in
\cand(s)$ such that $t\gamma$ is non-terminating, $\gamma$ respects $A$,
and $t'\gamma$ is terminating for all $t' \neq t$ such that $t
\bsuptermeq{B} t'$ for some $B$ respected by $\gamma$.
\end{lemma}

\begin{proof}
Let $S$ be the set of all pairs $t\ (A)$ such that (a) $s
\bsuptermeq{A} t$, (b) $t\gamma$ is non-terminating, and (c)
$\gamma$ respects $A$.  As the strict parts of the relations
$\bsuptermeq{\beta}$ and $\supseteq$ are both well-founded
orderings,
we can select a
pair $t\ (A)$ that is \emph{minimal} in $S$: for all $t'\ (A') \in
S$: if $t \bsuptermeq{\beta} t'$ then $t = t'$ and $A \subseteq A'$.
We observe that for all $t',B$ such that $t' \neq t$ and $t
\bsuptermeq{B} t'$ and $\gamma$ respects $B$ we cannot have $t'\ (A
\cup B) \in S$ by minimality of $t\ (A)$, so since $s \bsuptermeq{
\beta} t \bsuptermeq{\beta} t'$ and clearly $\gamma$ respects $A \cup
B$, we must have termination of $t'\gamma$.

Thus, the lemma holds if $t\ (A) \in \cand(s)$.  By minimality of $A$,
this is the case if $t$ has the form $\apps{\identifier{f}}{t_1}{t_n}$
with $n \geq \arity(\identifier{f})$, or the form $\apps{\meta{Z}{
t_1,\dots,t_\mac}}{s_1}{s_n}$ with $n > 0$ or $t_1,\dots,t_\mac$ not
distinct variables or $\mac > \arity(Z)$.  Thus, we will see that if
$t$ does \emph{not} have one of these forms, then $t$ is not minimal.
Consider the form of $t$:
\begin{itemize}
\item $t = \abs{x}{t'}$: non-termination of $t\gamma$ implies
  non-termination of $t'\gamma$, and $s \bsuptermeq{A} t \bsuptermeq{A}
  t'$;
\item $t = \apps{a}{t_1}{t_n}$ with $a \in \V \cup (\F \setminus
  \Defineds)$: since any reduction of $s$ must take place in some $t_i$
  (as the domain of $\gamma$ does not contain variables),
  non-termination of $t\gamma$ implies non-termination of some $t_i
  \gamma$, and $s \bsuptermeq{A} t \bsuptermeq{A} t_i$;
\item $t = \apps{\identifier{f}}{t_1}{t_n}$ with $\identifier{f} \in
  \Defineds$ but $n < \arity(\identifier{f})$: same as above, because
  the rules respect $\arity$;
\item $t = \apps{(\abs{x}{u})}{t_0}{t_n}$: if $t_0\gamma$ is
  non-terminating, we are done because $s \bsuptermeq{A} t
  \bsuptermeq{A} t_0$; if not, then $t'\gamma$ is non-terminating for
  $t' := \apps{u[x:=t_0]}{t_1}{t_n}$ (both an instance of $u\gamma$
  and all $t_i\gamma$ for $i \neq 0$ are subterms of $t\gamma$), and
  $s \bsuptermeq{A} t \bsuptermeq{A} t'$
\item $t = \meta{Z}{x_1,\dots,x_\mia}$ with all $x_i$ distinct variables
  and $\mia$ is the minimal arity of $Z$: by $\alpha$-conversion, we can
  write $\gamma(Z) = \abs{x_1 \dots x_\mia}{u}$, and since $\gamma(Z)$
  is terminating by assumption, so is $u = t\gamma$; contradiction
  with $t\ (A)$ being 
  in
  $S$.
\qedhere
\end{itemize}
\end{proof}

Let us now
consider
formative reductions.  We will prove that
reductions from a terminating term to some instance of a pattern may
be assumed to be formative.

\begin{lemma}\label{lem:formative}
Let $\ell$ be a pattern and $\gamma$ a substitution on domain
$\FMV(\ell)$.
Let $s$ be a terminating term.
If $s \arrr{\Rules} \ell\gamma$, then there exists a substitution
$\delta$ on the same domain as $\gamma$ such that each $\delta(Z)
\arrr{\Rules} \gamma(Z)$ and $s \arrr{\Rules} \ell\delta$ by an
$\ell$-formative reduction.
\end{lemma}

\begin{proof}
We prove the lemma by induction first on $s$ ordered by $\arr{\Rules}
\mathop{\cup} \supterm$, second on the length of the reduction $s \arrr{\Rules}
\ell\gamma$.  First observe that if $\ell$ is not a fully extended
linear pattern, then we are done choosing $\delta := \gamma$.
Otherwise, we consider four cases:
\begin{enumerate}
\item $\ell$ is a meta-variable application $\meta{Z}{x_1,\dots,x_\mia}$;
\item $\ell$ is not a meta-variable application, and the reduction
  $s \arrr{\Rules} \ell\gamma$ does not contain any headmost steps;
\item $\ell$ is not a meta-variable application, and the reduction
  $s \arrr{\Rules} \ell\gamma$ contains headmost steps, the first of
  which is a $\arr{\beta}$ step;
\item $\ell$ is not a meta-variable application, and the reduction
  $s \arrr{\Rules} \ell\gamma$ contains headmost steps, the first of
  which is not a $\arr{\beta}$ step.
\end{enumerate}

In the first case, if $\ell$ is a meta-variable application
$\meta{Z}{x_1,\dots,x_\mia}$, then by $\alpha$-conversion we may
write $\gamma = [Z:=\abs{x_1\dots x_\mia}{t}]$ with $\ell\gamma = t$.
Let $\delta$ be $[Z := \abs{x_1 \dots x_\mia}{s}]$.  Then $\delta$
has the same domain as $\gamma$, and indeed $\delta(Z) = \abs{x_1
\dots x_\mia}{s} \arrr{\Rules} \abs{x_1 \dots x_\mia}{(\ell\gamma)} =
\gamma(Z)$.

In the second case, a reduction without any headmost steps, we observe
that $s$ has the same outer shape as $\ell$: either (a) $s = \abs{x}{s'}$
and $\ell = \abs{x}{\ell'}$, or (b) $s = \apps{a}{s_1}{s_n}$ and $\ell =
\apps{a}{\ell_1}{\ell_n}$ for some $a \in \V \cup \F$ (since $\ell$ is
a pattern, $a$ cannot be a meta-variable application or abstraction if
$n > 0$).
In case (a), we obtain $\delta$ such that $s' \arrr{\Rules}
\ell'\delta$ by an $\ell'$-formative reduction and $\delta
\arrr{\Rules} \gamma$ by the induction hypothesis (as sub-meta-terms of
linear patterns are still linear patterns).  In case (b), we let
$\gamma_i$ be the restriction of $\gamma$ to $\FMV(\ell_i)$ for $1
\leq i \leq n$; by linearity of $\ell$, all $\gamma_i$ have
non-overlapping domains and $\gamma = \gamma_1 \cup \dots \cup
\gamma_n$.  The induction hypothesis provides $\delta_1,\dots,
\delta_n$ on the same domains such that each $s_i \arrr{\Rules}
\ell_i\delta_i$ by an $\ell_i$-formative reduction and $\delta_i
\arrr{\Rules} \gamma_i$; we are done choosing $\delta := \delta_1 \cup
\dots \cup \delta_n$.

In the third case, if the first headmost step is a
$\beta$-step,
note that $s$ must have the form $\apps{(\abs{x}{t})\ u}{q_1}{q_n}$,
and moreover $s \arrr{\Rules} \apps{(\abs{x}{t'})\ u'}{q_1'}{q_n'} \arr{\beta}
\apps{t'[x:=u']}{q_1'}{q_n'} \arrr{\Rules} \ell\gamma$ by steps in the
respective subterms.  But then also $s \arr{\beta} \apps{t[x:=u]}{
q_1}{q_n} \arrr{\Rules} \apps{t'[x:=u']}{q_1'}{q_n'} \arrr{\Rules}
\ell\gamma$, and we can
get
$\delta$ and an $\ell$-formative
reduction for $\apps{t[x:=u]}{q_1}{q_n} \arrr{\Rules} \ell\delta$ by
the induction hypothesis.

In the last case, if the first headmost step is not a
$\beta$-step, then we can write $s = \apps{\afun}{s_1}{s_n}
\arrr{\Rules} \apps{\afun}{s_1'}{s_n'} = \apps{(\ell'\eta)}{s_{i+1}'
}{s_n'} \arr{\Rules} \apps{(r\eta)}{s_{i+1}'}{s_n'} \arrr{\Rules}
\ell\gamma$ for some $\afun \in \Defineds$, terms
$s_j \arrr{\Rules} s_j'$ for $1 \leq j \leq n$,
rule $\ell' \arrz r$
and substitution $\eta$ on domain $\FMV(\ell')$.  But then
$\apps{\ell'}{Z_{i+1}}{Z_n} \arrz \apps{r}{Z_{i+1}}{Z_n} \in
\RulesEta$, and for $\eta' := \eta \cup [Z_{i+1}:=s_{i+1}',
\dots,Z_n:=s_n']$ we both have $s \arrr{\Rules} (\apps{\ell'}{Z_{i+1}
}{Z_n})\eta'$ without any headmost steps, and $(\apps{r}{Z_{i+1}
}{Z_n})\eta' \arrr{\Rules} \ell\gamma$.
By the second induction hypothesis, there exists a substitution
$\xi$ such that $s \arrr{\Rules} (\apps{\ell'}{Z_{i+1}}{Z_n})\xi$ by
a $(\apps{\ell'}{Z_{i+1}}{Z_n})$-formative reduction and $\xi
\arrr{\Rules} \eta'$.  This gives $s \arr{\Rules}^+
(\apps{r}{Z_{i+1}}{Z_n})\xi \arrr{\Rules} (\apps{r}{Z_{i+1}}{Z_n})
\eta' \arrr{\Rules} \ell\gamma$, so by the first induction
hypothesis we obtain $\delta$ such that $(\apps{r}{Z_{i+1}}{Z_n})\xi
\arrr{\Rules} \ell\delta$ by an $\ell$-formative reduction, and
$\delta \arrr{\Rules} \gamma$.
\end{proof}

Essentially, \refLemma{lem:formative} states that we can postpone
reductions that are not needed to obtain an instance of the given
pattern.  This is not overly surprising, but will help us eliminate
some proof obligations later in the termination proof.

At last, we can prove the first part of \refThm{thm:chain},
soundness.

\begin{lemma}\label{lem:chain:sound}
If $\Rules$ is non-terminating, then there is a minimal formative
$(\DDP(\Rules),\Rules)$-dependency chain.
\end{lemma}

\begin{proof}
Suppose $\Rules$ is non-terminating; then there exists a non-terminating
term $t_{-1}$ of which all strict subterms are terminating.
Now for $i \in \N$, suppose $t_{i-1}$ is a non-terminating term of which
all strict subterms are terminating.  We also assume (which we can do
safely for $t_{-1}$ by using a renaming) that $t_{i-1}$ does not contain
any variable that occurs freely in the right-hand side of any element
of $\DDP(\Rules)$.  We will now show how to select terms $s_i,t_i$ and
$\rho_i \in \DDP(\Rules) \cup \{\mathtt{beta}\}$ such that the sequence
$[(\rho_j,s_j^\sharp,t_j^\sharp) \mid j \in \N]$ is a dependency chain.
For this, consider the shape of $t_{i-1}$.

Clearly $t_{i-1}$ cannot have the form $\apps{x}{q_1}{q_n}$ or $\apps{
\identifier{c}}{q_1}{q_n}$ with $\identifier{c} \in \F \setminus
\Defineds$, since every reduction of such a term takes place in a
strict subterm and leaves the top shape intact; thus, we would obtain
non-termination of some $q_j$, contradicting minimality of $t_{i-1}$.
Similarly, $t_{i-1}$ cannot have the form $\abs{x}{q}$, as this would
give non-termination of $q$.

If $t_{i-1} = \apps{(\abs{x}{s})\ u}{q_1}{q_n}$, then by termination of
$s$, $u$ and each $q_j$, any infinite reduction starting in $t_{i-1}$
must eventually take a head-step:
\[
t_{i-1} \arrr{\Rules}
\apps{(\abs{x}{s'})\ u'}{q_1'}{q_n'} \arr{\beta} \apps{s'[x:=u']}{
q_1'}{q_n'} \arrr{\Rules} \dots
\]
Writing $t' := \apps{s[x:=u]}{q_1}{q_n}$ we then also have $t_{i-1}
\arr{\beta} t' \arrr{\Rules} \apps{s'[x:=u']}{q_1'}{q_n'} \arrr{\Rules}
\dots$, so $t'$ is still non-terminating.  We let $s_i := t_{i-1}$ and
$\rho_i := \mathtt{beta}$, and let $t_i$ be an arbitrary
non-terminating subterm of $t'$ that is minimal.

(We observe that, in this case, $t_{i-1}^\sharp = t_{i-1}$ and
$s_i^\sharp = s_i$.  Moreover, if $n > 0$ then $t_i^\sharp = t_i = t'$,
since in that case every strict subterm of $t'$ is also a strict
subterm of $t_{i-1}$ or a reduct thereof; thus, $t'$ obtains minimality
from the
minimality of $t_{i-1}$.  If $n = 0$, then $t' = s[x:=u]$ and since
both $s$ and $u$ are terminating by minimality of $t_{i-1}$,
necessarily $t_i = v[x:=u]$ for some subterm $v$ of $s$ that contains
the variable $x$.  It is easy to see that $t_i^\sharp =
v^\sharp[x:=u]$.)

The only remaining possibility is that $t_{i-1} = \apps{\afun}{q_1}{q_n}$
with $n \geq \arity(\afun)$ and there exist a rule $\ell \arrz r$, a
substitution $\gamma$ and an integer $j \leq n$ such that $t_{i-1}
\arrr{\Rules} \apps{\afun}{q_1'}{q_n'} = \apps{(\ell\gamma)}{
q_{j+1}'}{q_n'}$ and $\apps{(r\gamma)}{q_{j+1}'}{q_n'}$ is still
non-terminating.  Letting $\ell' = \apps{\ell}{Z_{j+1}}{Z_n}$ and
$r' = \apps{r}{Z_{j+1}}{Z_n}$ and $\gamma' = \gamma \cup [Z_{j+1}:=
q_{j+1}',\dots,Z_n:=q_n']$ we have $t_{i-1} \arrr{\Rules} \ell'\gamma'$
and $r'\gamma'$ non-terminating; by \refLemma{lem:formative} we find
a substitution $\delta$ such that $t_{i-1} \arrr{\Rules} \ell'\delta$ by
an $\ell'$-formative reduction and $r'\delta \arrr{\Rules} r'\gamma'$
is still non-terminating.  By \refLemma{lem:candcomplete}, we find
a minimal candidate $p\ (A)$ of $r'$ such that $p\delta$ is still
non-terminating and $\delta$ respects $A$.  Since the domain of
$\delta$ does not contain variables, these are left alone, and by the
assumption on $t_{i-1}$, they do not occur in any $\delta(Z)$; thus,
if we extend $\delta$ with a mapping from each of these variables to a
fresh variable of the same type, also $p\delta'$ is non-terminating.
We let $s_i := \ell'\delta = \ell'\delta'$ and $\rho_i :=
{\ell'}^\sharp \arrdp p^\sharp\ (A)$, and let $t_i$ be an arbitrary
non-terminating subterm of $p\delta'$ that is minimal.

(We observe that $t_{i-1}^\sharp$ reduces to $s_i^\sharp$ by reductions
inside the subterms $q_j$, and $s_i^\sharp = {\ell'}^\sharp\delta$.
Also, $\rho_i$ is clearly an element of $\DDP(\Rules)$, and by the
choice of $p\ (A)$, $\delta'$ respects $A$.  Since $p$ does not contain
any meta-variables not already occurring in $\FMV(r') \subseteq \FMV(
\ell')$, the domain of $\delta'$ is $\FMV({\ell'}^\sharp) \cup \FMV(
p^\sharp) \cup \FV(p^\sharp)$, and $\delta'$ maps all variables in
$p^\sharp$ to fresh variables.  Moreover, if $p$ has the form
$\apps{\afun}{p_1}{p_n}$, then all $p_j\delta'$ are terminating (by
the minimality condition of \refLemma{lem:candcomplete}), so then
$t_{i+1} = p\delta'$.  Similarly, if $p$ has the form
$\apps{\meta{Z}{u_1,\dots,u_\mac}}{p_1}{p_n}$ with $n > 0$ then both
all $p_j\delta'$ and $(\apps{\meta{Z}{u_1,\dots,u_k}}{p_1}{p_j})
\delta'$ for $j < n$ are all terminating.  So $p\delta' \supterm t_i$
can only hold if $p = \meta{Z}{u_1,\dots,u_\mac}$.  Writing $\delta'(Z
) = \abs{x_1 \dots x_n}{w}$ with $w$ not an abstraction, we have
$\delta'(Z) \approxp \abs{x_1 \dots x_\mac}{\apps{w}{x_{n+1} \linebreak
}{x_\mac}}$.
We know that $w$ itself is terminating, as is $u_j\delta'$ for every
$j$ such that $x_j \in \FV(\apps{w}{x_{n+1}}{x_\mac})$ -- the latter
by minimality of the candidate $p\ (A)$.  Thus, a minimal
non-terminating subterm of $(\apps{w}{x_{n+1}}{x_\mac})[x_1:=u_1
\delta',\dots,x_\mac:=u_\mac\delta']$ can only have the form $v[x_1:=
u_1\delta',\dots,x_k:=u_k\delta']$ for some subterm $v$ of $\apps{w}{
x_{n+1}}{x_\mac}$ in which at least one $x_j$ occurs.  It is also
clear that $t_i^\sharp = v^\sharp[x_1:=u_1\delta',\dots,x_k:=u_k
\delta']$ as $v$ is not a variable.)

Thus, we have constructed a $(\DDP(\Rules),\Rules)$-dependency chain
$[(\rho_i,s_i^\sharp,t_i^\sharp) \mid i \in \N]$.
\end{proof}

\oldcounter{\thmDdpChain}{\thmDdpChainSec}
\ddpChainTheThm
\startappendixcounters

\begin{proof}
This is the combination of Lemmas~\ref{lem:chain:complete}
and~\ref{lem:chain:sound}.
\end{proof}

In \refDef{def:ddp}, we have -- both for reasons of space and to keep
the definition simple -- included more dependency
pairs than strictly necessary.  In particular, we could improve the
definition much like Dershowitz' refinement by excluding all pairs
$(\apps{\ell}{Z_1}{Z_i})^\sharp \arrdp p^\sharp\ (A)$ if
$\DDP(\Rules)$ also contains a pair $\ell^\sharp
\arrdp p^\sharp\ (A)$ -- so most of those DPs generated from rules in
$\RulesEta \setminus \Rules$.
This is done in~\cite{kop:12}.

\begin{lemma}\label{lem:omitstupid}
Let $\P \subseteq \DDP(\Rules)$ be a set of dependency pairs, and
$\P' \subseteq \P$ be such that all pairs in $\P'$ have the form
$(\apps{\ell}{Z_1}{Z_i})^\sharp \arrdp p^\sharp\ (A)$ with $i > 0$
and $\ell^\sharp \arrdp p^\sharp\ (A) \in \P \setminus \P'$.  Then
for every minimal $(\P,\Rules)$-chain, all $\rho_i$ with $i > 1$ are
in $\P \setminus \P'$.
\end{lemma}

\begin{proof}
For $i > 1$, suppose $\rho_i \in \P'$, so $\rho_i =
(\apps{\ell}{Z_1}{Z_j})^\sharp \arrdp p^\sharp\ (A)$ with $j > 1$;
since $\ell^\sharp$ is the left-hand side of a DP, $\ell$
has the form $\apps{\afun}{\ell_1}{\ell_n}$ with $n \geq
\arity(\afun)$.
Therefore, the $\sharp$ has no effect, so $\rho_i$ is actually
$\apps{\ell}{Z_1}{Z_j} \arrdp p^\sharp\ (A)$.  Note that $s_i$ has
the form $\apps{\afun}{u_1}{u_{n+j}}$ and $t_{i-1}$ therefore has the
form $\apps{\afun}{u_1'}{u_{n+j}'}$, with each $u_m \arrr{\Rules}
u_m'$.
By minimality of $t_{i-1}$, clearly $\apps{\afun}{u_1'}{u_n'}$ is
terminating, and therefore so is its reduct $\apps{\afun}{u_1}{u_n}$.
However, the infinite $(\DDP(\Rules),\Rules)$-dependency chain
\[
[(\ell^\sharp \arrz p^\sharp\ (A),(\apps{\afun}{u_1}{u_n})^\sharp,
t_i)] \cdot [(\rho_j,s_j,t_j) \mid j \in \N \setminus \{0,\dots,i\}]
\]
(where ``$\cdot$'' denotes composition)
starts in the term $(\apps{\afun}{u_1}{u_n})^\sharp$.
\refLemma{lem:chain:complete} provides non-termination of this term;
contradiction.
\end{proof}

Thus, simply removing the first pair in the $(\DDP(\Rules),
\Rules)$-chain provides a sequence that does not use these
unnecessary pairs.

  We could either choose not to include these unnecessary DPs from the
  start, or remove them afterwards using the following DP processor:
  \removeUselessTheProc
  
  \begin{proof}
  Completeness will follow by \refLemma{lem:subsetcomplete}.
  Soundness follows by \refLemma{lem:omitstupid}.
  \end{proof}

  The advantages of using this processor over employing a definition of
  $\DDP$ that excludes such DPs from the start are: (1) it simplifies the
  definition, and (2)
  one may devise processors whose soundness or completeness
  relies on some set $\P'$ being contained in $\DDP(\Rules)$,
  as in the case for
  $\Proc_{\SDP_G}$.  Including more DPs in the initial definition of
  $\DDP(\Rules)$ may increase the applicability of such processors.

\section{Static dependency pairs}

In this appendix, we will first prove the main result from
\refSec{sec:static}.  Then, to provide a greater context to the
current work, we will discuss how the definitions in
\cite{kus:iso:sak:bla:09,suz:kus:bla:11} relate to the definitions
here.

\subsection{Static dependency pairs: the main result}\label{app:sdp}

In this appendix, we prove \refThm{thm:sschain}, which states that an
accessible function passing AFSM with rules $\Rules$ is terminating if
it admits no $\Rules$-computable formative $(\SDP(\Rules),\Rules)
$-dependency chains.
\refThm{thm:schain}, which states that an AFSM is terminating if
it admits no minimal formative $(\SDP(\Rules),\Rules)$-dependency
chains, follows as a corollary.

In the following, let $C = C_{\Rules}$ be a computability predicate
following \ref{thm:defC2} for $\comprel$ the rewrite relation
$\arr{\Rules}$; we will briefly call a term ``computable'' if it is
$C$-computable.
We start with an observation on the consequences of accessible
function passingness:

\begin{lemma}\label{lem:pfp}
Let $\ell$ be a closed pattern, $Z$ a meta-variable and $x_1,\dots,
x_\mia$ variables such that $\ell \gracc \meta{Z}{x_1,\dots,x_\mia}$.
If $\ell\gamma$ is a computable term, then so is $\gamma(Z)$.
\end{lemma}

\begin{proof}
Since $\ell$ is closed, $\ell(\gamma \cup \delta) = \ell\gamma$ is
computable for all computable substitutions $\delta$ whose domain is
contained in $\V$.  By \refLemma{lem:preservecompacc}, we
thus have computability of $\meta{Z}{x_1,\dots,x_\mia}(\gamma
\cup \delta)$ for all such $\delta$.
Since $\ell\gamma$ is a term, $Z \in \domain(\gamma)$ so
$\gamma(Z) = \abs{x_1 \dots x_\mia}{s}$.  If we let
$\delta := [x_1:=u_1,\dots,x_\mia:=u_\mia]$ for computable terms
$u_1,\dots,u_\mia$ we have computability of $\meta{Z}{x_1,\dots,
x_\mia}(\gamma \cup \delta) = s[x_1:=u_1,\dots,x_n:=u_\mia]$.  Since
this holds for \emph{all} computable $u_1,\dots,u_\mia$,
\refLemma{lem:abscomputable} provides computability of $\abs{x_1
\dots x_\mia}{s} = \gamma(Z)$.
\end{proof}

We continue with a variation of \refLemma{lem:candcomplete}:

\begin{lemma}\label{lem:scandcomplete}
Assume that the minimal arity equals the maximal arity for all
$\afun \in \F$.

Let $s$ be a meta-term and $\gamma$ a substitution on a finite domain
with $\FMV(s) \subseteq \domain(\gamma) \subseteq \M$, such that all
$\gamma(Z)$ are computable.
If there exists a computable substitution $\delta$ on a variable
domain (that is, $\domain(\gamma) \subseteq \V$) such that $s(\gamma
\cup \delta)$ is not computable, then there exists a pair $t\ (A) \in
\cand(s)$ such that all of the following hold:
\begin{itemize}
\item there is a computable substitution $\delta$ on variable
  domain such that $t(\gamma \cup \delta)$ is not computable;
\item $\gamma$ respects $A$;
\item for all $t' \neq t$ such that $t \bsuptermeq{B} t'$ holds for
  some $B$ respected by $\gamma$: $t'(\gamma \cup \delta)$ is
  computable for all computable substitutions $\delta$ on variable
  domain;
\item $t$ has the form $\apps{\afun}{t_1}{t_\maa}$.
\end{itemize}
\end{lemma}

\vspace{-4pt}
\begin{proof}
Let $S$ be the set of all pairs $t\ (A)$ such that (a) $s
\bsuptermeq{A} t$, (b) there exists a computable substitution $\delta$
on variable domain such that $t(\gamma \cup \delta)$ is
not computable, and (c) $\gamma$ respects $A$.  As the relations
$\bsuptermeq{\beta}$ and $\supseteq$ are both well-founded
quasi-orderings, we can select a pair $t\ (A)$ that is
\emph{minimal} in $S$: for all $t'\ (A') \in S$: if $t
\bsuptermeq{\beta} t'$ then $t = t'$ and $A \subseteq A'$.
We observe that for all $t',B$ such that $t' \neq t$ and $t
\bsuptermeq{B} t'$ and $\gamma$ respects $B$ we cannot have $t'\ (A
\cup B) \in S$ by minimality of $t\ (A)$, so since $s \bsuptermeq{
A \cup B} t \bsuptermeq{A \cup B} t'$ and clearly $\gamma$ respects
$A \cup B$, it can only follow that requirement (b) is not satisfied
for $t'$.

By the above reasoning and minimality of $A$, the lemma holds if $t$
has the form $\apps{\afun}{t_1}{t_\maa}$ with $\maa = \arity(\afun)$.
Thus, we will see that if $t$ does \emph{not} have this form, then
$t$ is not minimal.
Consider the form of $t$:
\begin{itemize}
\item $t = \abs{x}{t'}$: by \refLemma{lem:abscomputable},
  non-computability of $t(\gamma \cup \delta)$ implies
  non-computability of $t'(\gamma \cup \delta)[x:=u]$ for some
  computable $u$.  Since, by $\alpha$-conversion, we can assume that
  $x$ does not occur in domain or range of $\gamma$ or $\delta$,
  we have non-computability of $t'(\gamma \cup \delta \cup [x:=u])$,
  and $\delta \cup [x:=u]$ is a computable substitution on
  variable domain while $t \bsuptermeq{A} t'$.
\item $t = \apps{x}{t_1}{t_n}$ with $x \in \V$: whether $x \in
  \domain(\delta)$ or not, $\delta(x)$ is computable (either by the
  assumption on $\delta$ or by
  \refLemma{lem:compresults}(\ref{lem:compresults:vars})).  Therefore,
  the only way for $t(\gamma \cup \delta)$ to not be computable is if
  some $s_i(\gamma \cup \delta)$ is not computable, and $s
  \bsuptermeq{A} s_i$.
\item $t = \apps{\identifier{c}}{t_1}{t_\maa}$ with $\identifier{c}
  \in \F \setminus \Defineds$: $t(\gamma \cup \delta)$ is
  non-computable only if it is not in $C$, so if it is non-terminating
  or some $t_i(\gamma \cup \delta)$ is not computable.  Since
  head-reductions are impossible, non-termination implies
  non-termination of some $t_i(\gamma \cup \delta)$, which by
  \refLemma{lem:compresults}(\ref{lem:compresults:term}) implies
  non-computability of $t_i(\gamma \cup \delta)$ as well.  We are
  done because $t \bsuptermeq{A} t_i$.
\item $t = \apps{\afun}{t_1}{t_n}$ with $\identifier{f} \in \Defineds$
  but $n < \arity(\identifier{f})$: same as above, because the rules
  respect $\arity$;
\item $t = \apps{(\abs{x}{u})}{t_0}{t_n}$: $t(\gamma \cup \delta)$ is
  neutral, so by \refLemma{lem:neutralcomp} non-computability implies
  the non-computability of a reduct.  If the reduct
  $\apps{u(\gamma \cup \delta)[x:=t_0(\gamma \cup \delta)]}{
  (t_1(\gamma \linebreak
  \cup \delta))}{(t_n(\gamma \cup \delta))} =
  (\apps{u[x:=t_0]}{t_1}{t_n})(\gamma \cup \delta)$ is non-computable,
  we are done because $t \bsuptermeq{A} \apps{u[x:=t_0]}{t_1}{t_n}$.
  Otherwise, note that all many-step reducts of $t(\gamma \cup \delta)$
  are either also a reduct of $(\apps{u[x:=t_0]}{t_1}{t_n})(\gamma \cup
  \delta)$ -- and therefore computable -- or have the form
  $\apps{(\abs{x}{u'})}{t_0'}{t_n'}$ with $u(\gamma \cup \delta)
  \arrr{\Rules} u'$ and each $t_i(\gamma \cup \delta) \arrr{\Rules}
  t_i'$.  Thus, at least one of $u(\gamma \cup \delta)$ or $t_i(\gamma
  \cup \delta)$ has to be non-terminating.  But if $u(\gamma \cup
  \delta)$ is non-terminating, then so is $u[x:=u'](\gamma \cup \delta)$,
  contradicting computability of $(\apps{u[x:=t_0]}{t_1}{t_n})(\gamma
  \cup \delta)$.  The same holds if $t_i(\gamma \cup \delta)$ is
  non-terminating for some $i \geq 1$.  Thus, $t_0(\gamma \cup \delta)$
  is non-terminating and therefore non-computable, and we indeed have
  $t \bsuptermeq{A} t_0$.
\item $t = \apps{\meta{Z}{s_1,\dots,s_\mac}}{t_1}{t_n}$: let
  $\gamma(Z) \approxp \abs{x_1 \dots x_\mac}{u}$.  Then $u[x_1:=q_1,
  \dots,x_\mac:=q_\mac]$ is computable for all computable $q_1,\dots,
  q_\mac$:
  \begin{itemize}
  \item if $\gamma(Z) = \abs{x_1 \dots x_\mac}{u}$ then this holds by
    computability of all $\gamma(Z)$ and \refLemma{lem:abscomputable};
  \item if $\gamma(Z) = \abs{x_1 \dots x_i}{u'}$ with $i < \mac$ and
    $u = \apps{u'}{x_{i+1}}{x_\mac}$, then computability of
    $\gamma(Z)$ and \refLemma{lem:abscomputable} provide computability
    of $u'[x_1:=q_1,\dots,x_i:=q_i]$, which by definition of
    computability for higher-order terms implies computability
    for
    $\apps{u'[x_1:=q_1,\dots,x_i:=q_i]}{q_{i+1}}{q_n} = u[x_1:=q_1,
    \dots,x_n:=q_n]$.
  \end{itemize}
  Thus, if $u[x_1:=s_1(\gamma \cup \delta),\dots,x_\mac:=s_\mac(\gamma
  \cup \delta)]$ is non-computable, some $s_i(\gamma \cup \delta)$
  must be non-computable, and since substituting an unused variable
  has no effect, this must be the case for some $i$ with $x_i \in
  \FV(u)$.  So in this case, $\gamma$ respects $B := A \cup \{ Z : i
  \}$ and indeed $s \bsuptermeq{B} t \bsuptermeq{B} s_i$.
  If, on the other hand, $u[x_1:=s_1(\gamma \cup \delta),\dots,
  x_\mac:=s_\mac(\gamma \cup \delta)]$ is computable, then
  non-computability of $t(\gamma \cup \delta) = \apps{u[x_1:=
  s_1(\gamma \cup \delta),\dots,x_\mac:=s_\mac(\gamma \cup
  \delta)]}{(t_1(\gamma \cup \delta))}{(t_n(\gamma \cup \delta))}$
  implies non-computability of some $t_i(\gamma \cup \delta)$, and we
  indeed have $t \bsuptermeq{A} t_i$.
\qedhere
\end{itemize}
\end{proof}

From this, we have the main result on static dependency chains.

\oldcounter{\thmSSChain}{\thmSSChainSec}
\ssChainTheThm
\startappendixcounters

\begin{proof}
In the following, let a \emph{minimal non-computable term} be a term
$s := \apps{\afun}{s_1}{s_\maa}$ of base type, such that $\afun \in
\F$, and $s$ is non-terminating but all $s_i$ are $C$-computable.  We
say that $s$ is MNC.

We first observe that if $\arr{\Rules}$ is non-terminating, then there
exists a MNC term.  After all, if $\arr{\Rules}$ is non-terminating,
then there is a non-terminating term $s$, which (by
\refLemma{lem:compresults}(\ref{lem:compresults:term})) is also
non-computable.
Let $t\ (A)$ be the element of $\cand(s)$ that is given by
\refLemma{lem:scandcomplete} for $\gamma = \delta = []$.  Then $A =
\emptyset$ and $t$ has the form $\apps{\afun}{t_1}{t_\maa}$ with
$\maa = \arity(\afun)$ (so $t$ has base type), and there exists a
computable substitution $\delta$ such that $t\delta$ is not computable
but all $t_i\delta$ are.  By definition of $C$, this can only be
because $t\delta$ is non-terminating.

Thus, assuming $\arr{\Rules}$ is non-terminating, we can select a MNC
term $t_{-1}$.  Now for $i \in \N$, let a MNC term $t_{i-1} =
\apps{\afun}{q_1}{q_\maa}$ be given.  Because $t_{i-1}$ is
non-terminating and all $q_j$ are $C$-computable and therefore
terminating, we can write $t_{i-1} \arrr{\Rules,in}
(\apps{\afun}{\ell_1}{\ell_m})\gamma \linebreak
\arr{\Rules} r\gamma$ for some rule
$\apps{\afun}{\ell_1}{\ell_m} \arrz r$ and substitution $\gamma$,
where $\arrr{\Rules,in}$ indicates a reduction without root steps (so
each $q_j \arrr{\Rules} \ell_j\gamma$) and where $r\gamma$ is
non-terminating.  By \refLemma{lem:formative}, we can safely assume
that the reductions $q_j \arrr{\Rules} \ell_j\gamma$ are
$\ell_j$-formative if $\apps{\identifier{f}}{\ell_1}{\linebreak
\ell_m}$ is a
fully extended linear pattern; since $\ell$ is closed we can safely
assume that $\domain(\gamma) = \FMV(\ell)$.

Let $s_i := \apps{\afun^\sharp}{(\ell_1\gamma)}{(\ell_m\gamma)}$, and
note that all $\ell_j\gamma$ are computable by
\refLemma{lem:preservecomp}.
We observe that
for all $Z$ occurring in $r$ we have that $\gamma(Z)$ is $C$-computable
by a combination of accessible function passingness, computability
of $\ell_j\gamma$ and \refLemma{lem:pfp}.
As $r\gamma$ is non-computable, \refLemma{lem:scandcomplete} provides
a minimal element $t\ (A)$ of $\cand(r)$ and a computable
substitution $\delta$ on domain $\FV(t)$ such that $\gamma$ respects
$A$ and $t(\gamma \cup \delta)$ is not computable.
For $\FV(t) = \{ x_1,\dots,x_n\}$, let $Z_1,\dots,Z_n$ be fresh
meta-variables; then $p :=\linebreak t[x_1:=Z_1,\dots,x_n:=Z_n] = \mathit{
metafy}(t)$, and $p\eta = t\delta$ for $\eta$ the substitution
mapping $X \in \FMV(\ell)$ to $\gamma(\ell)$ and each $Z_j$ to
$\delta(x_j)$.

Set $\rho_i := \ell^\sharp \arrdp p^\sharp\ (A)$ and $t_i :=
p^\sharp\eta$.  We observe that $t_i$ is MNC, because the meta-term
$t$ supplied by \refLemma{lem:scandcomplete} has the form
$\apps{\bfun}{u_1}{u_\mia}$ with $\mia = \arity(\bfun)$ and $u_j(
\gamma \cup \delta)$ is $C$-computable for each $j$ because $t
\bsuptermeq{A} u_j$.  Thus, we can continue the infinite construction.

The chain $[(\rho_i,s_i,t_i) \mid i \in \N]$ thus constructed is an
infinite formative $(\SDP(\Rules),\Rules)$-dependency chain.  That it
is a $(\SDP(\Rules),\linebreak
\Rules)$-dependency chain is obvious because each
$\rho_i \in \SDP(\Rules)$ (since $t\ (A) \in \cand(r)$ and $t$ has the
right form), because $\gamma$ respects $A$ and $\eta$ corresponds with
$\gamma$ on all meta-variables that take arguments, and because
$\FV(p) = \emptyset$ and $\domain(\eta) = \FMV(\ell) \cup \{ Z_1,\dots,
Z_n \} = \FMV(\ell) \cup \FMV(p)$.
That it is a formative chain follows by the initial selection of
$\gamma$, as we assumed formative reductions to each $\ell_j\gamma$.

It is also a computable chain: clearly we have $t_i = p\eta$ in step $i$ in
the construction above.  Suppose $p^\sharp \bsuptermeq{B} v$ and
$\eta$ respects $B$, but $(\abs{x_1 \dots x_n}{v})\eta$ is not
computable for $\FV(v) = \{ x_1,\dots,x_n \}$ -- so by
\refLemma{lem:abscomputable}, $v(\eta \cup \zeta)$ is not computable for
some computable substitution $\zeta$ on domain $\FV(v)$.  Since the
meta-variables $Z_j$ do not occur applied in $p$, we can safely
assume that $B$ contains only conditions for the meta-variables
in $\domain(\gamma)$.  By renaming each $Z_j$ back to $x_j$, we
obtain that $\gamma$ respects $B$ and $t \bsuptermeq{B} v'$ with
$v = v'[x_1:=Z_1,\dots,x_n:=Z_n]$.  But then $r\gamma \bsuptermeq{A
\cup B} t \bsuptermeq{A \cup B} v'$ and $\gamma$ respects $A \cup
B$ and $v'(\gamma \cup \delta \cup \zeta)$ is non-computable.  By
minimality of the choice $t\ (A)$, we have $v' = t$, so $v = p^\sharp$.
However, $p^\sharp\eta$ has a marked symbol $\symb{g}^\sharp$ as a
head symbol, and thus cannot be reduced at the top; by
$C$-computability of its immediate subterms, it is terminating, and
therefore computable itself.
\end{proof}


We also prove the statement in the text that $\AlterRules$-computability
implies minimality:

\begin{lemma}\label{lem:computableminimal}
Every $\AlterRules$-computable $(\P,\Rules)$-dependency chain is 
\linebreak minimal.
\end{lemma}

\begin{proof}
Let $[(\rho_i,s_i,t_i) \mid i \in \N]$ be a $\AlterRules$-computable
$(\P,\Rules)$-chain and let $i \in \N$; we must prove that the strict
subterms of $t_i$ are terminating under $\arr{\Rules}$.  By
definition, since $\bsuptermeq{\emptyset}$ is a symmetric relationship,
$t_i$ is $C_{\AlterRules}$-computable where $C_{\AlterRules}$ is given by
\refThm{thm:defC} for a relation $\arr{\AlterRules} \mathop{\supseteq}
\arr{\Rules}$.  By
\refLemma{lem:compresults}(\ref{lem:compresults:term}), $t_i$ is
therefore terminating under $\arr{\AlterRules}$, so certainly under
$\arr{\Rules}$ as well.  The strict subterms of a terminating term are
all terminating.
\end{proof}

And we thus obtain:

\oldcounter{\thmSChain}{\thmSChainSec}
\sChainTheThm
\startappendixcounters

\begin{proof}
A direct consequence of \refThm{thm:sschain} and
\refLemma{lem:computableminimal}.
\end{proof}

\subsection{Original static dependency pairs}\label{app:sdp:original}

Since the most recent work on static dependency pairs has been defined
for a polymorphic variation of the HRS formalism, it is not evident
from sight how our definitions relate.  Here, we provide context by
showing how the definitions from~\cite{kus:iso:sak:bla:09,suz:kus:bla:11}
apply to the restriction of HRSs that can be translated to AFSMs.

\begin{definition}\label{def:pfp1}
An AFSM $(\F,\Rules,\arity)$ is plain function passing following
\cite{kus:iso:sak:bla:09} if:
\begin{itemize}
\item
for all rules $\apps{\afun}{\ell_1}{\ell_\maa} \arrz r$ and all
  $Z \in \FMV(r)$:
  if $Z$ does not have base type, then there are variables
  $x_1,\dots,x_n$ and some $i$ such that $\ell_i = \abs{x_1 \dots
  x_n}{\meta{Z}{x_{j_1},\dots,x_{j_\mia}}}$.
\end{itemize}

An AFSM $(\F,\Rules,\arity)$ is plain function passing following
\cite{suz:kus:bla:11} if:
\begin{itemize}
\item
  for all rules $\apps{\afun}{\ell_1}{\ell_\maa} \arrz r$ and all
  $Z \in \FMV(r)$: there are
     some variables
  $x_1,\dots,x_\mia$ and
  some $i \leq \maa$ such that $\ell_i \safesup \meta{Z}{x_1,\dots,
  x_k}$, where the relation $\safesup$ is given by:
  \begin{itemize}
  \item
    $s \safesup s$,
  \item
    $\abs{x}{t} \safesup s$ if $t \safesup s$,
  \item
    $\apps{x}{t_1}{t_n} \safesup s$ if $t_i \safesup s$ for some
    $i$ with $x \in \V \setminus \FV(t_i)$%
  \item
    $\apps{\identifier{f}}{t_1}{t_n} \safesup s$ if $t_i \safesup
    s$ for some $t_i$ of base type.%
    \footnote{The authors of \cite{suz:kus:bla:11} refer to such
    subterms as \emph{accessible}.  We do not use this terminology,
    as it does not correspond to the accessibility notion in
    \cite{bla:jou:oka:02,bla:jou:rub:15} which we follow here.
    In particular, the accessibility notion we use considers the
    relation $\ext{\gracsortup}$, which corresponds to the positive/negative
    inductive types in \cite{bla:jou:oka:02,bla:jou:rub:15}.  This is
    not used in \cite{suz:kus:bla:11}.}
  \end{itemize}
\end{itemize}

In addition, in both cases right-hand sides of rules are assumed to be
presented in $\beta$-normal form
and arities are
maximal (for all $\afun : \atype_1 \arrtype \dots \arrtype \atype_\maa
\arrtype \asort$ we have $\arity(\afun) = \maa$).
\end{definition}

The definitions of PFP in \cite{kus:iso:sak:bla:09
,suz:kus:bla:11} 
also capture some non-pattern HRSs,
but these cannot be represented as AFSMs.
Note that the key difference
between $\safesup$ and $\suptermeq$ for patterns is that the former
is not allowed to descend into a non-base argument of a function symbol.
The same difference applies when comparing $\safesup$ with $\gracc$:
$\safesup$ also cannot descend into the accessible higher-order
arguments.

\begin{example}
The rules from \refEx{ex:mapintro} are PFP
following both defi\-nitions.
The rules from \refEx{ex:ordrec} are not PFP
in either definition, since $\symb{lim}\ F \safesup F$ does
not hold.
Note that they \emph{are} AFP.
\end{example}

For a PFP AFSM, static dependency pairs are then
defined as \emph{pairs} $\ell^\sharp \arrdp \apps{\identifier{f}^{
\sharp}}{p_1}{p_m}$.
This allows for a very simple notion of chains, similar to the one in
the first-order setting:

\begin{definition}\label{def:sdpchain1}
A static dependency chain following
\cite{kus:iso:sak:bla:09,suz:kus:bla:11} is an infinite sequence
$[(\ell_i \arrdp p_i,\gamma_i) \mid i \in \N]$, where
$p_i\gamma_i \arrr{\Rules} \ell_{i+1}\gamma_{i+1}$ for all $i$.
It is \emph{minimal} if each $p_i\gamma_i$ is terminating under
$\arr{\Rules}$.
\end{definition}

Despite the absence of subterm steps in \refDef{def:sdpchain1},
absence of (minimal) static chains still implies termination:

\begin{theorem}[\cite{kus:iso:sak:bla:09,suz:kus:bla:11}]\label{thm:pfp1}
Let $\Rules$ be plain function passing following either definition in
\refDef{def:pfp1}.  Let $\P = \{ \ell^\sharp \arrdp \apps{
\identifier{f}^\sharp}{p_1}{p_\maa} \mid \ell \arrz r \in \Rules \wedge
r \suptermeq \apps{\identifier{f}}{p_1}{p_\maa} \wedge \identifier{f}
\in \Defineds \wedge \maa = \arity(\identifier{f}) \}$.  If $\arr{\Rules}$
is non-terminating, then there is an infinite minimal static
dependency chain with all $\ell_i \arrdp p_i \in \P$.
\end{theorem}

This definition is very close to the corresponding first-order notion.
However, its simplicity comes at a price: completeness is lost.  This
is both because meta-variable conditions are not considered, and for
the same reason that \refThm{thm:schain} is one-directional: the free
variables in the right-hand sides of dependency pairs may be
instantiated by anything.


Note that $\gracc$ corresponds to $\safesup$ (from
\refDef{def:pfp1}) if $\greqsort$\linebreak equates all sorts (
as then always
$\Acc(\afun) = \{$
the indices of all base type arguments
of $\afun\}$
).  Thus,
\refDef{def:apfp} includes both notions from \refDef{def:pfp1}.

\section{Dependency pair processors}\label{app:processors}

In this appendix, we prove the soundness -- and where applicable
completeness -- of all DP processors defined in the text.

We first observe:

\begin{lemma}\label{lem:subsetcomplete}
If $\Proc$ maps every DP problem to a set of problems such that for
all $(\P',\Rules',m',f') \in \Proc(\P,\Rules,m,f)$ we have that $\P'
\subseteq \P$, $\Rules' \subseteq \Rules$, $m' \succeq m$ and $f' =
f$, then $\Proc$ is complete.
%
\end{lemma}

\begin{proof}
$\Proc(\P,\Rules,m,f)$ is never \texttt{NO}.  Suppose $\Proc(\P,
\Rules,m,f)$ contains an infinite element $(\P',\Rules',m',f')$; we
must prove that then $(\P,\Rules,m,f)$ is infinite as well.  This is
certainly the case if $\arr{\Rules}$ is non-terminating, so assume
that $\arr{\Rules}$ is terminating.  Then certainly $\arr{\Rules'}
\mathop{\subseteq} \arr{\Rules}$ is
terminating as well, so $(\P',\Rules',m',f')$ can be infinite only
because there exists an infinite $(\P',\Rules')$-chain that is
$\AlterRules$-computable if $m' = \static_\AlterRules$, minimal if
$m' = \minimal$ and formative if $f' = \formative$.  By definition,
this is also a $(\P,\Rules)$-dependency chain, which is formative if
$f = f' = \formative$.  Since $\arr{\Rules}$ is terminating, this
chain is also minimal.
If we have $m = \static_S$, then also $m' = \static_S$
(since $\static_S$ is maximal under $\succeq$)
and the chain is indeed $S$-computable.
\end{proof}

\subsection{The dependency graph}

The dependency graph processor allows us to split a DP problem into
multiple smaller ones.  Despite the clear weakness that all collapsing
DPs are necessarily on the same cycle, this processor is useful both
when considering static and dynamic dependency pairs.  To prove
correctness of its main processor, we first prove a lemma that will
also aid in \refThm{thm:sdpproc}.

\begin{lemma}\label{lem:graphend}
Let $\adpprob = (\P,\Rules,m,f)$ and $G_\theta$ an approximation of
its dependency graph.  Then for every infinite $\adpprob$-chain
$[(\rho_i,s_i,t_i) \mid i \in \N]$ there exist $n \in \N$ and a cycle
$C$ in $G_\theta$ such that for all $i > n$: if $\rho_i \neq
\mathtt{beta}$ then $\theta(\rho_i) \in C$.
\end{lemma}

\begin{proof}
Suppose $\rho_i,\rho_j \in \P$.  We say that $\rho_i$ ``follows''
$\rho_j$ in the chain if $i = j + k$ for $k > 0$ and for all $j < m <
i$: $\rho_m = \mathtt{beta}$.  We claim: \emph{(**) if $\rho_i$
follows $\rho_j$ in the chain, then there is an edge from $\theta(
\rho_j)$ to $\theta(\rho_i)$.}

By definition of \emph{approximation}, the claim follows if $DG$ has
an edge from $\rho_j$ to $\rho_i$.  This is certainly the case if
$\rho_j$ is collapsing.  Otherwise, write $\rho_j = \ell_j \arrdp
p_j\ (A_j)$ with $p_j = \apps{\afun}{q_1}{q_n}$; since case
\ref{depchain:dp:instance} applies, there is a substitution
$\gamma_j$ that respects $A_j$ such that $t_j = p_j\gamma_j$.  Due
to the form of $p_j\gamma_j$, $\rho_{j+1}$ cannot be \texttt{beta},
so $i = j + 1$; we can write $\rho_i = \apps{\afun}{u_1}{u_n} \arrdp
p_i\ (A_i)$ and there exists a substitution $\gamma_i$ that respects
$A_i$ such that $s_i = \apps{\afun}{(u_1\gamma_i)}{(u_n\gamma_i)}$
and each $q_k\gamma_j \arrr{\Rules} u_k\gamma_i$.  Thus, all the
requirements are satisfied for there to be an edge from $\rho_j$ to
$\rho_i$ in $DG$.

Now, having (**), we see that the chain traces an infinite path in
$G_\theta$.  Let $C$ be the set of nodes that occur infinitely often
on this path; then for every node $d$ that is not in $C$, there is
an index $n_d$ after which $\theta(\rho_i)$ is never $d$ anymore.
Since $G_\theta$ is a \emph{finite} graph, we can take $n :=
\max(\{ n_d \mid d$ a node in $G_\theta \wedge d \notin C \})$.  Now
for every pair $d,b \in C$: because they occur infinitely often, there
is some $i > n$ with $\theta(\rho_i) = d$ and there is $j > i$ with
$\theta(\rho_j) = b$.  Thus, by (**) there is a path in $G_\theta$
from $d$ to $b$.  Similarly, there is a path from $b$ to $d$.  This
implies that they are on a cycle.
\end{proof}

Note that we did not modify the original chain at all, beyond looking
at a tail.  This is why the same flags apply to the resulting chain.
This makes it very easy to prove correctness of the main processor:

\oldcounter{\procDPGraph}{\procDPGraphSec}
\DPGraphTheProc
\startappendixcounters

\begin{proof}
Completeness follows by \refLemma{lem:subsetcomplete}.
Soundness follows because if $(\P,\Rules,m,f)$ admits an infinite
chain, then by \refLemma{lem:graphend} there is a cycle $C$ such that
a tail of this chain is mapped into $C$.  Let $C'$ be the strongly
connected component in which
$C$ lies, and $\P' = \{ \rho \in \P \mid \theta(\rho) \in
C' \}$.  Then clearly the same tail lies in $\P'$, giving an infinite
$(\P',\Rules,m,f)$-chain, and $(\P',\Rules,m,f)$ is one of the elements
of the set returned by the dependency graph processor.
\end{proof}

The dependency graph processor is essential to prove termination in
our framework because it is the only processor defined so far that
can remove a DP problem to $\emptyset$.

%
%

\subsection{Processors based on reduction triples}

For a complete picture when it comes to reduction triples, we first
present a processor that was \emph{not} used in the text, but that
most naturally corresponds to the reduction pair processor of
first-order rewriting.

\newcommand{\redpairTheProc}{
\begin{theorem}[Basic reduction triple processor]\label{thm:redpairproc}
Let $M = (\P_1 \uplus \P_2,\Rules,m,f)$ be a DP problem.
If $(\rge,\pge,\pgt)$ is a reduction triple such that
\begin{enumerate}
\item
for all $\ell \arrdp p\ (A) \in \P_1$, we have $\ell \pgt p$;
\item
for all $\ell \arrdp p\ (A) \in \P_2$, we have $\ell \pge p$;
\item
for all $\ell \arrz r \in \Rules$, we have $\ell \rge r$;
\item
if $\P$ contains a collapsing DP, then we have $\supterm \mathop{\subseteq} \pge$,
and\linebreak $\apps{\afun}{X_1}{X_{\arity(\afun)}} \pge \apps{\afun^\sharp}{X_1}{X_{\arity(\afun)}}$
for fresh meta-variables $X_1, \ldots, X_{\arity(\afun)}$
holds for all $\afun \in \F$,
\end{enumerate}
then the function that maps $M$ to $\{(\P_2,\Rules,m,f)\}$
is a sound and complete DP processor.
\end{theorem}
}
\redpairTheProc


\begin{proof}
Completeness follows by \refLemma{lem:subsetcomplete}.  Soundness
follows because every infinite $(\P_1 \uplus \P_2,\Rules)$-chain
$[(\rho_i,s_i,t_i) \mid i \in \N]$ with $\P_1,\P_2,\Rules$ satisfying
the given properties induces an infinite $\pgt \cup \pge \cup \rge$
sequence, and every occurrence of a DP in $\P_1$ in the chain
corresponds to a $\pgt$ step in the sequence.  By compatibility of
the relations, well-foundedness guarantees that there can only be
finitely many such steps, so there exists some $n$ such that
$[(\rho_i,s_i,t_i) \mid i \in \N \wedge i > n]$ is an infinite
$(\P_2,\Rules)$-chain.

To see that we indeed obtain the sequence, let $i \in \N$, and let
$\eta_i$ be a renaming on a finite set of variables (we let $\eta_1
:= []$).

If $\rho_i = \mathtt{beta}$, then $s_i \arr{\beta} \cdot
\suptermeq t_i'$ for some $t_i'$ with $t_i = {t_i'}^\sharp$; but then
also $s_i\eta_i \arr{\beta} \cdot \suptermeq t_i'\eta_i$ and
$t_i = {t_i'}^\sharp$.  Since $\arr{\beta}$ is included in $\rge$ and
$\suptermeq$ in $\rge$, as are marking steps, $s_i\eta_i \rge \cdot
\pge t_i\eta_i$.  We let $\eta_{i+1} := \eta_i$.

Otherwise, $\rho_i$ has the form $\ell \arrdp p\ (A)$ and there is a
substitution $\gamma$ on domain $\FMV(\ell) \cup \FMV(p) \cup \FV(p)$
that maps all variables to distinct fresh variables, such that
$s_i = \ell\gamma$ and $p\gamma \suptermeq t_i'$ for some $t_i'$ with
$t_i = {t_i'}^\sharp$.  But then, writing $\delta$ for the substitution
mapping each $\gamma(x)$ with $x \in \FV(p)$ back to $x$,
meta-stability gives us that $s_i\eta_i = \ell(\gamma\delta\eta_i)\ 
(\pge \cup \pgt)\ p(\gamma\delta\eta_i) = p\gamma(\delta\eta_i)
\suptermeq t_i'(\delta\eta_i)$.  Letting $\eta_{i+1}$ be the
limitation of $\delta\eta_i$ to the variables occurring freely in
$t_i$, and using that the marking steps are included in $\pge$, we
thus have $s_i\eta_i\ (\pge \cup \pgt)\; \cdot \pge\ t_i\eta_{i+1}$.
Moreover, a $\pgt$ step is used if $\rho_i \in \P_1$.

Since always $t_i \arrr{\Rules} s_{i+1}$, clearly also $t_i\eta_{i+1}
\arrr{\Rules} s_{i+1}\eta_{i+1}$, and $\arrr{\Rules}$ is included in
$\rge$ by the ordering requirements for $\Rules$, monotonicity and
transitivity.
\end{proof}

Now that we have seen
a basic processor using
reduction triples,
let us consider the base-type processor presented in the text.

\oldcounter{\procBasetypeRedpair}{\procBasetypeRedpairSec}
\basetypeRedpairTheProc
\startappendixcounters


\begin{proof}
Completeness follows by \refLemma{lem:subsetcomplete}.  Soundness
follows as in the proof for the basic reduction triple processor, as
we can generate a $\pgt \cup \pge \cup \rge$ sequence using a $\pgt$
step for every occurrence of an element of $\P_1$ in a given chain.

In the following, we say $s \sim t$ if $s : \atype_1 \arrtype \dots
\arrtype \atype_\maa \arrtype \asort$ with $\asort \in \Sorts$ and $t$
has the form $\apps{(s\eta)}{u_1}{u_\maa}$ for some substitution
$\eta$.  We first prove:
\emph{(**) if $\P$ contains a collapsing DP, and if $s \suptermeq t$
and $s \sim s'$, then there exists $t'$ such that $t \sim t'$ and $s'
\ (\pge \cup \rge)^*\ t'$.}
We prove this by induction on the derivation of $s \suptermeq t$.
\begin{itemize}
\item if $s = t$, then we can choose $t' := s'$;
\item if $s = \abs{x}{q}$ and $q \suptermeq t$, then $s' =
  \apps{(\abs{x}{q'\eta})}{u_1}{u_\maa} \rge \apps{q'\eta[x:=u_1]}{
  u_2}{u_\maa} =: q'$ because $\rge$ includes $\arr{\beta}$ (and
  $\maa > 0$ because $s'$ has base type); we complete by the
  induction hypothesis because $q \sim q'$.
\item if $s = \apps{u_0}{q_1}{q_n}$ with $u_0$ a variable, abstraction
  or function symbol, and $q_i \suptermeq t$ for some $i$, then
  \[
  s' = \apps{\apps{(u_0\eta)}{(q_1\eta)}{(q_n\eta)}}{u_1}{u_\maa}
  \pge \apps{(q_i\eta)}{\bot_{\atype_1}}{\bot_{\atype_l}} =: q_i',
  \]
  and since $q_i \sim q_i'$, we complete by the induction hypothesis.
\end{itemize}
Next we observe:
\emph{(***) if $s \sim s'$ and $s \arrr{\Rules} t$, then there
exists $t'$ such that $t \sim t'$ and $s' \rge t'$.}
After all, if $s \sim s'$, we can write $s' = \apps{(s\eta)}{u_1}{
u_\maa}$ and letting $t' := \apps{(t\eta)}{u_1}{u_\maa}$ we both have
$t \sim t'$ (clearly) and $s' \arrr{\Rules} t'$ because $\arr{\Rules}$
is stable under substitution and monotonic.  But then also $s' \rge
t'$, since $\rge$ contains $\arr{\Rules}$ (as $\rge$ contains
$\arr{\beta}$ and $\rge$ orients the rules in $\Rules$ and is
meta-stable and monotonic) and is transitive.

Now, let $s_1' := \apps{s_1}{\bot_{\atype_1}}{\bot_{\atype_\maa}}$;
then clearly $s_1 \sim s_1'$.  For $i \in \N$, suppose $s_i \sim s_i'$.

If $\rho_i = \mathtt{beta}$, then $s_i \arr{\beta} \cdot \suptermeq b$
for some $b$ with $t_i = b^\sharp$.  Write $s_i' =
\apps{(\apps{(\abs{x}{q})\ v}{w_1}{w_n})\eta}{u_1}{u_\maa}$.  Then,
because $\arr{\beta}$ is included in $\rge$, we have $s_i' \rge
\apps{(\apps{u[x:=v]}{w_1}{w_n})\eta}{u_1}{u_\maa}$, which by (**)
$(\pge \cup \rge)^*\ b'$ for some $b'$ such that $b \sim b'$.  By
the inclusion of the marking rules in $\pge$, we have $b' \pge t_i'$
for a suitable $t_i'$ as well, and $t_i' \rge s_{i+1}'$ for suitable
$s_{i+1}'$ by (***).

Otherwise, $\rho_i = \ell \arrdp p\ (A)$ and $s_i = \ell\gamma$ and
$p\gamma \suptermeq b$ for some $b$ with $t_i = b^\sharp$.  We can
write $s_i' = \apps{\ell(\gamma\eta)}{u_1}{u_\maa}$; writing
$\delta := \gamma\eta \cup [Z_1:=u_1,\dots,Z_\maa:=u_\maa]$ for fresh
meta-variables $Z_i$, we then have $s_i' = (\apps{\ell}{Z_1}{Z_\maa})
\delta$.
Let $p : \btype_1 \arrtype \dots \arrtype \btype_n \arrtype \bsort$
and write $\xi$ for the substitution mapping each $x : \atype \in
\FV(p)$ to $\bot_\atype$.
If $\rho_i \in \P_1$, we have $\apps{\ell}{Z_1}{Z_\maa} \pgt
(\apps{p}{\bot_{\btype_1}}{\bot_{\btype_n}})\xi$; otherwise
$\apps{\ell}{Z_1}{Z_\maa} \pge
(\apps{p}{\bot_{\btype_1}}{\bot_{\btype_n}})\xi$.
By meta-stability, we have $s_i\ (\pgt \cup \pge)\ 
(\apps{p}{\bot_{\btype_1}}{\bot_{\btype_n}})\xi\gamma =
\apps{p(\gamma|_{\FMV(p)} \cup \xi)}{\bot_{\btype_1}}{\bot_{\btype_n}}
=: w'$ and clearly $p\gamma \sim w'$.  Therefore by (**) and the
inclusion of the marking rules and (***), we have
$w'\ (\pge \cup \rge)^*\: \cdot \pge t_i' \rge s_{i+1}'$ for suitable
$t_i'$ and $s_{i+1}'$.
\end{proof}

To obtain correctness
for the abstraction-simple reduction triple
processor, we first obtain a number of lemmas.
In the following, we let $\afun^- = \afun$ for $\afun \notin \funs(
\P,\Rules)$ (so $\ttag$ is defined on all terms).  
We formalise \refDef{def:ttag} as follows:

\begin{itemize}
\item $\ttag(s_1\ s_2) = \ttag(s_1)\ \ttag(s_2)$ if $s_1$ respects
  $\arity$;
\item $\ttag(x) = x$ for $x \in \V$;
\item $\ttag(\abs{x}{s}) = \abs{x}{\ttag(s)}$;
\item $\ttag(\meta{Z}{s_1,\dots,s_\mac}) = \meta{Z}{\ttag(s_1),
  \dots,\ttag(s_\mac)}$;
\item $\ttag(\apps{\afun}{s_1}{s_\mia}) = \apps{\afun}{\ttag(s_1)}{
  \ttag(s_\mia)}$ if $\FV(\apps{\afun}{s_1}{s_\mia}) = \emptyset$ and
  $\ttag(\apps{\afun}{s_1}{s_\mia}) = \apps{\afun^-}{\ttag(s_1)}{
  \ttag(s_\mia)}$ if $\FV(\apps{\afun}{s_1}{\linebreak s_\mia})
  \neq \emptyset$,
  for $\mia = \arity(\afun)$
\end{itemize}

\begin{lemma}\label{lem:tagredpair:tagapplies}
If $s$ respects $\arity$, then $\ttag(s\ t) = \ttag(s)\ \ttag(t)$.
\end{lemma}

\begin{proof}
Trivial by case analysis on the form of $s$.
\end{proof}

\begin{lemma}\label{lem:tagredpair:tagmetavar}
Let $\gamma$ be a substitution.
If $\gamma(Z) \approxp \abs{x_1 \dots x_\mac}{s}$, then
$\gamma^\ttag(Z) \approxp \abs{x_1 \dots x_\mac}{\ttag(s)}$.
\end{lemma}

\begin{proof}
Immediate if $\gamma(Z) = \abs{x_1 \dots x_\mac}{s}$; otherwise
write $\gamma(Z)\linebreak = \abs{x_1 \dots x_n}{s'}$ with $n < \mac$ and $s =
\apps{s'}{x_{n+1}}{x_\mac}$.  By
\refLemma{lem:tagredpair:tagapplies}, $\ttag(s) = \apps{\ttag(s')}{
\ttag(x_{n+1})}{\ttag(x_\mac)} = \apps{\ttag(s')}{x_{n+1}}{x_\mac}$.
\end{proof}

\begin{lemma}\label{lem:tagredpair:tagmeta}
Suppose that $(\P,\Rules,m,f)$ and $(\rge,\pge,\pgt)$ satisfy the
conditions in \refThm{thm:tagredpair}.
If $s$ is a meta-term and $\gamma$ a substitution with $\FMV(s)
\subseteq \domain(\gamma)$ that maps everything in its domain to
closed terms, then $\ttag(s)\gamma^\ttag \rge \ttag(s\gamma)$.
\end{lemma}

Here, and in subsequent lemmas, $\gamma^\ttag =
[\ttag(\gamma(\avarormeta)) \mid \avarormeta \in \domain(\gamma)]$.

\begin{proof}
By induction first on the number of meta-variables occurring in $s$,
second on its form.
\begin{itemize}
\item
  If $s = s_1\ s_2$ and $s_1$ respects $\arity$, then
  by \refLemma{lem:tagredpair:tagapplies}
  $\ttag(s)\gamma^\ttag\linebreak = (\ttag(s_1)\gamma^\ttag)\ 
  (\ttag(s_2)
  \gamma^\ttag)$, which by
  the second part of the induction hypothesis $\rge
  \ttag(s_1\gamma)\ \ttag(s_2\gamma) = \ttag(s\gamma)$.
\item If $s = \abs{x}{s'}$, we easily complete by the second part of
  the induction hypothesis as well.
\item If $s = x \in \V$, then $\ttag(s)\gamma^\ttag = x\gamma^\ttag
  = \ttag(\gamma(x))$ (whether or not $x \in \domain(\gamma)$),
  $= \ttag(s\gamma)$.
\item If $s = \meta{Z}{s_1,\dots,s_\mac}$, then because $\FMV(s)
  \subseteq \domain(\gamma)$ we may denote $\gamma(Z) \approxp
  \abs{x_1 \dots x_\mac}{u}$ and by
  \refLemma{lem:tagredpair:tagmetavar}
  $\gamma^\ttag(Z) \approxp \abs{x_1 \dots x_\mac}{\ttag(u)}$.
  We have $\ttag(s)\gamma^\ttag = \meta{Z}{\ttag(s_1),\dots,\linebreak
  \ttag(
  s_\mac)}\gamma^\ttag = \ttag(u)[x_1:=\ttag(s_1)\gamma^\ttag,\dots,
  x_\mac:=\linebreak \ttag(s_\mac)
  \gamma^\ttag]$, which by the second part of the induction
  hypothesis $\rge \ttag(u)[x_1:=\ttag(s_1\gamma),\dots,x_\mac
  :=\ttag(s_\mac\gamma)]$.  By the first part of the IH,
  this $\rge \ttag(u[x_1:=s_1\gamma,\dots,x_\mac:=s_\mac\gamma]) =
  \ttag(s\gamma)$.
\item
  If $s = \apps{\afun}{s_1}{s_\mia}$ and $\FV(s) \neq \emptyset$, then
  we have
  $\ttag(s)\gamma^\ttag =
  (\apps{\afun^-}{\ttag(s_1)}{\ttag(s_\mia)})\gamma^\ttag \rge
  \apps{\afun^-}{\ttag(s_1\gamma)}{\ttag(s_\mia\gamma)}$ by the
  induction hypothesis.
  If $\FV(s\gamma) \neq \emptyset$ or $\afun \notin \funs(\P,
  \Rules)$, this is exactly $\ttag(s\gamma)$. Otherwise, it still
  $\rge \ttag(s\gamma)$ because $\apps{\afun^-}{X_1}{X_\mia} \rge
  \apps{\afun}{X_1}{X_\mia}$ by assumption and $\rge$ is meta-stable.
\item
  If $s = \apps{\afun}{s_1}{s_\mia}$ and $\FV(s) = \emptyset$, then
  the same holds for $\FV(s\gamma)$ because $\gamma$ maps only to
  closed terms.  We easily complete by the induction hypothesis.
\qedhere
\end{itemize}
\end{proof}

\begin{lemma}\label{lem:tagredpair:tagpattern}
Suppose that $(\P,\Rules,m,f)$ and $(\rge,\pge,\pgt)$ satisfy the
conditions in \refThm{thm:tagredpair}.
If $\ell$ is a closed fully extended pattern such that meta-variables
only occur in a context $\abs{x_1 \dots x_\mia}{\meta{Z}{x_1,\dots,
x_\mia}}$, and $\gamma$ is a substitution mapping to closed terms
with $\FMV(\ell) \subseteq \domain(\gamma)$, then $\ttag(\ell\gamma)
\rge \ell\gamma^{\ttag}$.
\end{lemma}

\begin{proof}
By induction on the form of $\ell$.
\begin{itemize}
\item
  If $\ell = \apps{\afun}{\ell_1}{\ell_n}$, then because both $\ell$
  and all $\gamma(Z)$ are closed, we have
  $\ttag(\ell\gamma) = \apps{\afun}{
  \ttag(\ell_1\gamma)}{\ttag(\ell_n\gamma)} \rge
  \apps{\afun}{(\ell_1\gamma^\ttag)}{\linebreak
  (\ell_n\gamma^\ttag)}$ (by the
  induction hypothesis), $= (\apps{\afun}{\ell_1}{\ell_n})\gamma^{
  \ttag} = \ell\gamma^\ttag$.
\item
  If $\ell = \abs{x_1 \dots x_\mia}{\meta{Z}{x_1,\dots,x_\mia}}$,
  then 
  $\ttag(\ell\gamma) = \ttag(\gamma(Z)) = \gamma^\ttag(Z) = \ell
  \gamma^\ttag$.
\item
  Otherwise, because $\ell$ is closed, it can have only the form
  $\abs{x}{\ell'}$ where $\ell'$ does not contain meta-variables (as
  all meta-variable occurrences must be fully extended but occur with
  all their abstracted variables directly above them). Thus,\linebreak
  $\ttag(\ell\gamma) = \ttag(\ell)$, which $\rge \ell$ just by
  replacing every $\afun^-$ by $\afun$ step by step ($\rge$ is
  monotonic and transitive).
\qedhere
\end{itemize}
\end{proof}

\begin{lemma}\label{lem:tagredpair:beta}
Suppose that $(\P,\Rules,m,f)$ and $(\rge,\pge,\pgt)$ satisfy the
conditions in \refThm{thm:tagredpair}.
If $s \arr{\beta} t$, then $\ttag(s) \rge \ttag(t)$.
\end{lemma}

\begin{proof}
By induction on the form of $s$.  If $s = \apps{x}{s_1}{s_n}$ or $s =
\apps{(\abs{x}{s_0})}{s_1}{s_n}$ and the reduction takes place in one
of the $s_i$, we immediately complete by the induction hypothesis and
monotonicity of $\rge$; the same holds if $s = \apps{\afun}{s_1}{s_n}$
and $\FV(\apps{\afun}{s_1}{s_\mia}) = \emptyset$ for $\arity(\afun) =
\mia \leq n$ (since $\beta$-reduction cannot introduce a new free
variable).  If $s = \apps{\afun}{s_1}{s_n}$ and
$\FV(\apps{\afun}{s_1}{s_\mia}) \neq \emptyset$ and the reduction takes
place in $s_i$, then $\ttag(s) = \apps{\afun^-}{\ttag(s_1)}{\ttag(s_n)}
\rge \apps{\afun^-}{\ttag(s_1)}{\ttag(s_i')} \cdots \ttag(s_n)$ by
the induction hypothesis, and this either $= \ttag(t)$ or $\rge
\ttag(t)$ by a single tag removal at the head (if the reduction in
$s_i$ removed all free variables of the term).
Finally, if $s = \apps{(\abs{x}{u})\ v}{w_1}{w_n}$ and $t = \apps{
(u[x:=v])}{w_1}{w_n}$, then $\ttag(s) =
\apps{(\abs{x}{\ttag(u)})\ \ttag(v)}{\ttag(w_1)}{\ttag(w_n)} \rge
\apps{\ttag(u)[x:=\ttag(v)]}{\ttag(w_1)}{\ttag(w_n)}$ by the
inclusion of rootmost $\beta$-steps in $\rge$ and monotonicity;
by \refLemma{lem:tagredpair:tagmeta} and
\refLemma{lem:tagredpair:tagapplies}, this term $\rge
\apps{\ttag(u[x:=v])}{\ttag(w_1)}{\ttag(w_n)} = \ttag(t)$.
\end{proof}

\begin{lemma}\label{lem:tagredpair:formative}
Suppose that $(\P,\Rules,m,f)$ and $(\rge,\pge,\pgt)$ satisfy the
conditions in \refThm{thm:tagredpair}.
Let $\ell$ be a closed fully extended pattern 
whose
meta-variables
only occur in a context $\abs{x_1 \dots x_\mia}{\meta{Z}{x_1,\dots,
,x_\mia}}$, and $\gamma$ be a substitution mapping to closed terms
with $\FMV(\ell) \subseteq \domain(\gamma) \subseteq \M$.
If $s \arrr{\Rules} \ell\gamma$ by an $\ell$-formative reduction,
then $\ttag(s) \rge \ell\gamma^\ttag$.
\end{lemma}

\begin{proof}
By induction on the definition of an $\ell$-formative reduction.
In the induction, we do not limit interest to \emph{closed} $\ell$,
but rather consider $\ell$ such that: if $\FMV(\ell) \neq \emptyset$
then $\FV(\ell) = \emptyset$.
\begin{itemize}
\item If there are $\ell' \arrz r' \in \RulesEta$ and a substitution
  $\delta$ such that $s \arrr{\Rules} \ell'\delta$ by an
  $\ell'$-formative reduction and $r'\delta \arrr{\Rules} \ell\gamma$
  by an $\ell$-formative reduction, then the induction hypothesis
  gives that $\ttag(s) \rge \ell'\delta^\ttag$ and $\ttag(r'\delta)
  \rge \ell\gamma^\ttag$.  Moreover, since $\ell' \rge \ttag(r')$
  and $\rge$ is meta-stable, we have $\ell'\delta^\ttag \rge
  \ttag(r')\delta^\ttag \rge \ttag(r'\delta)$ by
  \refLemma{lem:tagredpair:tagmeta}. We complete by transitivity of $\rge$.
\item If $s = \apps{(\abs{x}{u})\ v}{w_1}{w_n}$ and $\apps{u[x:=v]}{
  w_1}{w_n} \arrr{\Rules} \ell\gamma$ by an $\ell$-formative
  reduction, then
  $\ttag(s) \rge
  \ttag(\apps{u[x:=v]}{w_1}{w_n})$ by \refLemma{lem:tagredpair:beta}.
  We complete by the induction hypothesis and transitivity of $\rge$.
\item If $\ell = \apps{\afun}{\ell_1}{\ell_n}$ and $s = \apps{\afun}{
  s_1}{s_n}$ with each $s_i \arrr{\Rules} \ell_i\gamma$, we have
  $\ttag(s_i) \rge \ell_i\gamma^\ttag$ by the induction
  hypothesis.  Whether or not $\FV(s) = \emptyset$, we have
  $\ttag(s) \rge \apps{\afun}{\ttag(s_1)}{\linebreak\ttag(s_n)}$ (either
  because $\rge$ is reflexive or by using an untagging step),
  which $\rge \ell\gamma^\ttag$ by monotonicity.
\item If $\ell = \abs{x_1 \dots x_\mia}{\meta{Z}{x_1,\dots,x_\mia}}$
  and $s = \abs{x_1}{s'}$, we observe that in fact $s \arrr{\beta}
  \ell\gamma$:
  \begin{itemize}
  \item if $\mia = 0$ then $s = \ell\gamma$ because the reduction is
    formative;
  \item if $\mia = 1$ then (because we have already considered a
    reduction using a headmost step) $s = \abs{x_1}{s'}$ with
    $s' \arrr{\Rules} \meta{Z}{x_1}\gamma$ by a $\meta{Z}{x_1
    }$-formative reduction; that is, $s' =\linebreak \meta{Z}{x_1}\gamma$, so
    $s = \ell\gamma$;
  \item if $\mia > 1$ then by definition of abstraction-simplicity
    all rules have base type.  Thus, $s \arrr{\beta} \abs{x_1 \dots
    x_\mia}{s''}$ with $s'' = \meta{Z}{x_1,\dots,x_\mia}\gamma$, so
    $\abs{x_1 \dots x_\mia}{s''} = \ell\gamma$.
  \end{itemize}
  Thus, $\ttag(s) \rge \ttag(\ell\gamma)$ by
  \refLemma{lem:tagredpair:beta} and $\ttag(\ell\gamma) \rge \ell
  \gamma^\ttag$ by \refLemma{lem:tagredpair:tagpattern}.
\item If $\FMV(\ell) \neq \emptyset$ and therefore $\FV(\ell) =
  \emptyset$, there are no other options: $\ell$ cannot be
  $\apps{x}{\ell_1}{\ell_n}$, and if $\ell = \abs{x}{\ell'}$, then
  any meta-variable occurrence in $\ell'$ must include $x$ (as
  $\ell$ is fully extended); given the required form of meta-variable
  occurrences in $\ell$, this only leaves $\ell' = \abs{x_2 \dots
  x_\mia}{\meta{Z}{x,x_2,\dots,x_\mia}}$.
\item If $\FMV(\ell) = \emptyset$ (so $\ell\gamma = \ell$), we must
  consider two more cases.
  First, if $s = \apps{x}{s_1}{s_n}$ and $\ell =
  \apps{x}{\ell_1}{\ell_n}$ with each $s_i \arrr{\Rules}
  \ell_i\gamma$, we immediately conclude by the induction hypothesis:
  $\ttag(s) = \apps{x}{\ttag(s_1)}{\ttag(s_n)} \rge
  \apps{x}{\ell_1}{\ell_n} = \ell$.
  Second, if $s = \abs{x}{s'}$ and $\ell = \abs{x}{\ell'}$ and
  $s' \arrr{\Rules} \ell'$ then $\ttag(s) = \abs{x}{\ttag(s')}
  \rge \abs{x}{\ell'} = \ell$ by the induction hypothesis and
  monotonicity of $\rge$.
\qedhere
\end{itemize}
\end{proof}

\begin{lemma}\label{lem:tagredpair:subterm}
Suppose that $(\P,\Rules,m,f)$ and $(\rge,\pge,\pgt)$ satisfy the
conditions in \refThm{thm:tagredpair}, and $\P$ contains a
collapsing DP.
Let $s \suptermeq t$ and $\emptyset \neq \FV(s) \subseteq \FV(t)$;
that is, $s$ is closed save for the variables that occur freely in
$t$, and there is at least one such variable.
Let $\gamma$ be a substitution on domain $\FV(s)$, mapping to closed
terms.
Let $u_1,\dots,u_\maa$ be closed terms such that $\apps{s}{u_1}{
u_\maa}$ is a well-typed term of base type.
Then there exist closed terms $v_1,\dots,v_n$ and a substitution
$\delta$ on domain $\FV(t) \setminus \FV(s)$ mapping to closed terms
such that $\apps{\ttag(s)\gamma^\ttag}{\ttag(u_1)}{\ttag(u_\maa)}\ 
(\pge \cup \rge)^*\ \ttag(\apps{((t\delta)^\sharp\gamma)}{v_1}{v_n})$,
with the latter term having base type as well.
\end{lemma}

\begin{proof}
By induction on the derivation of $s \suptermeq t$.
\begin{itemize}
\item If $s = t$ then
  by \refLemma{lem:tagredpair:tagmeta} and monotonicity,
  $\apps{\ttag(s)\gamma^\ttag\linebreak
  }{\ttag(u_1)}{\ttag(u_\maa)} \rge
  \apps{\ttag(s\gamma)}{\ttag(u_1)}{\ttag(u_\maa)}$, which
  by \refLemma{lem:tagredpair:tagapplies}
  $= \ttag(\apps{(s\gamma)}{u_1}{u_\maa})$.
  This suffices if $s$ does not
  have the form $\apps{\afun}{s_1}{s_\mia}$ with $\mia =
  \arity(\afun)$ (taking $\delta := []$ and $v_1:=u_1,\dots,
  v_\maa:=u_\maa$), as then $s^\sharp = s$.

  If, however, $s$ does have this form, then note that
  \linebreak
  $\ttag(s)\gamma^\ttag = \apps{\afun^-}{(\ttag(s_1)\gamma^\ttag)}{
  (\ttag(s_\mia)\gamma^\ttag)}$ because $\FV(s) \neq \emptyset$ by
  assumption.  Therefore, by \refLemma{lem:tagredpair:tagmeta},
  \linebreak
  $\ttag(s)\gamma^\ttag\ \ttag(\vec{u}) \rge
  \apps{\apps{\afun^-}{\ttag(s_1\gamma)}{\ttag(s_\mia\gamma)}}{\ttag(
  u_1)}{\linebreak\ttag(u_\maa)}$.
  Since always  $\apps{\afun^-}{X_1}{X_{\mia+\maa}} \pge
  \apps{\afun^\sharp}{X_1}{X_{\mia+\maa}}$ and by meta-stability,
  this term $\pge \apps{\apps{\afun^\sharp}{\ttag(s_1\gamma)}{
  \ttag(s_\mia\gamma)}}{\linebreak\ttag(u_1)}{\ttag(u_\maa)}$, which by
  \refLemma{lem:tagredpair:tagapplies},
  $= \ttag(\apps{\apps{\afun^\sharp}{(s_1\gamma)\linebreak}{(s_\mia\gamma)}}{
  u_1}{u_\maa}) = \ttag(\apps{(s^\sharp\gamma)}{u_1}{u_\maa})$.
\item If $s = s_1\ s_2$ with $s_1$ respecting $\arity$, then
  $\apps{\ttag(s)\gamma^\ttag}{\ttag(u_1)\linebreak
  }{\ttag(u_\maa)} =
  \apps{\ttag(s_1)\gamma^\ttag\ \ttag(s_2)\gamma^\ttag}{
  \ttag(u_1)}{\ttag(u_\maa)}$ by \refLemma{lem:tagredpair:tagapplies}.
  If $s_1 \suptermeq t$, then by
  \refLemma{lem:tagredpair:tagmeta} we observe that this term
  $\rge \apps{\ttag(s_1)\gamma^\ttag}{\ttag(u_0)}{
  \ttag(u_\maa)}$ with $u_0 = s_2\gamma$, and we complete with the induction
  hypothesis.
  If $s_2 \suptermeq t$, then note that $\apps{Z}{X_0}{X_1}{X_\maa}
  \pge X_0\ \vec{\bot}$, so by meta-stability of $\pge$ we have
  $\apps{\ttag(s_1)\gamma^\ttag\ \ttag(s_2)\gamma^\ttag}{
  \ttag(u_1)}{\ttag(u_\maa)}\linebreak \pge \ttag(s_2)\gamma^\ttag\ 
  \vec{\bot}$, and since $\ttag(\bot_\atype) = \bot_\atype$, we
  again complete by the induction hypothesis.
\item If $s = \abs{x}{s'}$ and $s' \suptermeq t$ then we can safely
  assume $x$ does not occur in $\domain(\gamma)$; we also know that
  $\maa > 0$ because $s\ \vec{u}$ has base type.
  Thus $\apps{\ttag(s)\gamma^\ttag}{\ttag(u_1)}{
  \ttag(u_\maa)} \rge \apps{\ttag(s')\gamma^\ttag[x:=\ttag(u_1)]}{
  \ttag(u_2)}{\ttag(u_\maa)}$ because $\rge$ includes
  $\arr{\beta}$-steps, and this
  term $\rge \apps{\ttag(s'[x:=u_1])\gamma^\ttag}{\linebreak\ttag(u_2)}{
  \ttag(u_\maa)}$
  by \refLemma{lem:tagredpair:tagmeta}
  since $\gamma$ maps to closed terms.  But if $s'
  \suptermeq t$ then $s'[x:=u_1] \suptermeq t[x:=u_1]$, so by the
  induction hypothesis, this term $(\pge \cup \rge)^*\ 
  \ttag(\apps{((t[x:=u_1]\delta)^\sharp \gamma)}{v_1}{v_n})$, which
  suffices because $[x:=u_1]\delta$ is clearly a closed substitution
  on domain $\FV(t) \setminus \FV(s)$ for $\delta$ a closed
  substitution on domain $\FV(t[x:=u_1]) \setminus \FV(s'[x:=u_1])$.
\item Finally, if $s = \apps{\afun}{s_1}{s_\mia}$ with $s_i
  \suptermeq t$, then because $\FV(s) \neq \emptyset$ we
  have $\apps{\ttag(s)\gamma^\ttag}{\ttag(u_1)}{\ttag(u_\maa)}
  = \apps{\apps{\afun^-}{(\ttag(s_1)\gamma^\ttag)\linebreak
  }{(\ttag(s_\mia)\gamma^\ttag)}}{\ttag(u_1)}{\ttag(u_\maa)}$
  . 
  As
  $\apps{\afun^-}{X_1}{X_{\mia+\maa}} \linebreak
  \pge
  X_i\ \vec{\bot}$ and $\pge$ is meta-stable, this term $\pge
  \ttag(s_i)\gamma^\ttag\ \vec{\bot}$; we complete with the
  induction hypothesis.
\qedhere
\end{itemize}
\end{proof}

\oldcounter{\procAbsSimpleRedpair}{\procAbsSimpleRedpairSec}
\absSimpleRedpairTheProc
\startappendixcounters


\begin{proof}
Completeness follows by \refLemma{lem:subsetcomplete}.  Soundness
follows as in the proof for the base-type reduction triple processor,
since we can generate a $\pgt \cup \pge \cup \rge$ sequence using a
$\pgt$ step for every occurrence of an element of $\P_1$ in a given
chain.  We let  $s \sim s'$ if $s' = \ttag(\apps{(s\eta)}{u_1}{
u_\maa})$ with $s'$ a closed term of base type.
For a given formative $(\P,\Rules)$-chain $[(\rho_i,s_i,t_i)
\mid i \in \N]$, let $t_1' := \ttag(\apps{t_1[\vec{x}:=\vec{\bot}]}{
\bot_{\atype_1}}{\bot_{\atype_\maa}})$; then clearly $s_1 \sim s_1'$.
For $i \in \N$ with $i > 1$, suppose $t_{i-1} \sim t_{i-1}'$.
Consider the possibilities for $\rho_i$.

First, if $\rho_i = \mathtt{beta}$, then $t_{i-1} = s_i$, and we can
safely assume that this only occurs if $\P$ contains a collapsing DP.
Let $s_i' := t_{i-1}'$; then certainly $s_i \sim s_i'$.  There are
two options:
\begin{itemize}
\item If 
  $s_i = \apps{(\abs{x}{q})\ v}{
  w_1}{w_n}$ with $n > 0$, then $t_i = \apps{q[x:=v]}{w_1}{w_n}$.
  By definition of $\sim$, we
  can write $s_i' =\linebreak
  \ttag(\apps{\apps{(\abs{x}{(q\eta)})\ (v\eta)}{
  (w_1\eta)}{(w_n\eta)}}{u_1}{u_\maa})$ for some substitution $\eta$.
  By choosing $t_i' := \ttag(\apps{\apps{q\eta[x:=v\eta]}{(w_1\eta)}{
  \linebreak
  (w_n\eta)}}{u_1}{u_\maa})$, we have both $t_i \sim t_i'$ and $s_i'
  \rge t_i'$, the latter by \refLemma{lem:tagredpair:beta}.
\item If 
  $s_i = (\abs{x}{q})\ v$ and
  $t_i = w^\sharp[x:=v]$ for some $w$ with $q \suptermeq w$ and
  $x \in \FV(w)$ but $x \neq w$, then write $s_i' = \ttag(
  \apps{(\abs{x}{q\eta})\ (v\eta)}{\linebreak
  u_1}{u_\maa})$ and observe that
  $x$ does not occur in domain or range of $\eta$, so $\FV(q\eta) =
  \{x\}$ since $\eta$ maps all variables in $\FV(s_i)$ to
  closed terms.  What is more, clearly $q\eta \suptermeq w\eta$.
  By \refLemma{lem:tagredpair:subterm},
  $s_i\ (\pge \cup \rge)^*\ \ttag(\apps{((w\eta\delta)^\sharp[x:=
  v\eta])}{v_1}{v_n})$
  for some closed substitution $\delta$ on domain
  $\FV(w) \setminus \{x\}$. As $w$ is not a variable, this is
  exactly $\ttag(\apps{w^\sharp[x:=v]\delta\eta}{v_1}{v_n})\linebreak =: t_i'$;
  indeed $t_i \sim t_i'$.
\end{itemize}

Otherwise, $\rho_i = \ell \arrdp p\ (A)$ and $s_i = \ell\gamma$ and
$t_{i-1} \arrr{\Rules} s_i$ by an $\ell$-formative reduction.
Write $t_{i-1}' = \ttag(\apps{(t_{i-1}\eta)}{u_1}{u_\maa})$.  As an
induction on the definition of formative reductions shows that they
are preserved under substitution -- that is, also $t_{i-1}\eta \arrr{
\Rules} \ell\gamma\eta$ by an $\ell$-formative reduction -- it is not
hard to see that also $\apps{(t_{i-1}\eta)}{u_1}{u_\maa} \arrr{\Rules}
(\apps{\ell}{Z_1}{Z_\maa})(\gamma\eta \cup [Z_1:=u_1,\dots,Z_\maa:=
u_\maa])$ by a $(\apps{\ell}{Z_1}{Z_\maa})$-formative reduction.
We obtain $t_{i-1} \rge (\apps{\ell}{Z_1}{Z_\maa})(\gamma\eta \cup
[\vec{Z}:=\vec{u}])^\ttag = \apps{\ell(\gamma\eta)^\ttag}{\ttag(u_1)}{
\ttag(u_\maa)} =: s_i'$ by \refLemma{lem:tagredpair:formative}.
Now, there are once more two possibilities:
\begin{itemize}
\item If $t_i = p\gamma$ and $\rho_i \in \P_1$, then
  $\apps{\ell}{Z_1}{Z_\maa} \pgt \ttag(p[\vec{x}:=\bot]\ \vec{\bot})$.
  By meta-stability, $s_i' \pgt \ttag(p[\vec{x}:=\bot]\ 
  \vec{\bot})(\gamma\eta)^\ttag$,
  which by \refLemma{lem:tagredpair:tagmeta} $\rge \ttag(\apps{
  (p[\vec{x}:=\bot]\gamma\eta)}{\bot_1}{\bot_n}) =: t_i'$.
  Thus, $s_i' \pgt \cdot \rge t_i'$.
\item If $t_i = p\gamma$ and $\rho_i \in \P_2$, we similarly have
  $s_i' \pge \cdot \rge t_i'$.
\item If $\rho_i \in \P_1$ and $p = \meta{Z}{p_1,\dots,p_\mac}$, then
  write $(\gamma\eta)(Z) \approxp \abs{y_1 \dots y_\mac}{u}$; there is
  some non-variable $v$ with $u \suptermeq v$ and $\FV(v) \cap
  \{y_1,\dots,y_\mac\} \neq \emptyset$ and $t_i = v[y_1:=p_1\gamma,
  \dots,y_\mac:=p_\mac\gamma]$.
  Again using meta-stability and the way tagged dependency pairs are
  oriented, we obtain $s_i' \pgt \ttag(p[\vec{x}:=\bot]\ \linebreak
  \vec{\bot})(\gamma\eta)^\ttag = \ttag(u)[y_1:=\ttag(p_1)[\vec{x}:=
  \vec{\bot}](\gamma\eta)^\ttag,\dots,\linebreak
  y_\mac:=\ttag(p_\mac)
  [\vec{x}:=\vec{\bot}](\gamma\eta)^\ttag]\ \vec{\bot}$,
  which by \refLemma{lem:tagredpair:tagmeta} and the observation that
  the variables $x_j$ that occur free in $p$ can be assumed to not
  occur in the range of $\gamma$,
  $\rge \ttag(u)[y_1:=\ttag(p_1\gamma[\vec{x}:=\vec{\bot}]\eta),\dots,
  y_\mac:=\ttag(p_\mac\gamma[\vec{x}:=\vec{\bot}]\eta)]\ \vec{\bot} =
  \ttag(u)[y_1:=p_1\gamma[\vec{x}:=\vec{\bot}]\eta,\dots,y_\mac:=
  p_\mac\gamma[\vec{x}:=\vec{\bot}]\eta]^\ttag\ \vec{\bot}\linebreak =: t_i''$.
  Note that since $\gamma\eta$ maps the meta-variables in $\FMV(\ell)$
  to closed terms, $\FV(u) \subseteq \{y_1,\dots,y_\mac\}$.

  Let $j$ be such that $y_j \in \FV(v)$; at least one such $j$ exists.
  Writing $\xi := [y_n:=p_n\gamma[\vec{x}:=\vec{\bot}]\eta \mid 1 \leq
  n \leq \mac \wedge n \neq j]$, we have
  $t_i'' = \ttag(u)\xi^\ttag[y_j:=p_j\gamma[\vec{x}:=\vec{\bot}]
  \eta]^\ttag\ \vec{\bot}$ (because $\xi$ maps to closed terms), which
  by \refLemma{lem:tagredpair:tagmeta} $\rge \ttag(u\xi))[y_j:=
  p_j\gamma[\vec{x}:=\vec{\bot}]\eta]^\ttag\ \vec{\bot}$, which we can
  also write as $\apps{\ttag(u\xi))[y_j:=p_j\gamma[\vec{x}:=
  \vec{\bot}]\eta]^\ttag}{\ttag(\bot_1)}{\ttag(\bot_{n'})} =: t_i'''$.
  Then $\{y_j\} = \FV(u\xi) \subseteq \FV(v\xi)$ and we still have
  $u\xi \suptermeq v\xi$.

  As $\rho_i$ is collapsing, we can apply
  \refLemma{lem:tagredpair:subterm} to obtain some $\delta$ and terms
  $\vec{v}$ such that
  $t_i'''\ (\pge \cup \rge)^*\ \ttag(((u\xi\delta)^\sharp[y_j:=
  p_j\gamma[\vec{x}:=\vec{\bot}]\eta])\ \vec{v}) =: t_i'$.  Since $u$
  is not a variable, we can move the $\sharp$, and since $\xi$ and
  $\delta$ both map to closed terms, we can swap them around; thus,
  $t_i' = \ttag(u^\sharp\delta[y_n:=p_n\gamma[\vec{x}:=\vec{\bot}]\eta
  \mid 1 \leq n \leq \mac]\ \vec{v}) =
  \ttag(u^\sharp[y_1:=p_1\gamma,\dots,y_n:=p_n\gamma][\vec{x}:=
  \vec{\bot}]\eta\delta\ \vec{v})$.  Thus, also $t_i \sim t_i'$.
\item If $\rho_i \in \P_2$ and $p = \meta{Z}{p_1,\dots,p_\mac}$, we
  similarly obtain $t_i'$ with $s_i' \pge \cdot \rge \cdot (\pge \cup
  \rge)^*\ t_i'$.
\end{itemize}
All in all, for all $i$ we have $t_{i-1}' \rge s_i'$ and
$s_i' \pgt \cdot (\pge \cup \rge)^* t_i'$ if $\rho_i \in \P_1$,
otherwise $s_i'\ (\pge \cup \rge)^* t_i'$.
\end{proof}

\subsection{From dynamic to static DPs}

The SDP processor is straightforwardly proved with the results
we already have.

\oldcounter{\procToSDP}{\procToSDPSec}
\toSDPTheProc
\startappendixcounters

\begin{proof}
Completeness follows by \refLemma{lem:subsetcomplete}.  As for
soundness: if there is an infinite $(\P,\Rules)$-dependency
chain and $\P \subseteq \DDP(\Rules)$, then $\Rules$ is non-terminating
by \refThm{thm:chain}, so by
\refThm{thm:sschain} there is an $\Rules$-computable and formative
infinite $(\SDP(\Rules),\Rules)$-dependency chain.
By \refLemma{lem:graphend}, this chain has a
tail that is fully in $\SDP(\Rules)_G$.
\end{proof}

  %
  %
  %

\subsection{Modifying collapsing dependency pairs}


\oldcounter{\procExtend}{\procExtendSec}
\extendTheProc
\startappendixcounters

\begin{proof}
Completeness if $\P \subseteq \DDP(\Rules)$ follows
similarly to
the proof
of \refLemma{lem:chain:complete}.

As for soundness, suppose there is an infinite
$(\P,\Rules,m,f)$-chain $[(\rho_i,s_i,t_i) \mid i \in \N]$ and let
$\gamma_i$ be the corresponding substitution for $i \in \N$.
We will construct a $(\P',\Rules,m,f)$-chain. 
Let $t'_{-1} := t_0$.  Now for $j \in \N$, suppose that we have
some $i > 0$ such that $t'_{j-1} = t_{i-1}$.  Let $(\rho_j',s_j',
t_j') := (\rho_i,s_i,t_i)$ if $\rho_i$ does not have the form
$\ell \arrdp \apps{\meta{Z}{u_1,\dots,u_\mac}}{u_{\mac+1}}{u_n}\ (A)$
with $n > \mac$.  If $\rho_i$ \emph{does} have this form, then write
$\gamma_i(Z) = \abs{x_1 \dots x_m}{q}$ with $q$ not an abstraction.
Let $\rho_j' := \ell \arrdp \meta{Z}{u_1,\dots,u_n}\ (A)$ and
$s_j' := s_i$; for $t_j'$ consider 
$m$:
\begin{itemize}
\item if $m \leq \mac$, then let $t_j' := t_i$;
\item if $\mac < m \leq n$ then let $t_j' := t_{i+m-\mac}$;
\item if $n < m$ then let $t_j' := t_{i+n-\mac}$.
\end{itemize}
We claim that the sequence
built
like this
is a $(\P',\Rules,m,f)$-chain,
where $\P' = \{ \mathtt{extend}(\rho) \mid \rho \in \P \}$.  It is
clear that all $\rho_j \in \P' \cup \{ \mathtt{beta} \}$, and that
the reduction property (formative or not) from $t_{j-1}'$ to $s_j'$ is
always satisfied (as it is satisfied for $t_{i-1},s_i$).  This
gives us
property \ref{depchain:reduce} of \refDef{def:chain} for all $j \in
\N$.  Also, if $\rho_j' = \rho_i$, then the properties for $j$ in both
cases \ref{depchain:beta} and \ref{depchain:dp} of \refDef{def:chain}
are satisfied.  Otherwise, $\rho_i = \ell \arrdp \apps{\meta{Z}{u_1,
\dots,u_\mac}}{u_{\mac+1}}{u_n}\ (A)$ and $\rho_j' = \ell \arrdp
\meta{Z}{u_1,\dots,u_n}\ (A)$.  Clearly $s_j' = \ell\gamma_i$.
Then:
\begin{itemize}
\item If $m \leq \mac$, then $(\apps{\meta{Z}{u_1,\dots,u_\mac}}{
  u_{\mac+1}}{u_n})\gamma_i = \meta{Z}{u_1,\dots,\linebreak
  u_n}\gamma_i =
  \apps{q[x_1:=u_1\gamma_i,\dots,x_m:=u_m\gamma_i]}{(u_{m+1}
  \gamma_i)}{(u_n\gamma_i)}\linebreak
  = t_j'$.  This choice is adequate since we can write
  \[\gamma_i(Z) \approx_{n} \abs{x_1
  \dots x_{n}}{\apps{q}{x_{m+1}}{
  x_{\mac+n}}}\] and \[(\apps{q}{x_{m+1}}{x_n})^\sharp =
  \apps{q}{x_{m+1}}{x_n}\] since $\sharp$ has no effect on
  applications if the head respects $\arity$.
\item If $\mac < m < n$, then $(\apps{\meta{Z}{u_1,\dots,
  u_\mac}}{u_{\mac+1}}{u_n})\gamma_i$ is an application, which
  reduces in $m-\mac$ headmost $\beta$-steps to
  the term $\apps{q[x_1:=
  u_1\gamma_i,\dots,x_m:=u_m\gamma_i]}{(u_{m+1}\gamma_i)}{(u_n\gamma_i
  )} = \meta{Z}{u_1,\dots,u_n}\gamma_i$.  Since, in a dependency
  chain, only \texttt{beta} can be applied if $s_i$ has a
  $\beta$-redex at its head, necessarily $\rho_{i+1},\dots,
  \rho_{i+m-\mac}$ are all \texttt{beta}.  Moreover, since the
  $\beta$-redex does \emph{not} occur at the top but only at the
  head, the subterm case is not applied.  Thus, $t_j' = t_{i+m-\mac}
  = \meta{Z}{u_1,\dots,u_n}\gamma_i$.  This choice is adequate as
  in the case $m \leq \mac$.
\item Finally, if $n \leq m$, then
  $t_i = (\apps{\meta{Z}{u_1,\dots,u_\mac}}{u_{\mac+1}}{u_n})
  \gamma_i$ is an application of length $n-\mac$ with a
  $\lambda$-abstraction at the head; $\rho_{i+1},\dots,
  \rho_{i+n-\mac}$ are all \texttt{beta}, but in the last step, the
  $\beta$-redex occurs at the \emph{top}.  That is,
  $s_{i+n-\mac} = (\abs{x_n \dots x_m}{\linebreak{}q[x_1:=u_1\gamma_i,\dots,
  x_{n-1}:=u_{n-1}\gamma_i]})\ (u_n\gamma_i)$ and there exists a
  non-variable term $w$ with $\abs{x_{n+1} \dots x_m}{q} \suptermeq
  w$ such that $t_j' = t_{i+n-\mac} = w^\sharp[x_1:=u_1\gamma_i,
  \dots,x_n:=u_n\gamma_i]$; what is more, $x_n \in \FV(w)$.
  Then also $\gamma_i(Z) \approx_n \abs{x_1 \dots x_n}{
  \abs{x_{n+1} \dots x_m}{q}}$ and $x_n \in \FV(\abs{x_{n+1} \dots
  x_m}{q}) \cap \{ x_1,\dots,x_n \}$, so $t_j'$ satisfies the
  requirements.
\end{itemize}
As for the minimality flags: if $m = \minimal$, then all strict
subterms of $t_j'$ are terminating since they are also strict subterms
of some $t_n$.  If $m = \static_\AlterRules$, then by definition
\texttt{beta} does not occur in $[(\rho_i,s_i,t_i) \mid i \in \N]$,
so also the altered pairs are never used.
\end{proof}

\oldcounter{\procAddCond}{\procAddCondSec}
\addCondTheProc
\startappendixcounters


\begin{proof}
For soundness, suppose $(\P,\Rules,m,f)$ is infinite. Thus, let $[(\rho_i,
s_i,t_i) \mid i \in \N]$ be an infinite $(\P,\Rules)$-chain and
$\gamma_i$ be the corresponding substitution for $i \in \N$.  If
$\rho_i = \ell \arrdp \meta{Z}{p_1,\dots,p_\mac}\ (A) \in \P_2$, then
we cannot have $\gamma_i(Z) \approxp \abs{x_1 \dots x_\mac}{s}$ with
$\FV(s) \cap \{x_1,\dots,x_\mac\}\linebreak = \emptyset$, as then $t_i = s$ is
non-terminating: since $Z \in \FMV(\ell)$ by assumption, this
contradicts minimality (which we have both if $m = \minimal$ and by
\refLemma{lem:computableminimal} if $m = \static_\AlterRules$).

For completeness, suppose $(\P',\Rules,m,f)$ is infinite; if $\Rules$
is non-terminating then also $(\P,\Rules,m,f)$ is infinite, so suppose
there is an infinite $(\P',\Rules)$-dependency chain.  An infinite
$(\P,\Rules)$-dependency chain is obtained just by replacing all DPs
by their original: if $\gamma$ respects $A \cup \{Z:i\}$ then it also
respects $A$.
\end{proof}

\subsection{Rule removal without search for orderings}

\oldcounter{\procFormative}{\procFormativeSec}
\formativeTheProc
\startappendixcounters

\begin{proof}
Completeness follows by \refLemma{lem:subsetcomplete}.  Soundness
follows by definition of a formative rules approximation (a formative
infinite
$(\P,\Rules)$-dependency chain can be built using only rules in
$\FR(\P,\Rules)$).
\end{proof}

An example of a formative rules approximation is the following,
adapted from \cite[Def.~6.10]{kop:raa:12}:

\begin{definition}\label{def:FR}
A meta-term $s : \atype$ \emph{has shape} $(a,\btype)$ for $a \in \F
\cup \{ \lambda,\bot \}$ if $\atype = \btype$ and (1) if $s = \abs{x}{
s'}$ then $a = \lambda$, (2) if $s = \apps{\afun}{s_1}{s_n}$ with
$\afun \in \F$ then $a = \afun$, (3) if $s = \apps{(\abs{x}{u})}{v_0
}{v_n}$ then $\apps{u[x:=v_0]}{v_1}{v_n}$ has shape $(a,\btype)$.
Let $\Rules^{(a,\btype)} = \{ \ell \arrz r \in \Rules \mid r$ has
shape $(a,\btype) \}$.
For a pattern $\ell : \btype$, let $\FRA(\ell,\Rules)$ be any
set such that:
\begin{itemize}
\item if $\ell = \apps{\afun}{\ell_1}{\ell_n}$, then
  $\Rules^{(\afun,\btype)} \subseteq \FRA(\ell,\Rules)$ and
  $\FRA(\ell_i,\Rules) \subseteq \FRA(\ell,\Rules)$ for $1 \leq i \leq
  n$;
\item if $\ell = \apps{x}{\ell_1}{\ell_n}$ with $x \in \V$, then
  $\Rules^{(\bot,\btype)} \subseteq \FRA(\ell,\Rules)$ and
  $\FRA(\ell_i,\Rules) \subseteq \FRA(\ell,\Rules)$ for $1 \leq i \leq
  n$;
\item if $\ell = \abs{x}{\ell'}$, then
  $\Rules^{(\lambda,\btype)} \subseteq \FRA(\ell,\Rules)$ and
  $\FRA(\ell',\Rules) \subseteq\linebreak \FRA(\ell,\Rules)$;
\item if $\ell' \arrz r' \in \FRA(\ell,\Rules)$, then
  $\FRA(\ell',\Rules) \subseteq \FRA(\ell,\Rules)$.
\end{itemize}
Such a set always exists, as $\Rules$ itself qualifies, but we can
also find the smallest such set through an inductive process (even for
infinite $\Rules$ by choosing the smallest fixed point of a monotone
function).

Let $\FR(\ell,\Rules) := \{ \ell' \arrz r' \in \Rules \mid
\FRA(\ell,\RulesEta) \cap \{ \ell' \arrz r'\}^{\mathtt{ext}} \neq
\emptyset\}$.
\end{definition}

Since meta-terms $\apps{\meta{Z}{s_1,\dots,s_\mac}}{t_1}{t_n}$ of type
$\atype$ have shape $(a,\atype)$ for all $a$, collapsing rules will
often be included (depending on types).  However, dependency pairs of
the form $\apps{\afun^\sharp}{X_1}{X_n} \arrdp p\ (A)$ do not generate
any formative rules at all.  We observe:

\begin{lemma}\label{lem:FRA}
If $s \arrr{\Rules} \ell\gamma$ by an $\ell$-formative reduction, then
\linebreak
$s \arrr{\FR(\ell,\RulesEta)} \ell\gamma$.
\end{lemma}

\begin{proof}
We will show a bit more: 
that every step in a given
$\ell$-formative reduction $s \arrr{\Rules} \ell\gamma$ is in
$\FRA(\ell,\RulesEta)$, by induction first on the length of that
reduction, second on the size of $s$.

The statement clearly holds if $s = \ell\gamma$, and if $s \arr{\beta}
s' \arrr{\Rules} \ell\gamma$ (with the latter part $\ell$-formative)
we complete by the first induction hypothesis.
If $s = \abs{x}{s'}$ and $\ell = \abs{x}{\ell'}$ and $s \arrr{\Rules}
\ell'\gamma$ by an $\ell'$-formative reduction, we complete by the
second induction hypothesis because $\FRA(\ell',\RulesEta) \subseteq
\FRA(\ell,\RulesEta)$; we are similarly done if $s = \apps{a}{
s_1}{s_n}$,\ $\ell = \apps{a}{\ell_1}{\ell_n}$ and each
$s_i \arrr{\Rules} \ell_i\gamma$ for $a \in \F^\sharp \cup \V$.

Finally, suppose $s \arrr{\Rules} \ell'\delta$ by an $\ell'$-formative
reduction and $r'\delta \arrr{\Rules} \ell\gamma$ by an $\ell$-formative
reduction for $\ell' \arrz r' \in \Rules$.  We must show that then
$\ell' \arrz r' \in \FRA(\ell,\RulesEta)$.  Having this, the induction
statement follows easily, because (1) the reduction $r'\delta
\arrr{\Rules} \ell\gamma$ uses only rules in $ \FRA(\ell,\RulesEta)$ by
the first induction hypothesis, and (2) the reduction $s \arrr{\Rules}
\ell'\delta$ uses only rules in $\FRA(\ell',\RulesEta)$ by the first
induction hypothesis, and by assumption $\FRA(\ell',\RulesEta)
\subseteq \FRA(\ell,\gamma)$.

If $\ell' \arrz r'$ is collapsing, then $r'$ has shape $(a,\btype)$
for all $a$, where $\btype$ is the type of both $s$ and $\ell'$.
Since $\ell$ is not a meta-variable application (as $s = \ell\gamma$
in that case), $\ell' \arr r' \in (\RulesEta)^{(a,\btype)} \subseteq
\FRA(\ell,\RulesEta)$.

If $r'$ has the form $\abs{x}{r''}$, then necessarily $\ell$ has the
form $\abs{x}{\ell''}$, so $\ell' \arrz r' \in (\RulesEta)^{(\lambda,
\btype)} \subseteq \FRA(\ell,\RulesEta)$.

The only alternative is that $r'$ has the form $\apps{\afun}{r_1}{
r_n}$.  If $\ell' = \apps{\afun}{\ell_1}{\ell_n}$, then we are done
because $\ell' \arrz r\ \in (\RulesEta)^{(\lambda,\btype)}$.
Otherwise there is another root step in the $\ell$-formative reduction
$r'\delta \arrr{\Rules} \ell\gamma$: $r'\delta \arrr{\Rules}
\ell''\eta$ and $r''\eta \arrr{\Rules} \ell\gamma$.  By the first
induction hypothesis, $\ell'' \arrz r'' \in \FRA(\ell,\RulesEta)$.
But $s \arrr{\Rules} \ell'\delta \arr{\Rules} r'\delta \arrr{\Rules}
\ell''\eta$ is actually a $\ell''$-formative reduction.  Thus, by the
first induction hypothesis, it uses only rules in $\FR(\ell'',
\RulesEta) \subseteq \FR(\ell,\RulesEta)$.
\end{proof}


\begin{corollary}
The function $\FR$ from \refDef{def:FR} is a formative rules
approximation.
\end{corollary}

We presented \refLemma{lem:FRA} separately because it might be useful
in particular for reduction triples to use a variation of a formative
rules approximation that does not need to be a subset of $\Rules$.

\medskip
We now turn our attention to usable rules:


\oldcounter{\procUsable}{\procUsableSec}
\usableTheProc
\startappendixcounters

\begin{proof}
Suppose there is an infinite minimal $(\P,\Rules,m,f)$-chain
$[(\rho_i,
s_i,t_i) \mid i \in \N]$ (minimality can be assumed by
\refLemma{lem:computableminimal}) and $\P$ is non-collapsing.  Let
$\UR(\P,\Rules)$ be a usable rules approximation, and $\varphi$ the
corresponding function.  Define $\varphi'(\apps{\afun}{s_1}{s_n}) =
\apps{\afun}{\varphi(s_1)}{\varphi(s_n)}$.

As $\P$ is non-collapsing, we can identify $\gamma_i$ for $i \in \N$
such that always $\rho_i = \ell_i \arrdp p_i\ (A_i)$ with $s_i =
\ell_i \gamma$ and $t_i = r_i\gamma_i$; moreover, both $s_i$ and $t_i$
have a ``functional'' shape $\apps{\afun}{u_1}{u_n}$ with each $u_j$
terminating by minimality, so $s_i' := \varphi'(s_i)$ and $t_i' :=
\varphi'(t_i)$ are well-defined.  By definition of usable rules
approximation, $s_i' = \ell_i\gamma^\varphi$ and $t_i' =
p_i\gamma^\varphi$,
and $\varphi(u_j) \arrr{\UR(\P,\Rules)} \varphi(v_j)$ whenever
$u_j \arrr{\Rules} v_j$.
%
Thus, $[(\rho_i,s_i',t_i') \mid i \in \N]$ is a
$(\P,\UR(\P,\Rules))$-chain.
\end{proof}

\subsection{Subterm criterion processors}

Next, we move on to the subterm processors.  We first present the
basic one -- which differs little from its first-order counterpart,
but is provided for context.

\oldcounter{\procSubtermCriterion}{\procSubtermCriterionSec}
\subtermCriterionTheProc
\startappendixcounters

\begin{proof}
Completeness follows by \refLemma{lem:subsetcomplete}.  Soundness
follows because an infinite $(\P,\Rules,m,f)$-chain with the
properties above induces an infinite sequence $\project(s_1)
\suptermeq \project(t_1) \arrr{\Rules} \project(s_2) \suptermeq
\project(t_2) \arrr{\Rules} \dots$.
Since the chain is minimal (either because $m = \minimal$, or by
\refLemma{lem:computableminimal} if $m = \static_S$),
$\project(p_1)$ is terminating, and therefore it is terminating under
$\arr{\Rules} \mathop{\cup} \supterm$.
Thus, there is some index $n$ such
that for all $i \geq n$: $\project(s_i) = \project(t_i) =
\project(s_{i+1})$.  But this can only be the case if
$\project(\ell_i) = \project(p_i)$.
But then the tail of the chain starting at
position $n$ does not use any pair in $\P_1$, and is therefore an
infinite $(\P_2,\Rules,m,f)$-chain.
\end{proof}

\ext{We now turn} to the proof of the static subterm
criterion processor.  This proof is very similar to the one for the
normal \ext{subterm criterion}, but it fundamentally uses the definition of
a computable chain.

\oldcounter{\procStaticSubtermCriterion}{\procStaticSubtermCriterionSec}
\staticSubtermCriterionTheProc
\startappendixcounters

\begin{proof}
Completeness follows by \refLemma{lem:subsetcomplete}.  Soundness
follows because, for $C := C_\AlterRules$ the computability predicate
corresponding to $\arr{S}$, an infinite $(\P,\Rules,\static_S,f)$-chain
induces an infinite $\accreduce{C} \cup \arr{\Rules}$ sequence starting
in the $C$-computable term $\project(s_1)$, with always $s_i
(\accreduce{C} \cup \arr{\Rules})^* t_i$ if $\rho_i \in \P_1$ and
$\project(s_i) = \project(t_i)$ if $\rho_i \in \P_2$; like in the proof
of the subterm criterion, this proves that the chain has a tail that is
a $(\P_2,\Rules,\static_S,f)$-chain because, by definition of $C$,
$\arr{\Rules} \cup \accreduce{C}$ is terminating on $C$-computable
terms.

It remains to be seen that we indeed have $\project(s_i)\ (\accreduce{C}
\cup\ \ext{\arr{\Rules})^+}\linebreak \project(t_i)$ whenever $\rho_i \in \P_1$.  So
suppose that $\rho$ is a dependency pair $\ell \arrdp p\ (A) \in \P_1$
such that $\project(\ell) \sqsupset \project(p)$; we must see that
$\project(\ell\gamma)\ (\accreduce{C} \cup \arr{\beta})^+
\project(p\gamma)$ for \ext{any} substitution $\gamma$ on domain
$\FMV(\ell) \cup \FMV(r)$ such that $v\gamma$ is $C$-computable for
all $v,B$ such that $r \bsuptermeq{B} v$ and $\gamma$ respects $B$.

Write $\ell = \apps{\afun}{\ell_1}{\ell_\maa}$ and $p = \apps{\bfun}{
p_1}{p_{\maa'}}$; then $\project(\ell\gamma) = \ell_{\nu(\afun)}
\gamma$ and $\project(p\gamma) = p_{\nu(\bfun)}\gamma$.
Since, by definition of a dependency pair, $\ell$ is closed, we also
have $\FV(\ell_{\nu(\afun)}) = \emptyset$.
Consider the two possible reasons why $\ell_{\nu(\afun)} \sqsupset
p_{\nu(\bfun)}$.

\begin{itemize}
\item $\ell_{\nu(\afun)} \gracc p_{\nu(\bfun)}$:
  since both sides have base type by assumption and $\ell_{\nu(\afun)}$
  is closed, 
  by \ext{\refLemma{lem:preservecompacchelper} also
  $\ell_{\nu(\afun)}\gamma\ (\accreduce{C} \cup \arr{\beta})^*
  p_{\nu(\bfun)}\gamma$.}
\item $\ell_{\nu(\afun)} \gracc \meta{Z}{x_1,\dots,x_\mia}$ and
  $p_{\nu(\bfun)} = \apps{\meta{Z}{u_1,\dots,u_\mac}}{v_1}{v_n}$:
  denote $\gamma(Z) = \abs{x_1 \dots x_\mia}{q}$ and also $\gamma(Z)
  \approxp \abs{x_1 \dots x_\mac}{q'}$.  Then we can write $q =
  \abs{x_{\mia+1} \dots x_i}{q''}$ as well as $q' = \apps{q''}{x_{i+1}}{
  x_\mac}$ for some $\mia \leq i \leq \mac$.
  Moreover:
  \[
  p_{\nu(\bfun)}\gamma = \apps{q'[x_1:=u_1\gamma,\dots,
  x_\mac:=u_\mac\gamma]}{v_1\gamma}{v_n\gamma}
  \]

  By definition of an $S$-computable chain, $v_j\gamma$ is computable
  for each $1 \leq j \leq n$, and $u_j\gamma$ is computable for each
  $1 \leq j \leq \mac$ such that $x_j \in \FV(q')$.  Write $v_j' :=
  v_j\gamma$ and let $u_j' := u_j\gamma$ if $x_j \in \FV(q')$,
  otherwise $u_j' :=$ a fresh variable; then all $u_j'$ and $v_j'$
  are computable, and still:
  \[
  \begin{array}{cl}
    & p_{\nu(\bfun)}\gamma\\
  = & \apps{q'[x_1:= u_1',\dots,x_\mac:=u_\mac']}{v_1'}{v_n'}\\
  = & \apps{\apps{q''[x_1:=u_1',
      \dots,x_i:=u_i']}{u_{i+1}'}{u_{\mac'}'}}{v_1'}{v_n'}
  \end{array}
  \]
  \ext{On the other hand, by \refLemma{lem:preservecompacchelper} and the
  observation that $\FV(\ell_{\nu(\afun)}) = \emptyset$, we have
  $\ell_{\nu(\afun)
  }\gamma\ (\accreduce{C} \cup \arr{\beta})^+\
  \apps{\apps{q[x_1:=u_1',
  \dots,
  \linebreak
  x_\mia:=u_\mia']}{u_{\mia+1}'}{u_\mac}'}{v_1'}{v_n'}$,
  and as
  $q = \abs{x_{\mia+1} \dots x_i}{q''}$ this term $\beta$-reduces to
  $\apps{\apps{q''[x_1:=u_1',\dots,x_i:=u_i']}{u_{i+1}'}{\linebreak
  u_{\mac'}'}}{v_1'}{v_n'} = p_{\nu(\bfun)}\gamma$.}
\qedhere
\end{itemize}
\end{proof}

\subsection{Non-termination}

Soundness and completeness of the non-termination processor
in \refThm{def:nontermproc} are both direct consequences of
\refDef{def:dpproblem} and \refDef{def:proc}.
%
%

%
%

\end{document}